\documentclass{article}
\usepackage{iclr2026_conference,times}

\iclrfinalcopy

\usepackage{graphicx}
\usepackage{float}
\usepackage{placeins}
\usepackage{algorithm}
\usepackage{algpseudocode}
\usepackage{bbm}
\usepackage{bm}
\usepackage{multirow}
\usepackage{booktabs}
\usepackage{xcolor}
\usepackage{tikz}
\usetikzlibrary{arrows.meta,decorations.pathreplacing}
\usepackage{amsmath}
\usepackage{amsthm}
\usepackage{amsfonts}
\usepackage{amssymb}
\usepackage{subcaption}
\usepackage{enumitem}
\usepackage{hyperref}
\hypersetup{hidelinks}
\usepackage{url}


\newtheorem{theorem}{Theorem}
\newtheorem{assumption}{Assumption}

\newtheorem{proposition}{Proposition}

\newtheorem{remark}{Remark}

\newcommand{\rev}[1]{#1}

\title{\protect\parbox{\textwidth}{\centering\Large Geometric Pareto Control:\\
Physics-Supervised Pareto Representation Learning\\
via Riemannian Energy-Gradient Flow}}

\author{Tong Wu$^1$ \qquad Anna Scaglione$^2$\\
$^1$University of Central Florida \qquad $^2$Cornell University\\
\texttt{tong.wu@ucf.edu, as337@cornell.edu}}

\begin{document}
\raggedbottom

\maketitle
\fancyhead{}

\begin{abstract}
We study multi-objective sequential control problems in physical
systems whose dynamics and operational constraints are known or can be
represented by accurate physics-based models. Although reinforcement
learning has been widely explored in such settings, standard policy
learning still faces high-dimensional action search, fixed or externally
supplied objective trade-offs, retraining under changed priorities, and
hard feasibility requirements at deployment. Our premise is that known
physics and constraints expose a compressible action structure: they
allow GPC to construct a homotopy latent space over which
Pareto-relevant actions can be navigated rather than searched directly.
The resulting feasible Pareto control geometry is organized before
deployment, so online control becomes navigation over this structure
rather than retraining a policy or repeatedly solving a nonlinear
program. We propose Geometric Pareto
Control (GPC), which embeds the supported family of dynamically
feasible Pareto-optimal control responses into a continuous latent
homotopy space using offline scalarized optimization under the known
dynamics. At deployment, the measured physical state induces a semantic
priority coordinate on this space, and GPC follows a geometry-aware
flow on the learned map to decode a feasible action. We state formal
assumptions under which decoded actions remain feasible and the
navigation error does not accumulate over the rollout horizon. Across
analytical control, safe multi-agent navigation, and optimal power flow,
GPC improves over strong optimization and safe-RL baselines
in safety, feasibility, and real-time multi-objective performance. These
results show that known dynamics and constraints can construct an
offline Pareto control geometry that can be navigated online under
changing priorities.
\end{abstract}

\section{Introduction}
\subsection{Background and Motivation}

Reinforcement learning is often proposed for sequential decision-making
in cyber-physical systems such as power grids, transportation networks,
and wireless systems \cite{sutton1998reinforcement, arulkumaran2017deep,
bertsekas2019reinforcement, su2025review, haydari2020deep,
zong2025deep, luong2019applications, feriani2021single}. These domains
are attractive for learning-based control because operating conditions
change over time and the controller must balance several objectives,
such as cost, tracking quality, congestion relief, and safety recovery.
At the same time, many of these systems are not unknown black-box
environments. Operators often have physical equations, network models,
engineering limits, and historical measurements. Even when some
parameters are not known exactly at deployment, their distributions can
often be estimated from offline data. This creates a structured regime
between classical model-free RL and online nonlinear optimization.

The structure is useful, but it does not by itself solve the
multi-objective control problem. A standard RL formulation still turns
the task into policy search: the agent learns an observation-to-action
map from sampled experience, while reward weights or preference inputs
specify how competing objectives are combined. In physical operation,
however, the set of admissible actions is already strongly shaped by
the dynamics, temporal coupling, and hard constraints. Treating these
objects only as an environment for rollouts leaves several bottlenecks:

\begin{itemize}
    \item \textbf{High-dimensional policy search.} RL must still learn
    a mapping from observations to high-dimensional continuous actions,
    even though the known dynamics and constraints already restrict the
    dynamically feasible control space.

    \item \textbf{Fixed or externally supplied objective trade-offs.}
    Standard RL typically optimizes a scalar reward, while
    preference-conditioned or multi-objective RL must cover many reward
    weights. In physical operation, objective priorities may shift
    online as the measured state and constraint margins change.

    \item \textbf{Hard deployment-time constraints.} A policy with high
    expected reward can still propose an infeasible action. Safety
    filters or projections can help, but they add a separate correction
    mechanism and may disrupt the learned policy near the physical
    feasibility boundary.

    \item \textbf{Retraining and distribution shift.} When operating
    conditions, uncertain physical parameters, or objective priorities
    change, a trained policy may require fine-tuning or a new
    preference-conditioning dataset.
\end{itemize}

This paper studies the complementary question suggested by this regime:
if the physics and constraints are available, can they be used to
organize the multi-objective solution structure before deployment, so
that online control becomes navigation over that structure rather than
fresh policy learning or repeated nonlinear optimization? The target is
a controller that reacts to the measured state, follows the appropriate
safety--performance trade-off, and returns a feasible action within a
strict decision deadline. The controller should remain useful when the
operating condition, uncertain physical parameters, or objective
priorities change within the modeled regime. \rev{\textit{GPC is most
useful when the finite-horizon optimal-control responses are
compressible; otherwise, the offline dataset may have to resolve the
same scenario-branching complexity that makes sequential policy
learning difficult. This work therefore focuses on physical systems with
known dynamics and constraints, where that search can be structured
offline.}}

\subsection{Related Work}

Several existing lines of work address parts of this problem. Classical
reinforcement learning frames sequential decision-making as learning
from interaction with a dynamic environment \cite{kaelbling1996reinforcement,
sutton1998reinforcement}. Safe and constrained RL methods use
constrained MDPs, trust regions, Lyapunov arguments, control barrier
functions, safe exploration rules, or runtime shields to reduce
constraint violations \cite{altman2021constrained, wachi2020safe,
altman1998constrained, achiam2017constrained, chow2018lyapunov,
wang2024control, garcia2012safe, wabersich2023data, yang2025cbf,
alshiekh2018safe, carr2023safe}. These methods are natural baselines
for safety-critical control, but they still train a policy for a
particular reward or constraint formulation and often need additional
machinery to guarantee hard feasibility at deployment.

Offline RL and imitation learning reduce the need for online
environment interaction by learning from stored trajectories or expert
data \cite{kostrikov2021offline, kumar2020conservative,
fujimoto2021minimalist, chen2021decision}; batch RL analyses also show
how representation choices affect bias and overfitting when learning
from fixed data \cite{francois2019overfitting}. Multi-objective
sequential decision-making and multi-objective RL study policies,
preference vectors, Pareto sets, and fairness-aware dominance criteria
\cite{roijers2013survey, van2014multi, abels2019dynamic,
reymond2022pareto, michailidis2026scalable}. These methods are closer
to our setting, but they still treat runtime control primarily as
policy evaluation or policy selection. They do not directly turn
offline Pareto solutions into a geometric object that can be traversed
by the controller as the physical state and operating priorities
change.

Model-based and physics-informed methods exploit known dynamics to
improve data efficiency or enforce physical consistency
\cite{raissi2019physics, karniadakis2021physics, de2018end,
hu2019difftaichi, chen2018neural, greydanus2019hamiltonian,
brunton2016discovering, korda2018linear}. Classical multi-objective
optimization methods can compute Pareto fronts directly, and MPC can
solve a constrained optimization problem online. The difficulty is that
front construction is usually offline and discrete, while online MPC or
nonlinear OPF solvers may be too slow for fast dispatch. GPC is designed
for the space between these approaches. It uses offline Pareto solutions
to build a continuous control map shaped by the known dynamics and
constraints, and it uses this map online in place of a repeated
nonlinear solve.

\subsection{Contributions}
We propose Geometric Pareto Control (GPC) for this setting. Relative to
RL, GPC moves the main search problem from online policy learning to
the offline construction of a dynamically feasible Pareto control
geometry. Relative to online optimization, it replaces repeated
nonlinear solves with deterministic movement on the constructed map.
The main contributions are:

\begin{itemize}
    \item \textbf{Compressibility-aware Pareto-solution embedding.} From sampled operating
    conditions and scalarized constrained solves under the known
    dynamics, GPC embeds the supported family of dynamically feasible
    Pareto-optimal control responses into a continuous latent homotopy
    space. Online control therefore starts from the geometry of offline
    Pareto solutions rather than from reward maximization through trial
    and error. \rev{This makes compressibility of the optimal-response
    family an explicit modeling target.}

    \item \textbf{State-dependent Pareto navigation.} At deployment, the
    measured physical state induces a semantic coordinate on the
    embedded Pareto space. GPC navigates the map through a
    geometry-aware flow, allowing the controller to move continuously
    between safety recovery and performance optimization as operating
    conditions and active constraints change.

    \item \textbf{Real-time navigation without an online nonlinear
    solver.} We derive a non-iterative online controller that moves on
    the embedded Pareto map and decodes an action in milliseconds. The
    online computation avoids both a full nonlinear optimization solve
    and retraining of an RL policy.

    \item \textbf{Formal guarantees for feasibility and
    non-accumulation.} We give the formal control objective,
    state the regularity assumptions needed by GPC, and prove
    front-loaded feasibility, drift, and end-to-end action-error
    guarantees. The conditions are expressed in terms of training-time
    feasibility margins, Lipschitz and local-injectivity constants of
    the learned decoders, and bounded latent velocity.
\end{itemize}

We validate GPC on analytical control, safe multi-agent navigation, and
real-time OPF benchmarks against RL, MPC, online scalarization, and
white-box optimization references.

The remainder of this paper is organized as follows.
Section~\ref{sec:problem} defines the constrained multi-objective
control problem. Section~\ref{sec:formal-guarantees} states the
assumptions, main theorems, and proof sketches before the algorithmic
details. Section~\ref{sec:offline} describes the offline Pareto map,
and Section~\ref{sec:online} gives the online navigation algorithm.
Section~\ref{sec:validation} defines the navigation and OPF
instantiations, and Section~\ref{sec:experiments} reports the
experiments. The appendix provides the RK2 predictor, expanded proofs,
the closed-loop non-accumulation result, and the action-decoder training
pipeline.

\begin{figure}[!htbp]
  \centering
\resizebox{\textwidth}{!}{\begin{tikzpicture}[
  font=\sffamily,
  >=Latex,
  box/.style={
    rounded corners=4pt,
    draw=#1,
    fill=white,
    line width=0.75pt,
    align=center,
    inner sep=5pt,
    minimum width=3.25cm,
    minimum height=1.30cm
  },
  band/.style 2 args={
    rounded corners=7pt,
    draw=#1!45,
    fill=#2,
    line width=0.7pt
  },
  flow/.style={->, line width=0.85pt, draw=#1},
  brace/.style={decorate, decoration={brace, amplitude=5pt}, line width=0.85pt, draw=#1}
]
\definecolor{gpcInk}{HTML}{17212B}
\definecolor{gpcMuted}{HTML}{5B6570}
\definecolor{gpcOrange}{HTML}{B86B18}
\definecolor{gpcOrangeLight}{HTML}{FFF4E4}
\definecolor{gpcTeal}{HTML}{167A70}
\definecolor{gpcTealLight}{HTML}{EAF8F5}
\definecolor{gpcBlue}{HTML}{4267B2}

\path[use as bounding box] (0,0) rectangle (18.6,7.05);

\node[font=\sffamily\bfseries\large, text=gpcInk, anchor=west] at (0.25,6.72)
  {Geometric Pareto Control (GPC)};
\node[font=\sffamily\scriptsize, text=gpcMuted, anchor=west] at (0.25,6.38)
  {Build a feasible Pareto geometry offline; navigate it online for real-time multi-objective control.};

\draw[band={gpcOrange}{gpcOrangeLight}] (0.25,3.65) rectangle (16.9,5.90);
\node[font=\sffamily\bfseries\small, text=gpcOrange, anchor=west] at (0.55,5.62)
  {Offline Pareto-map construction};

\draw[band={gpcTeal}{gpcTealLight}] (0.25,0.55) rectangle (16.9,2.80);
\node[font=\sffamily\bfseries\small, text=gpcTeal, anchor=west] at (0.55,2.52)
  {Online Pareto navigation};

\node[box=gpcOrange] (op) at (2.45,4.55)
  {\textcolor{gpcOrange}{\bfseries\scriptsize Operating}\\[-1pt]
   {\scriptsize sample $p\in\mathcal{P}$}\\
   {\scriptsize loads, limits, obstacles}\\
   {\scriptsize and initial states}};

\node[box=gpcOrange] (nlp) at (6.25,4.55)
  {\textcolor{gpcOrange}{\bfseries\scriptsize Constrained Pareto solves}\\[-1pt]
   {\scriptsize scalarized oracle}\\
   {\scriptsize $\min_{\mathbf u}\mathbf w^T\mathbf J$}\\
   {\scriptsize s.t. dynamics and safety}};

\node[box=gpcOrange] (train) at (10.05,4.55)
  {\textcolor{gpcOrange}{\bfseries\scriptsize Representation learning}\\[-1pt]
   {\scriptsize encode state and priority}\\
   {\scriptsize train $E_x,E_\sigma,D_u,D_s$}\\
   {\scriptsize with feasibility margins}};

\node[box=gpcOrange, fill=white!88!gpcOrangeLight] (map) at (13.85,4.55)
  {\textcolor{gpcOrange}{\bfseries\scriptsize Pareto geometry}\\[-1pt]
   {\scriptsize learned manifold}\\
   {\scriptsize $\mathcal{M}^{*}\subset G$}\\
   {\scriptsize feasible action decoder}};

\draw[flow=gpcOrange] (op) -- (nlp);
\draw[flow=gpcOrange] (nlp) -- (train);
\draw[flow=gpcOrange] (train) -- (map);

\node[box=gpcTeal] (obs) at (2.45,1.55)
  {\textcolor{gpcTeal}{\bfseries\scriptsize Observe}\\[-1pt]
   {\scriptsize measure $x_t,p_t$}\\
   {\scriptsize identify active}\\
   {\scriptsize operating condition}};

\node[box=gpcTeal] (pri) at (6.25,1.55)
  {\textcolor{gpcTeal}{\bfseries\scriptsize Set priority}\\[-1pt]
   {\scriptsize compute $\boldsymbol\sigma_t$}\\
   {\scriptsize shift between safety}\\
   {\scriptsize and performance}};

\node[box=gpcTeal] (nav) at (10.05,1.55)
  {\textcolor{gpcTeal}{\bfseries\scriptsize Geometric move}\\[-1pt]
   {\scriptsize update latent state}\\
   {\scriptsize $z_{t+1}=\mathrm{Retr}_{z_t}(\Delta t F)$}\\
   {\scriptsize along $\mathcal{M}^{*}$}};

\node[box=gpcTeal, fill=white!88!gpcTealLight] (act) at (13.85,1.55)
  {\textcolor{gpcTeal}{\bfseries\scriptsize Decode control}\\[-1pt]
   {\scriptsize output action}\\
   {\scriptsize $u_t=\mathrm{clip}(D_u(z_t))$}\\
   {\scriptsize enforce box limits}};

\draw[flow=gpcTeal] (obs) -- (pri);
\draw[flow=gpcTeal] (pri) -- (nav);
\draw[flow=gpcTeal] (nav) -- (act);

\draw[flow=gpcBlue] (map.south) -- ++(0,-0.43) -| (nav.north);
\node[font=\sffamily\scriptsize, text=gpcBlue, fill=gpcTealLight, inner sep=1pt]
  at (11.65,3.12) {reuse $\mathcal{M}^{*}$};

\draw[brace=gpcBlue] (17.22,5.62) -- (17.22,0.82);
\node[font=\sffamily\scriptsize, text=gpcBlue, align=center, anchor=west]
  at (17.42,3.22) {offline builds\\online navigates};

\node[font=\sffamily\scriptsize, text=gpcMuted, anchor=center] at (8.55,0.20)
  {Reported metrics: feasibility, reference-safe failures, percentage optimality gap, and closed-loop safety diagnostics.};
\end{tikzpicture}}
\caption{Overview of Geometric Pareto Control (GPC). The offline stage
  constructs a feasible Pareto geometry $\mathcal{M}^*$ from constrained
  multi-objective solves, while the online stage reuses this geometry to
  choose priorities, move on $\mathcal{M}^*$, and decode a feasible
  clipped control action.}
  \label{fig:gpc_framework}
\end{figure}

\section{Problem and Method}\label{sec:problem}\label{sec:formal-guarantees}

We consider a constrained multi-objective control problem with state
$\mathbf{s}_t$, observation $\mathbf{x}_t$, action $\mathbf{u}_t$, and
operating parameter $\mathbf{p}_t$.  For an objective vector
$\mathbf{J}=(J_1,\ldots,J_M)$ and priority
$\boldsymbol{\sigma}_t\in\Delta^{M-1}$, the deployment goal is to
select actions that are feasible for the known dynamics and close to
the Pareto-optimal response for the current state:
\begin{align}
  \min_{\mathbf{u}_{t:t+H-1}} \quad
  & \sum_{h=0}^{H-1}
  \boldsymbol{\sigma}_{t,h}^{\top}
  \mathbf{J}(\mathbf{s}_{t+h},\mathbf{u}_{t+h},\mathbf{p}_{t+h})
  \label{eq:iclr-control}\\
  \mathrm{s.t.}\quad
  & \mathbf{s}_{t+h+1}=F(\mathbf{s}_{t+h},\mathbf{u}_{t+h},\mathbf{p}_{t+h}),
  \quad
  \mathbf{g}=0,\quad
  \mathbf{f}\leq 0,\quad
  \mathbf{u}_{t+h}\in\mathcal{U},\quad
  h=0,\ldots,H-1.
  \notag
\end{align}
Instead of solving~\eqref{eq:iclr-control} online, GPC constructs a
latent Pareto map offline and performs deterministic navigation on it
at runtime.

\paragraph{Offline Pareto map.}
For sampled parameters $\mathbf{p}^{(i)}$ and weights
$\mathbf{w}^{(r)}\in\Delta^{M-1}$, we solve constrained scalarized
reference optimization problems
\begin{equation}
  \mathbf{u}^{*,i,r}
  \in
  \arg\min_{\mathbf{u}}
  (\mathbf{w}^{(r)})^\top \mathbf{J}(\mathbf{s}^{(i)},\mathbf{u},\mathbf{p}^{(i)})
  \quad
  \mathrm{s.t.}\quad
  \mathbf{g}=0,\ \mathbf{f}\le 0,\ \mathbf{u}\in\mathcal{U}.
  \label{eq:iclr-offline}
\end{equation}
Weighted scalarization recovers the supported Pareto subset.  This is
the subset used by GPC: the method does not claim to recover
non-supported Pareto points that no positive scalarization weight can
select.  In nonconvex domains, the same pipeline can be paired with
$\varepsilon$-constraint or Chebyshev sweeps if coverage of unsupported
regions is required.  The experiments use the scalarized supported front
because it is the portion naturally selected by online priority weights
and because the white-box references are evaluated with the same
objective family.
The resulting dataset contains state, semantic priority, optimal control
response, objective value, and feasibility margin information.  The
stored response $\mathbf{u}_k^*$ may be a single action or a
finite-horizon plan, depending on the domain.  In the
implementation, operating parameters are sampled by Latin hypercube
sampling so that the offline map covers the modeled envelope without
leaving large unsampled regions.  For each scenario and scalarization
weight we store
\begin{equation}
  \mathcal{D}
  =
  \{(\mathbf{x}_k,\mathbf{u}_k^*,\mathbf{s}_k,
  \boldsymbol{\sigma}_k,\mathbf{J}_k,\boldsymbol{\gamma}_k)\}_{k=1}^{N},
  \label{eq:iclr-dataset}
\end{equation}
where $\boldsymbol{\gamma}_k$ stores physical feasibility margins such
as obstacle clearance, branch-flow slack, voltage slack, and ramp slack.
The semantic label is not a manually chosen preference vector.  For
each objective or safety mode, a normalized violation indicator
$\delta_i(\mathbf{s}_k)\in[0,1]$ is converted to an urgency potential
\begin{equation}
  \phi_i(\mathbf{s}_k)
  =
  \exp\!\left(k_i\delta_i(\mathbf{s}_k)/\epsilon_i\right)-1,
  \qquad
  \sigma_{k,i}
  =
  \frac{\phi_i(\mathbf{s}_k)+\rho}
       {\sum_j(\phi_j(\mathbf{s}_k)+\rho)}.
  \label{eq:iclr-offline-semantic}
\end{equation}
Thus nominal states default to a balanced economic priority, while
states near a safety boundary continuously move toward the corresponding
recovery objective without a discrete mode switch.  GPC learns an
observation encoder $E_x$, an offline oracle encoder $E_o$,
and decoders $D_u,D_s$.  \rev{This gives a dual path: the $s$-path
$E_o$ organizes offline samples, while the $x$-path $E_x$ localizes
online observations on the same map.}  The deployment code
$z^x=E_x(\mathbf{x})$ is trained to match the oracle code
$z^o=E_o(\mathbf{s},\boldsymbol{\sigma},\mathbf{J})$ on a Pareto
manifold $\mathcal{M}^*\subset G$ in an ambient Lie group, while
$D_u(z)$ reconstructs feasible Pareto control responses: a single action
in one-action domains and a finite-horizon plan in receding-horizon
domains.  The training objective
combines action reconstruction, state reconstruction, cross-modal
consistency, neighborhood preservation, and margin-aware rollout
coverage; implementation details are expanded in Appendix~\ref{app:du_training}.
In the main experiments we use the compact objective
\begin{align}
  \mathcal{L}_{\mathrm{GPC}}
  =
  &\lambda_u\|D_u(z)-\mathbf{u}^*\|_2^2
  +\lambda_s\|D_s(z)-\mathbf{s}\|_2^2
  +\lambda_c\|E_x(\mathbf{x})-E_o(\mathbf{s},\boldsymbol{\sigma},\mathbf{J})\|_2^2
  \notag\\
  &+\lambda_{\mathrm{loc}}
  \sum_{j\in\mathcal{N}(i)}
  \left|
  \|z_i-z_j\|_2-\alpha\|\mathbf{u}^{*,i}-\mathbf{u}^{*,j}\|_2
  \right|
  +\lambda_{\mathrm{roll}}\mathcal{L}_{\mathrm{roll}} .
  \label{eq:iclr-training-loss}
\end{align}
The first two terms make the latent representation decode to the
reference action and physical state.  The consistency and locality terms
prevent the latent coordinate from becoming an arbitrary regression
code: nearby Pareto optima in action/state space must remain nearby on
the learned map; in the locality term, $i$ denotes the anchor sample
and $\mathcal{N}(i)$ its neighborhood on the offline Pareto map.  The
supervised loss definitions are given in
Appendix~\ref{sec:loss}.  The rollout term trains the decoder on latent
states generated by the same horizon-dependent update used at deployment
rather than only on isolated offline samples, which reduces the mismatch
between static offline supervision and online navigation.
The rollout-aware decoder refinement freezes the encoders and updates
$D_u$ on latent rollout states over the task horizon:
\begin{align}
  \mathcal{L}_{\mathrm{roll}}
  =
  \frac{1}{|\mathcal{B}|}
  \sum_{\xi\in\mathcal{B}}
  \sum_{h=1}^{H_{\xi}}
  \omega_{\xi,h}^{\mathrm{risk}}
  \left\|
  \Pi_hD_u(\tilde z_{\xi,h})-\mathbf{u}_{\xi,h}^{*}
  \right\|_2^2,
  \notag\\
  \tilde z_{\xi,h}
  =
  \mathrm{Retr}_{\tilde z_{\xi,h-1}}
  (\Delta z_{\xi,h-1}^{\mathrm{Euc}}+\gamma_{\mathrm{L}}\Delta z_{\xi,h-1}^{\mathrm{Lie}}).
  \label{eq:iclr-rollout-loss}
\end{align}
Here $\xi$ indexes a training rollout, $\tilde z_{\xi,0}$ is initialized
from the localized oracle code, $H_{\xi}$ is the horizon used by that
domain, and $\Pi_h$ extracts the $h$th action from the decoded plan.
Single-action domains are recovered by $H_{\xi}=1$ and $\Pi_1=I$, while
multi-agent navigation uses the full receding-horizon joint plan.
This refinement is important because deployment does not query only
training codes $E_x(\mathbf{x}_k)$; it queries intermediate codes created
by RK2 prediction, Lie-local correction, retraction, and output clamping.
We therefore evaluate the same guarded update used by the controller,
assign larger weights $\omega_{\xi,h}^{\mathrm{risk}}$ to low-margin
rollout states, and calibrate $\delta_{\mathrm{dec}}$ in
eq.~\eqref{eq:iclr-margin-decoder} on both oracle codes and rollout
codes.  The learned decoder is therefore trained on the actual latent
distribution induced by online navigation rather than only on isolated
offline Pareto samples.

\paragraph{Online navigation.}
At deployment, GPC first localizes the current observation on the
learned map, then updates a semantic coordinate from constraint
urgencies.  The measurement-consistent subset of the learned Pareto
manifold is
\begin{equation}
  \mathcal{M}_t
  =
  \left\{
  z\in\mathcal{M}^*
  \;\middle|\;
  \|\mathbf{x}_t-D_s(z)\|_2^2
  \le
  \tau_t
  \right\},
  \qquad
  \tau_t=\hat\sigma_t^2+\tau_{\mathrm{geom}},
  \label{eq:iclr-measurement-consistent}
\end{equation}
where $\hat\sigma_t^2$ is an online noise estimate and
$\tau_{\mathrm{geom}}$ is calibrated from held-out reconstruction
residuals.  Localization uses a smoothed projection
\begin{equation}
  z_t
  =
  (1-\alpha)z_{t-1}
  +
  \alpha
  \arg\min_{z\in\mathcal{M}_t}
  \|z-E_x(\mathbf{x}_t)\|_2^2 .
  \label{eq:iclr-localization}
\end{equation}
Given indicators $\delta_i(\hat{\mathbf{s}}_t)$ for the
objectives or active constraints, we use
\begin{equation}
  \phi_i(\hat{\mathbf{s}}_t)
  =
  \exp(k_i\delta_i(\hat{\mathbf{s}}_t)/\epsilon_i)-1,
  \qquad
  \sigma_{t,i}
  =
  \frac{\phi_i(\hat{\mathbf{s}}_t)+\rho}
       {\sum_j(\phi_j(\hat{\mathbf{s}}_t)+\rho)}.
  \label{eq:iclr-semantic}
\end{equation}
The latent state follows a geometry-aware field $\mathbf{F}(z_t)$ with a
manifold-restoration component and a tangential priority component:
\begin{equation}
\begin{aligned}
  V(z_t,\mathbf{x}_t)
  &=
  \frac{1}{\epsilon}\Phi_{\mathrm{normal}}(z_t,\mathbf{x}_t)
  +\Phi_{\mathrm{tangent}}(z_t),
  \quad
  \mathbf{F}_t
  =\mathbf{F}(z_t,\mathbf{x}_t;\boldsymbol{\sigma}_t)
  =-\nabla_z V(z_t,\mathbf{x}_t),\\[-0.2em]
  \Phi_{\mathrm{normal}}(z_t,\mathbf{x}_t)
  &=
  \|\mathbf{x}_t-D_s(z_t)\|_2^2,\quad
  \Phi_{\mathrm{tangent}}(z_t)
  =
  \sum_{i=1}^{M}
  \sigma_{t,i}J_i(D_s(z_t),D_u(z_t)).
\end{aligned}
\label{eq:iclr-energy}
\end{equation}
The $1/\epsilon$ term creates a fast restoring ``snap'' back toward
the measurement-consistent Pareto manifold, while the tangential term
creates a slower ``slide'' along the manifold toward the current
semantic priority.  A decoder-pullback metric
$G(z)=J_{D_s}(z)^\top J_{D_s}(z)+\lambda_{\mathrm{m}} I$ defines the local
geometry used by the optional Lie residual.  With $B_k$ and $\Lambda_k$
denoting the top-$k$ eigenvectors and corresponding eigenvalues of
$G(z_t)$, the implemented step is
\begin{align}
  \Delta z_{\mathrm{Euc}}
  &=
  \Delta t\,\mathbf{F}_t
  \quad\text{or the RK2 midpoint analogue}, \notag\\
  \Delta z_{\mathrm{Lie}}
  &=
  B_k\!\left[
  \mathrm{Rot}_{\mathfrak{so}(3)}
  \!\left(B_k^\top\Delta z_{\mathrm{Euc}},
  (B_k^\top \mathbf{F}_t)\oslash\sqrt{\Lambda_k};\Delta t\right)
  -B_k^\top\Delta z_{\mathrm{Euc}}\right],
  \label{eq:iclr-lie-residual}\\
  z_{t+1}
  &=
  \mathrm{Retr}_{z_t}
  \left(\Delta z_{\mathrm{Euc}}+\gamma_{\mathrm{L}}\Delta z_{\mathrm{Lie}}\right),
  \qquad
  \mathbf{u}_t=\mathrm{clip}(D_u(z_{t+1}),\mathcal{U}).
  \notag
\end{align}
Setting $\gamma_{\mathrm{L}}=0$ recovers the plain RK2 navigator, and setting the
midpoint correction to zero gives the first-order analysis step
$z_{t+1}=\mathrm{Retr}_{z_t}(\Delta t\mathbf{F}_t)$.  The experiments use the
geometry-aware RK2/Lie implementation in~\eqref{eq:iclr-lie-residual};
in all cases, the control output is clamped to physical box constraints
after decoding.

\begin{algorithm}[!htbp]
\caption{Online Geometric Pareto Control}
\label{alg:iclr-online}
\begin{algorithmic}[1]
\Require learned map $(E_x,D_u,D_s)$, observation $\mathbf{x}_t$, parameter $\mathbf{p}_t$
\State localize $z_t \leftarrow E_x(\mathbf{x}_t)$ and reconstruct $\hat{\mathbf{s}}_t\leftarrow D_s(z_t)$
\State compute violation indicators $\delta_i(\hat{\mathbf{s}}_t,\mathbf{p}_t)$
\State compute semantic priority $\boldsymbol{\sigma}_t$ using eq.~\eqref{eq:iclr-semantic}
\State form the manifold-restoring and tangential priority field $\mathbf{F}(z_t,\mathbf{x}_t,\boldsymbol{\sigma}_t)$
\State update $z_{t+1}\leftarrow \mathrm{Retr}_{z_t}(\Delta t F)$
\State decode and clamp $\mathbf{u}_t\leftarrow \mathrm{clip}(D_u(z_{t+1}),\mathcal{U})$
\State log feasibility margins and execute $\mathbf{u}_t$
\end{algorithmic}
\end{algorithm}

\paragraph{Guarantees.}
The main guarantee is a local tube result: if the learned Pareto map has
a positive physical margin and the online navigator is re-localized at
each control cycle, then decoded actions remain feasible and the action
error does not accumulate with horizon length.  Let
$\mathbf{f}_t(\mathbf{u})\le 0$ collect all non-box physical and temporal
constraints after conditioning on $(\mathbf{s}_t,\mathbf{u}_{t-1})$;
box constraints are enforced by the final clamp.  For
$z^*\in\mathcal{M}^*$ define the certified training margin and decoder
calibration error
\begin{align}
  \gamma_{\min}
  &:=
  \inf_{t,z^*\in\mathcal{M}^*}
  \min_i[-f_{t,i}(D_u(z^*))],
  \qquad
  \delta_{\mathrm{dec}}
  :=
  \sup_{z^*\in\mathcal{M}^*}
  \|D_u(z^*)-\mathbf{u}^*(z^*)\|_2 .
  \label{eq:iclr-margin-decoder}
\end{align}
\begin{assumption}[Certified Pareto tube]
\label{ass:iclr-certified-tube}
The offline optimizer returns feasible Pareto samples with
$\gamma_{\min}>0$, the learned map satisfies
$\mathcal{M}^*\subseteq\{z:\mathbf{f}_t(D_u(z))\le0\}$ over the modeled
operating envelope, and $D_u$ and $\mathbf{f}_t$ are respectively
$K_{D_u}$- and $L_f$-Lipschitz on the decoded tube around
$\mathcal{M}^*$.
\end{assumption}

\begin{assumption}[Observable localization and bounded geometric step]
\label{ass:iclr-local-step}
The state decoder is locally identifiable on the Pareto tube:
\begin{equation}
  \sigma_{\min}(J_{D_s}(z))\ge\sigma_0>0,
  \qquad
  \|z-\Pi_{\mathcal{M}^*}(z)\|_2
  \le
  \sigma_0^{-1}\sqrt{2\Phi_{\mathrm{normal}}(z,\mathbf{x}_t)} .
  \label{eq:iclr-identifiability}
\end{equation}
Each control cycle localizes to within $\delta_{\mathrm{loc}}$ of
$\mathcal{M}^*$, and the RK2/Lie update is velocity capped so that
$\|\dot z_t\|\le V_{\max}$.
\end{assumption}

\begin{proposition}[Decoded feasibility near the Pareto manifold]
\label{prop:iclr-decoded-feasibility}
Under Assumption~\ref{ass:iclr-certified-tube}, define the feasible
latent tube radius
\begin{equation}
  \delta_{\mathrm{feas}}
  =
  \frac{\gamma_{\min}}{L_fK_{D_u}} .
  \label{eq:iclr-feas-radius}
\end{equation}
If $\|z-\Pi_{\mathcal{M}^*}(z)\|_2\le\delta_{\mathrm{feas}}$, then
$\mathbf{f}_t(D_u(z))\le0$.  Consequently, if
\begin{equation}
  \delta_{\mathrm{loc}}+V_{\max}\Delta t
  \le
  \delta_{\mathrm{feas}},
  \label{eq:iclr-feasible-step-condition}
\end{equation}
then the decoded non-box constraints remain feasible at each localized
GPC cycle; box constraints are enforced by the final output clamp.
\end{proposition}

This is the feasibility mechanism used in both navigation and OPF: the
offline optimizer supplies a strict margin, and online control is safe as
long as localization keeps the decoded action inside that margin tube.

\begin{proposition}[Non-accumulation of geometric drift]
\label{prop:iclr-drift}
Assume the restoring normal flow contracts at rate $1/\epsilon$, and
the RK2/Lie discretization injects at most $F_\Delta$ normal energy per
step:
\begin{equation}
  \Phi_{\mathrm{normal}}(z+\delta_{\mathrm{Lie}})
  \le
  \Phi_{\mathrm{normal}}(z)+F_\Delta,
  \qquad
  \|\delta_{\mathrm{Lie}}\|_2
  \le
  \tfrac12 V_{\max}^2\Delta t^2 .
  \label{eq:iclr-lie-perturb-bound}
\end{equation}
Under Assumption~\ref{ass:iclr-local-step}, the per-cycle re-anchoring
gives the normal-error recursion
\begin{equation}
  \Phi_{\mathrm{normal}}(z_{t+1})
  \le
  e^{-\Delta t/\epsilon}\Phi_{\mathrm{normal}}(z_t)
  +F_\Delta+\sigma_0^2\delta_{\mathrm{loc}}^2 .
  \label{eq:iclr-normal-recursion}
\end{equation}
and therefore
\begin{equation}
  \limsup_{t\to\infty}\Phi_{\mathrm{normal}}(z_t)
  \le
  \frac{F_\Delta+\sigma_0^2\delta_{\mathrm{loc}}^2}
  {1-e^{-\Delta t/\epsilon}} .
  \label{eq:iclr-drift-bound}
\end{equation}
Thus off-manifold error is bounded by a steady-state radius that does
not grow with the rollout horizon.
\end{proposition}

\begin{theorem}[End-to-end GPC deployment guarantee]
\label{thm:iclr-main}
Under Assumptions~\ref{ass:iclr-certified-tube}--\ref{ass:iclr-local-step}
and Propositions~\ref{prop:iclr-decoded-feasibility}--\ref{prop:iclr-drift},
decoded actions remain feasible whenever
eq.~\eqref{eq:iclr-feasible-step-condition} holds.  If the action
decoder has calibration error $\delta_{\mathrm{dec}}$ as in
eq.~\eqref{eq:iclr-margin-decoder}, then
\begin{equation}
  \limsup_{t\to\infty}\|\mathbf{u}_t-\mathbf{u}_t^*\|_2
  \le
  \frac{K_{D_u}}{\sigma_0}
  \sqrt{
  \frac{2(F_\Delta+\sigma_0^2\delta_{\mathrm{loc}}^2)}
  {1-e^{-\Delta t/\epsilon}}}
  +\delta_{\mathrm{dec}} .
  \label{eq:iclr-action-error-main}
\end{equation}
Equivalently,
\begin{equation}
  \sup_{t\ge0}\|\mathbf{u}_t-\mathbf{u}_t^*\|_2
  =
  O(\sqrt{\epsilon}V_{\max}\Delta t
  +K_{D_u}\delta_{\mathrm{loc}}/\sigma_0)
  +\delta_{\mathrm{dec}} .
  \label{eq:iclr-error}
\end{equation}
\end{theorem}

Thus the guarantee is stated in quantities that are either validated
offline or monitored online: minimum training margin, decoder error,
nearest-neighbor/local Jacobian distortion, and localization residual.
Full statements and proofs are in Appendix~\ref{appendix:theory}.

\section{Instantiations}

\paragraph{Safe multi-agent navigation.}
The navigation benchmark uses up to five unicycle agents with state
$(\mathbf{q}_{i,t},\theta_{i,t})$ and control
$(v_{i,t},\omega_{i,t})$:
\begin{equation}
  \mathbf{q}_{i,t+1}
  =
  \mathbf{q}_{i,t}
  +\Delta t\,v_{i,t}(\cos\theta_{i,t},\sin\theta_{i,t}),
  \qquad
  \theta_{i,t+1}=\theta_{i,t}+\Delta t\,\omega_{i,t}.
  \label{eq:iclr-nav-dynamics}
\end{equation}
For a receding horizon $H$, the reference and training solves share
the same control limits, slew limits, moving-obstacle prediction, and
pairwise safety constraints:
\begin{align}
  v_i^{\min}\le v_{i,t+k}\le v_i^{\max},\qquad
  |\omega_{i,t+k}|\le \omega_i^{\max},\qquad
  \|\mathbf{u}_{i,t+k}-\mathbf{u}_{i,t+k-1}\|_\infty\le \Delta u^{\max}.
  \label{eq:iclr-nav-limits}
\end{align}
Moving obstacles are predicted by
$\mathbf{o}_{m,t+k+1}=\mathbf{o}_{m,t+k}+\Delta t\,\mathbf{v}_{m,t+k}$
with boundary reflection.  The hard safety margins are
\begin{align}
  h_{i,m}^{\mathrm{obs}}(t+k)
  &=
  \|\mathbf{q}_{i,t+k}-\mathbf{o}_{m,t+k}\|_2
  -r_i-r_m-d_{\mathrm{safe}}\ge 0,
  \notag\\
  h_{i,j}^{\mathrm{pair}}(t+k)
  &=
  \|\mathbf{q}_{i,t+k}-\mathbf{q}_{j,t+k}\|_2
  -r_i-r_j-d_{\mathrm{pair}}\ge 0 .
  \label{eq:iclr-nav-margins}
\end{align}
The dynamic setting is deliberately hard for a finite-budget optimizer:
the reference can fail to find a jointly safe rollout even though GPC
has learned a smooth feasible control map from many offline solves.
We therefore report both the reference-safe rate and the percentage
optimality gap conditioned on reference-safe episodes.
The full six-term objective is expanded in
Appendix~\ref{sec:multiagent-navigation}.  In the main text we write it
compactly as
\begin{equation}
  \mathbf{J}_{\mathrm{nav}}
  =
  (J_{\mathrm{goal}},J_{\mathrm{safe}},J_{\mathrm{effort}},
  J_{\mathrm{viol}},J_{\mathrm{term}},J_{\mathrm{smooth}}),
  \label{eq:iclr-multiagent-objectives-full}
\end{equation}
where the first three terms are preference-weighted and the last three
remain always active.  The white-box reference and the GPC offline solver solve
\begin{equation}
\begin{aligned}
  \mathbf{U}^{\star}(\mathbf{w};\xi_t)
  \in
  \arg\min_{\mathbf{U}_{t:t+H-1}}\quad
  &w_1\lambda_gJ_{\mathrm{goal}}
  +w_2\lambda_sJ_{\mathrm{safe}}
  +w_3\lambda_eJ_{\mathrm{effort}}\\
  &+\lambda_vJ_{\mathrm{viol}}
  +\lambda_TJ_{\mathrm{term}}
  +\lambda_{\Delta}J_{\mathrm{smooth}}\\
  \mathrm{s.t.}\quad
  &\eqref{eq:iclr-nav-dynamics}-\eqref{eq:iclr-nav-margins},
  \quad k=0,\ldots,H-1 .
\end{aligned}
\label{eq:iclr-multiagent-full-opt}
\end{equation}
At deployment, GPC decodes and executes only the first joint action and
then re-localizes at the next observation.
Let $\mathcal{I}_{\mathrm{ref}}$ denote the episodes in which the
white-box reference both reaches the joint goal and remains
collision-free.  The reported navigation gap is
\begin{equation}
  \mathrm{Gap}_{\mathrm{nav}}
  =
  100\cdot
  \frac{
  \sum_{e\in\mathcal{I}_{\mathrm{ref}}}
  \left(J_{\mathrm{GPC}}^{(e)}-J_{\mathrm{ref}}^{(e)}\right)}
  {\sum_{e\in\mathcal{I}_{\mathrm{ref}}}J_{\mathrm{ref}}^{(e)}} .
  \label{eq:iclr-nav-gap}
\end{equation}
This avoids the misleading comparison in which an unsafe reference
trajectory is treated as an optimal trajectory.  In the moving-obstacle
case, only $25\%$ of reference rollouts satisfy this safety condition.

\paragraph{Optimal power flow.}
The OPF benchmark uses the IEEE 30-bus AC network.  The action contains
generator dispatch and voltage-control variables, and the constraints
include AC power-flow equations, generator limits, voltage limits,
thermal branch limits, and ramp-rate limits.  The objective vector
contains thermal recovery, voltage recovery, and economic dispatch:
$\mathbf{J}=(J_{\mathrm{thermal}},J_{\mathrm{voltage}},
J_{\mathrm{econ}})$.  The first objective penalizes branch loading near
thermal ratings, the second penalizes voltage deviations outside the
admissible band, and the third is the quadratic economic generation
cost.  The full formulas are given in Appendix~\ref{sec:opf}.  The
semantic coordinate in~\eqref{eq:iclr-semantic} therefore shifts weight
toward $J_{\mathrm{thermal}}$ or $J_{\mathrm{voltage}}$ near constraint
boundaries and returns weight to $J_{\mathrm{econ}}$ during nominal
operation.
For the OPF instantiation, the real-time indicators used in
eqs.~\eqref{eq:iclr-offline-semantic} and~\eqref{eq:iclr-semantic} are
computed from reconstructed voltage phasors
$\hat{\mathbf{s}}_t=D_s(z_t)$ and decoded dispatch
$\hat{\mathbf{u}}_t=D_u(z_t)$:
\begin{equation}
  \boldsymbol{\delta}_{\mathrm{OPF}}(\hat{\mathbf{s}}_t,\hat{\mathbf{u}}_t)
  =
  (\delta_{\mathrm{th}},\delta_{\mathrm{v}},\delta_{\mathrm{e}})
  \in [0,1]^3 .
  \label{eq:iclr-opf-indicators}
\end{equation}
The thermal and voltage indicators measure proximity to hard operating
boundaries, while the economic indicator is normalized by offline cost
quantiles and remains a low-level background drive when the grid is
safe; their expanded definitions are given in
Appendix~\ref{sec:opf}.  Passing
$(\delta_{\mathrm{th}},\delta_{\mathrm{v}},\delta_{\mathrm{e}})$ through
the exponential urgency map gives a smooth priority shift from economic
dispatch to thermal or voltage recovery, which is the OPF counterpart of
the obstacle-clearance semantic shift in navigation.
For uncertainty-aware training, branch admittances are sampled as
$Y_{ij}=Y^0_{ij}(1+\xi^Y_{ij})$ with
$\xi^Y_{ij}\in[-0.03,0.03]$ and held fixed within a trajectory.
Equivalently, for each branch $\ell$,
\begin{equation}
  b_{\ell}=1+\xi_{\ell}^{Y},\qquad
  \xi_{\ell}^{Y}\sim\mathrm{clip}(\mathcal{N}(0,\sigma_Y^2),-\rho_Y,\rho_Y),
  \qquad
  Y_{\ell}^{(\mathrm{traj})}=b_{\ell}Y_{\ell}^{0},
  \label{eq:iclr-branch-uncertainty}
\end{equation}
with $(\sigma_Y,\rho_Y)=(0.01,0.03)$ in the uncertainty experiment.
The vector $\mathbf{b}$ is sampled once per trajectory and reused for
all dispatch steps in that trajectory, matching the physical separation
between slow line-parameter drift and fast OPF dispatch.  If any step
in the ramp-coupled offline optimization chain fails to converge, the trajectory is
discarded rather than partially included in the Pareto map.
For each horizon step, the offline solver solves the scalarized AC-OPF problem
\begin{align}
  \min_{\{P_G,Q_G,V,\theta\}_{0:H}}\quad
  & \sum_{h=0}^{H}
  \mathbf{w}_{h}^{\top}
  (J_{\mathrm{thermal},h},J_{\mathrm{voltage},h},J_{\mathrm{econ},h})
  \label{eq:iclr-opf-opt}\\
  \mathrm{s.t.}\quad
  & P_{G,i,h}-P_{D,i,h}
  =
  \sum_{j=1}^{n_b}
  V_{i,h}V_{j,h}
  (G_{ij}\cos\theta_{ij,h}+B_{ij}\sin\theta_{ij,h}),
  \notag\\
  & Q_{G,i,h}-Q_{D,i,h}
  =
  \sum_{j=1}^{n_b}
  V_{i,h}V_{j,h}
  (G_{ij}\sin\theta_{ij,h}-B_{ij}\cos\theta_{ij,h}),
  \notag\\
  & |S_{\ell,h}(V,\theta)|\le S_{\ell}^{\max},\quad
  V_i^{\min}\le V_{i,h}\le V_i^{\max},\quad
  P_{G,i}^{\min}\le P_{G,i,h}\le P_{G,i}^{\max},
  \notag\\
  & Q_{G,i}^{\min}\le Q_{G,i,h}\le Q_{G,i}^{\max},\quad
  |P_{G,i,h}-P_{G,i,h-1}|\le R_i .
  \notag
\end{align}
The online controller decodes the first dispatch action from the
learned Pareto geometry; it does not run IPOPT or solve a new AC-OPF
inside the real-time loop.

\section{Experiments}

The main text focuses on the two settings where the learned Pareto
geometry is most needed at deployment: safe multi-agent navigation, where
a finite-budget white-box reference often fails to remain safe, and IEEE
30-bus OPF under branch-admittance uncertainty, where the controller must
remain feasible when the network model drifts from nominal.  Auxiliary
Case~1 and Case~3 studies are reported in
Appendices~\ref{sec:exp:analytical} and~\ref{sec:exp:nobattery}.
We report feasibility,
percentage optimality gap, and wall-clock decision latency; the
navigation gap is conditioned on reference-safe episodes because unsafe
reference trajectories are not valid optima.

\begin{table}[!b]
\centering
\caption{Simulation organization.  Case~2 and Case~4 are the main-text
stress tests; Case~1 and Case~3 are retained in the appendix as
mechanism and sanity-check studies.}
\label{tab:iclr-setup}
\small
\setlength{\tabcolsep}{4pt}
\begin{tabular}{p{0.14\textwidth}p{0.22\textwidth}p{0.31\textwidth}p{0.21\textwidth}}
\toprule
\textbf{Case} & \textbf{System} & \textbf{Hard constraints} & \textbf{Role in paper}\\
\midrule
Case~1 & 2D dynamic navigation & Obstacle, input box, slew-rate coupling & Appendix mechanism check\\
Case~2 & Five-agent robot navigation & Static/moving obstacles, pairwise separation, control slew & Main safe-navigation stress test\\
Case~3 & IEEE 30-bus nominal OPF & AC flow, thermal, voltage, generator, ramp limits & Appendix OPF sanity check\\
Case~4 & IEEE 30-bus uncertain OPF & Same limits with trajectory-level branch-admittance drift & Main robustness stress test\\
\bottomrule
\end{tabular}
\end{table}

The navigation benchmark uses five-agent unicycle dynamics with
robot-obstacle, moving-obstacle, and robot-robot safety margins.  The OPF
benchmark uses the IEEE 30-bus AC network with power-flow, thermal,
voltage, generator, and ramp constraints under fixed trajectory-level
branch-admittance drift.  Implementation, training, and baseline details
are given in Appendix~\ref{sec:experiments}.

\paragraph{Metrics.}
For a navigation episode $e$, define the minimum closed-loop margins
\begin{equation}
  m_{\mathrm{obs}}^{(e)}
  =
  \min_{t,i,m}h_{i,m}^{\mathrm{obs}}(t),
  \qquad
  m_{\mathrm{pair}}^{(e)}
  =
  \min_{t,i<j}h_{i,j}^{\mathrm{pair}}(t).
  \label{eq:iclr-nav-metric-margins}
\end{equation}
An episode is counted as safely successful only if all active agents
reach their goals and both margins are nonnegative.  Scene-collision-free
counts only obstacle and pairwise safety, independent of terminal goal
completion.  The reference-safe set is
\begin{equation}
  \mathcal{I}_{\mathrm{ref}}
  =
  \{e:\mathrm{goal}_{\mathrm{ref}}^{(e)}=1,\ 
  m_{\mathrm{obs,ref}}^{(e)}\ge0,\ 
  m_{\mathrm{pair,ref}}^{(e)}\ge0\},
  \label{eq:iclr-refsafe-set}
\end{equation}
and the navigation optimality gap is the percentage aggregate gap in
eq.~\eqref{eq:iclr-nav-gap}.  For OPF, feasibility is the fraction of
dispatch steps satisfying AC balance residual tolerance, thermal limits,
voltage limits, generator limits, and ramp limits.  The OPF percentage
gap is
\begin{equation}
  \mathrm{Gap}_{\mathrm{OPF}}
  =
  100\cdot
  \frac{\sum_t J_{\mathrm{total}}^{\mathrm{method}}(t)
        -\sum_t J_{\mathrm{total}}^{\mathrm{MO\text{-}IPOPT}}(t)}
       {\sum_t J_{\mathrm{total}}^{\mathrm{MO\text{-}IPOPT}}(t)}.
  \label{eq:iclr-opf-gap}
\end{equation}
These definitions are deliberately conservative: an unsafe low-cost
trajectory is reported as a safety failure, not as an optimal solution.

\begin{table}[!htbp]
\centering
\caption{Main quantitative results with navigation baselines.  Navigation
rows report percentages over 100 evaluation episodes: Joint is all-agent
goal success, Coll.-free is scene-collision-free rollout rate, and Safe
requires both.  Navigation gaps are percentage gaps conditioned on
reference-safe episodes.  For the moving-obstacle GPC row, $\dagger$
reports the reachable-conditioned rate after excluding one
reference-unreachable scene.
For OPF, Feas. reports dispatch feasibility over 300 steps and Viol. is
the number of infeasible steps.}
\label{tab:iclr-main-results}
\scriptsize
\setlength{\tabcolsep}{2.5pt}
\resizebox{\textwidth}{!}{%
\begin{tabular}{llcccccc}
\toprule
\textbf{Case} & \textbf{Method} &
\textbf{Joint/feas.} &
\textbf{Coll.-free} &
\textbf{Safe/viol.} &
\textbf{Margin} &
\textbf{Gap (\%)}$\downarrow$ &
\textbf{Latency}\\
\midrule
Static nav. & \textbf{GPC} & \textbf{100\%} & \textbf{100\%} & \textbf{100\%} & \textbf{0.301} & 0.28 & --\\
Static nav. & Finite-budget ref. & -- & -- & 60\% & -- & 0.00 & solver\\
Static nav. & PPO & 0\% & 5\% & 0\% & -0.558 & -- & policy\\
Static nav. & PPO-Lagrangian & 0\% & 5\% & 0\% & -0.540 & -- & policy\\
Static nav. & CPO & 0\% & 13\% & 0\% & -0.467 & -- & policy\\
Static nav. & PPO-shielded & 0\% & 8\% & 0\% & -0.429 & -- & policy\\
\midrule
Moving nav. & \textbf{GPC} & \textbf{99\%}$^\dagger$ & \textbf{100\%} & \textbf{99\%}$^\dagger$ & \textbf{0.203} & -3.43 & --\\
Moving nav. & Finite-budget ref. & 83\% & 42\% & 35\% & -- & 0.00 & solver\\
Moving nav. & PPO & 0\% & 0\% & 0\% & -0.663 & -- & policy\\
Moving nav. & PPO-Lagrangian & 0\% & 0\% & 0\% & -0.672 & -- & policy\\
Moving nav. & CPO & 0\% & 0\% & 0\% & -0.637 & -- & policy\\
Moving nav. & PPO-shielded & 0\% & 0\% & 0\% & -0.567 & -- & policy\\
\midrule
Uncertain OPF & MO-IPOPT & 100.0\% & -- & 0 viol. & -- & 0.00 & $\gg10^3$ ms\\
Uncertain OPF & \textbf{GPC} & \textbf{100.0\%} & -- & \textbf{0 viol.} & -- & \textbf{0.25} & \textbf{12.6 ms}\\
Uncertain OPF & MPC & 99.67\% & -- & 1 viol. & -- & 0.23 & 1228.7 ms\\
Uncertain OPF & TD3/CPO & 0.0\% & -- & 300 viol. & -- & 91.8 / 30.1 & $<1$ ms\\
\bottomrule
\end{tabular}}
\end{table}

\paragraph{Safe multi-agent navigation.}
Table~\ref{tab:iclr-main-results} and
Figure~\ref{fig:iclr-multiagent} show the main safe-navigation result.
Over 100 evaluation episodes, GPC obtains $100\%$ safe success in static
scenes, whereas the finite-budget reference is safe in only $60\%$ and
the RL baselines have $0\%$ safe successes.  In moving scenes,
GPC reaches $99\%$ reachable-conditioned safe success and $100\%$
scene-collision-free rollouts after excluding one reference-unreachable
scene.  The
finite-budget reference is scene-collision-free in only $42\%$
and safe in only $35\%$; the RL baselines again have $0\%$ safe
success.  This is
the key reliability result: neither finite-budget online optimization
nor standard safe-RL variants are dependable deployment oracles in
dynamic multi-agent scenes.
Conditioned on the episodes where the reference is safe, GPC attains
a $-3.43\%$ aggregate cost gap.  Static and dynamic rollout
GIFs are included in the supplementary artifacts.

The margin columns in Table~\ref{tab:iclr-main-results} separate goal
success from geometric safety.  The dynamic benchmark has smaller
margins than the static benchmark, as expected, but the aggregate
minimum margins remain positive for scene safety, ego-obstacle
clearance, agent-obstacle clearance, and agent-agent separation.  This
supports the claim that the learned controller is not merely reaching the
goal with occasional near misses; it is preserving the explicit
geometric safety constraints throughout the rollout.

The RL results also clarify the role of shielding.  One-step shielded
PPO raises the scene-collision-free rate from roughly $5\%$ to $8\%$
in the static benchmark and remains near zero in the moving benchmark,
but it still obtains $0\%$ safe success because goal-reaching and
multi-agent safety must be satisfied simultaneously.

\begin{figure}[H]
\centering
\begin{subfigure}[t]{0.48\textwidth}
  \centering
  \includegraphics[width=\linewidth]{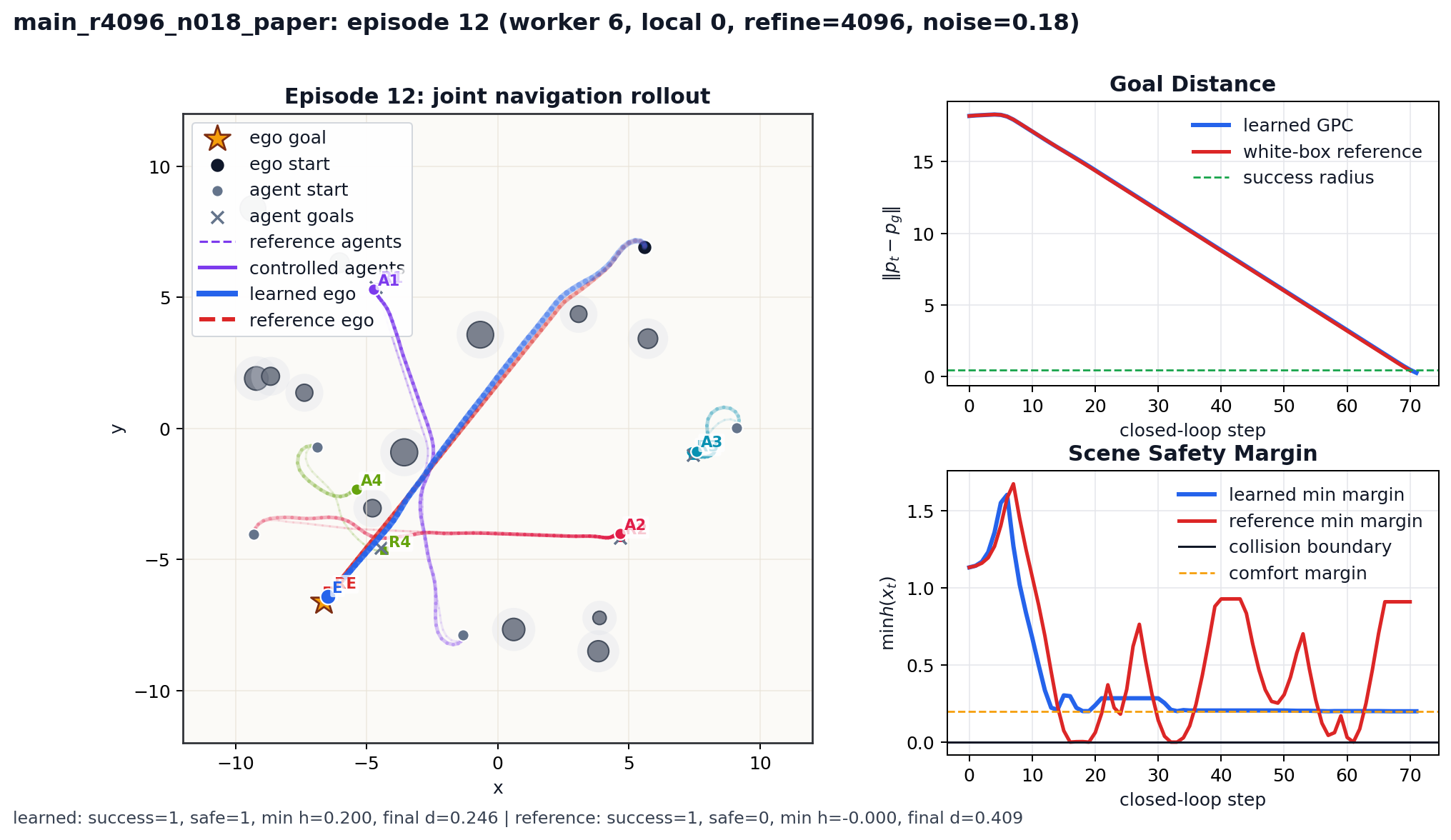}
  \caption{Static obstacles}
\end{subfigure}
\hfill
\begin{subfigure}[t]{0.48\textwidth}
  \centering
  \includegraphics[width=\linewidth]{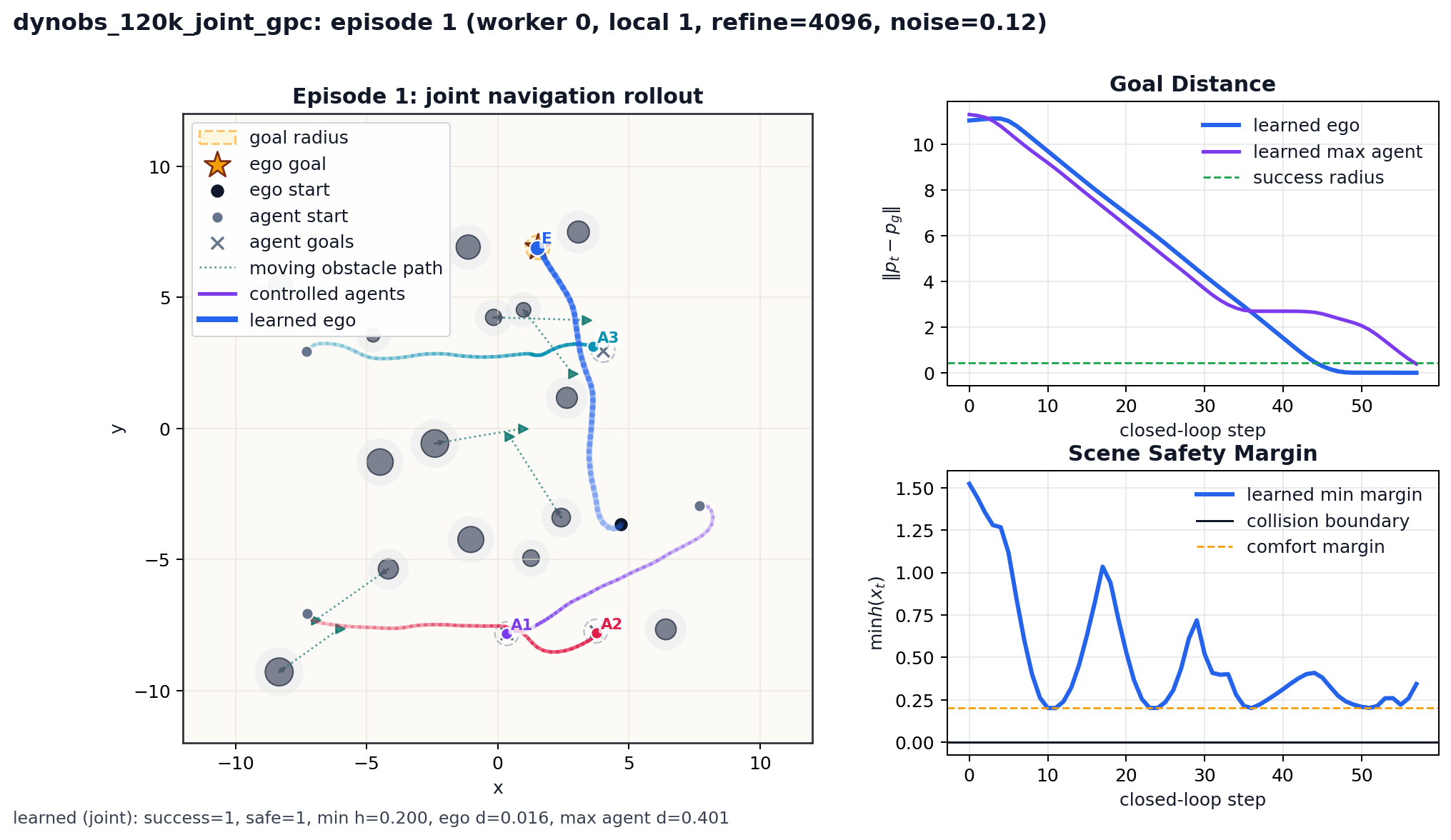}
  \caption{Moving obstacles: episode 1}
\end{subfigure}
\caption{Safe multi-agent navigation.  Each panel shows the joint
trajectory, goal-distance trace, and closed-loop scene-safety margin for
one held-out rollout.  The static case verifies the multi-agent
coordination pattern, while the moving-obstacle case shows that GPC
keeps positive safety margins in a dynamic scene.  Additional
moving-obstacle rollouts are shown in Appendix~\ref{sec:exp:multiagent-navigation}.
The finite-budget reference is safe in only $35\%$ of moving episodes,
whereas GPC is scene-collision-free in $100\%$ and reaches $99\%$
reachable-conditioned safe success.}
\label{fig:iclr-multiagent}
\end{figure}

\begin{figure}[!htpb]
\centering
\begin{subfigure}[t]{0.53\textwidth}
  \centering
  \includegraphics[height=4.0cm]{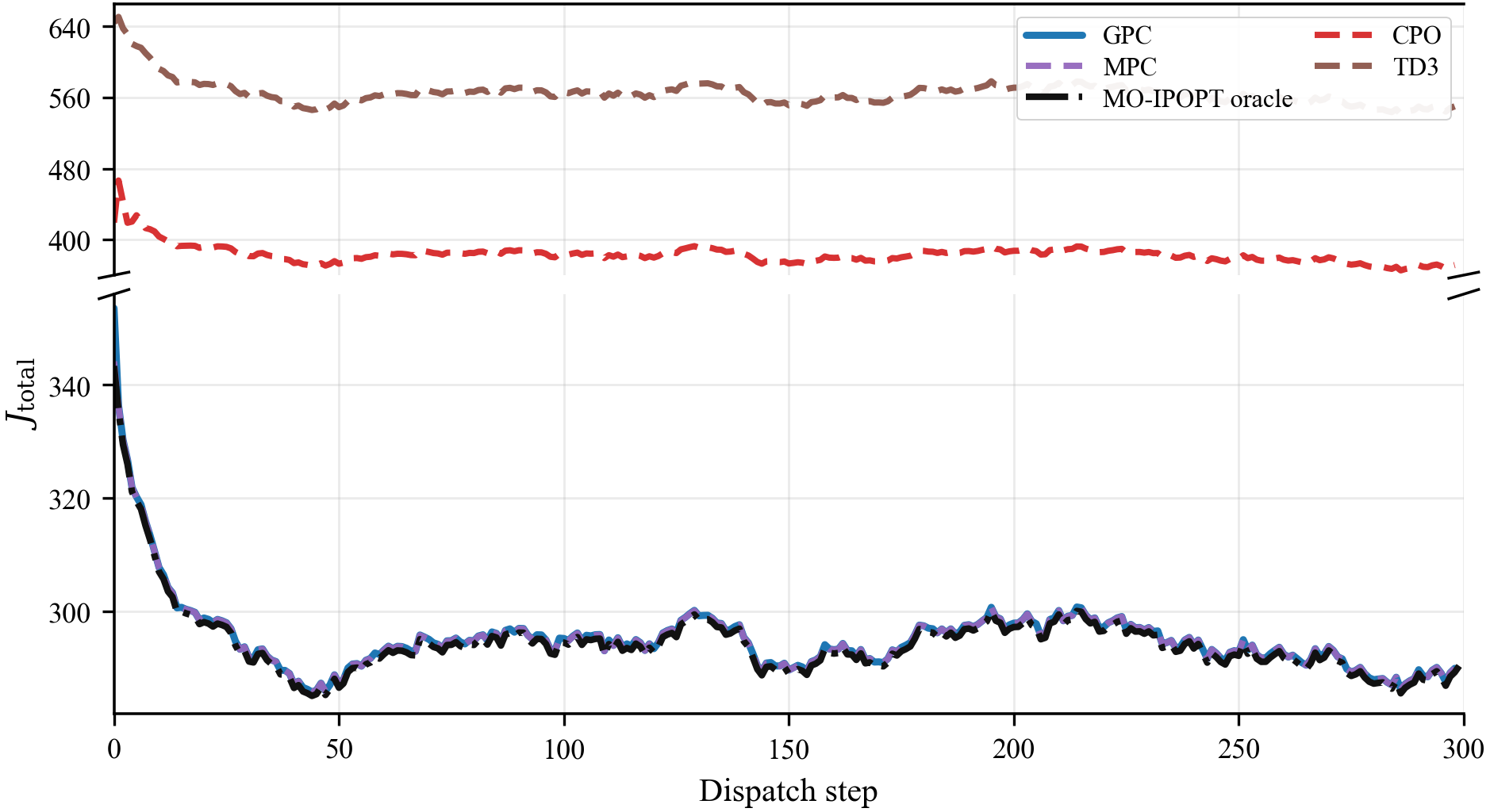}
  \caption{Uncertain OPF: total objective}
\end{subfigure}
\hfill
\begin{subfigure}[t]{0.42\textwidth}
  \centering
  \includegraphics[width=\linewidth]{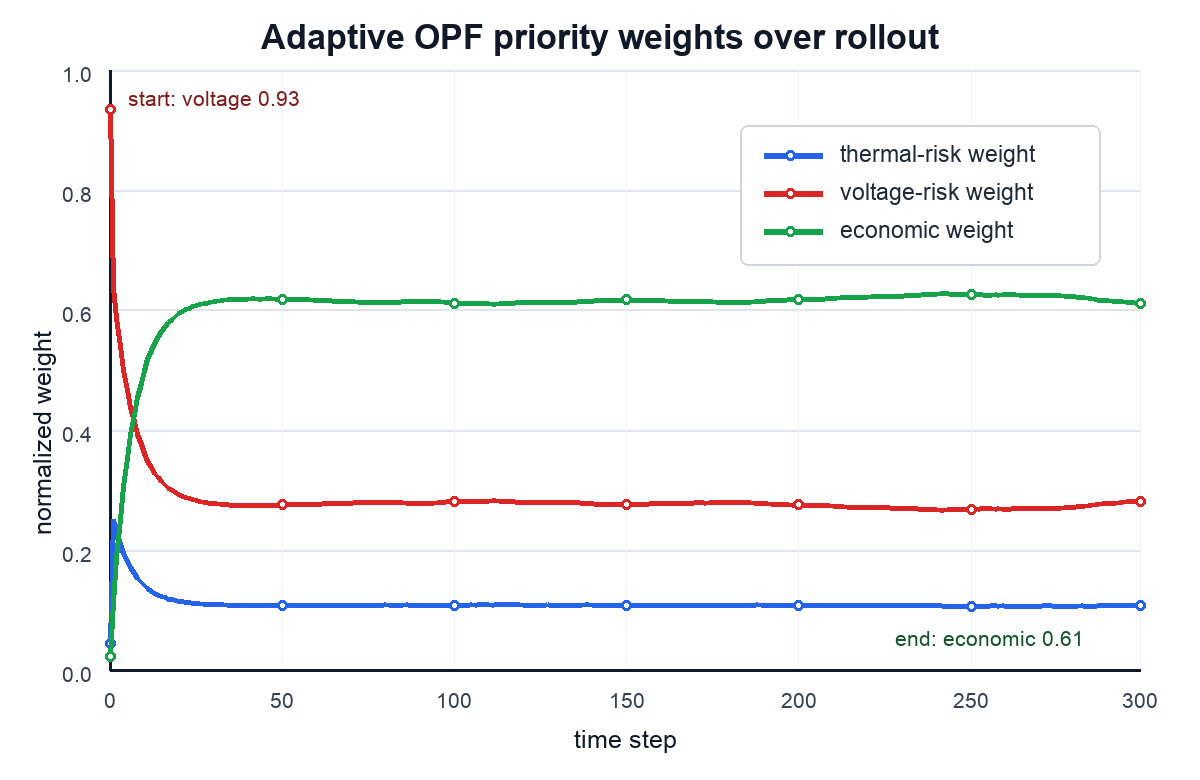}
  \caption{Uncertain OPF: adaptive weights}
\end{subfigure}
\caption{IEEE 30-bus OPF under branch-admittance uncertainty.  GPC
tracks the offline nonlinear oracle while avoiding repeated online
MO-IPOPT solves.  The adaptive semantic weights initially prioritize
voltage recovery and then return toward economic operation while
maintaining feasibility.}
\label{fig:iclr-opf}
\end{figure}

\paragraph{IEEE 30-bus OPF under branch-admittance uncertainty.}
The main OPF result evaluates the controller when each trajectory has a
fixed branch-admittance perturbation, so the online dispatch problem is
not identical to the nominal model used by a standard controller.  GPC
obtains 100\% empirical feasibility over the 300-step rollout and
0.25\% gap to the MO-IPOPT Pareto oracle at 12.6 ms per decision.  MPC
remains near-oracle in objective value but is about $97.5\times$ slower
and has one infeasible step.  TD3 and CPO, trained on the nominal
environment, produce no feasible dispatches under this protocol.  The
nominal OPF comparison is reported in Appendix~\ref{sec:exp:nobattery};
the uncertainty setting is the stronger test because the controller must
use the learned Pareto geometry while the physical network drifts from
the model used at training time.

\section{Conclusion}

GPC uses known dynamics and constraints to construct an offline Pareto
control geometry, then navigates this geometry online using a
state-dependent semantic coordinate.  The method is not a replacement
for physics models; it is a way to amortize constrained multi-objective
solves into a feasible real-time controller.  The formal margin and
non-accumulation results state when this amortized controller remains in
a feasible tube around the learned Pareto map.  The main navigation and
uncertainty-aware OPF results show that this geometry can improve
runtime reliability when finite-budget online optimization or model-free
RL fails to preserve safety; the analytical and nominal OPF appendix
studies verify the same snap, slide, semantic-priority, and
geometry-aware integration mechanisms in simpler oracle-comparison
settings.  Scaling GPC to larger systems will require many high-quality
Pareto fronts across operating conditions, making parallel Pareto-front
generation and uncertainty-aware training the natural next steps.

\FloatBarrier

\appendix
\section{Appendix Overview}
\label{app:extended-draft}
The appendix contains supplementary material that supports the main text without restarting the paper: extended Pareto-map construction details, the full online navigation update, benchmark definitions, appendix-only simulation cases, proof details, and the action-decoder training pipeline.

\section{Supplementary Offline Pareto Map Details}
\label{sec:offline}

The offline stage constructs a continuous, differentiable map 
$\mathcal{M}^* \subset \mathcal{Z}$ that encodes the supported family of
Pareto-optimal solutions as a low-dimensional latent manifold endowed 
with Lie group structure. The construction proceeds in three steps: 
dataset generation via multi-objective optimization, semantic coordinate 
labeling, and manifold learning via a structured loss system.

\subsection{Dataset Construction}
The offline dataset is constructed by sampling diverse operating  conditions from the parameter space and solving the multi-objective  optimization problem to Pareto-optimality for each scenario.
\subsubsection{Input-Space Sampling}
Let $\mathcal{P} \subseteq \mathbb{R}^{n_p}$ denote the disturbance or 
parameter space. We generate $N$ scenarios 
$\{\mathbf{p}_k\}_{k=1}^N \subset \mathcal{P}$ using Latin Hypercube 
Sampling (LHS), which partitions each dimension of $\mathcal{P}$ into $N$ 
equal-probability intervals and draws exactly one sample per interval. 
LHS provides uniform marginal stratification of $\mathcal{P}$ with 
$O(N)$ samples, avoiding the clustering artifacts of Monte Carlo sampling 
that would leave regions of the operational envelope unrepresented on 
$\mathcal{M}^*$.

\subsubsection{Pareto-Optimal Solution Generation}
For each scenario $\mathbf{p}_k$, the system is simulated or evaluated to 
obtain the corresponding true system state $\mathbf{s}_k \in \mathcal{S}$ 
and observation $\mathbf{x}_k = h(\mathbf{s}_k) \in \mathcal{X}$. 
We then solve the multi-objective optimization problem via weighted
scalarization:
\begin{equation}
  \mathbf{u}_k^*(\mathbf{w}) = \arg\min_{\mathbf{u} \in
  \mathcal{U}_{\mathbf{s}_k}^{\mathrm{feas}}} \sum_{i=1}^{M} w_i\,
  J_i(\mathbf{s}_k, \mathbf{u}), \qquad \mathbf{w} \in \Delta^{M-1}
\end{equation}
sweeping $\mathbf{w}$ uniformly across $\Delta^{M-1}$ to obtain diverse
Pareto-optimal solutions.
\begin{remark}[Supported Pareto Points]
Weighted scalarization recovers \emph{supported} Pareto-optimal solutions,
i.e.\ points on the convex hull of the Pareto front in objective space.
For problems with a nonconvex Pareto front, interior (non-supported) points
are not recoverable by any positive weight vector $\mathbf{w}\in\Delta^{M-1}$.
In such cases the dataset $\mathcal{D}$ covers only the supported subset of
$\mathcal{M}^*$; claims about the completeness of the learned Pareto manifold
should be understood as holding over this supported subset.
For the OPF benchmarks in this work the Pareto front is empirically
well-approximated by its supported points, so the distinction does not affect
practical performance; users with strongly nonconvex objectives should supplement
with $\varepsilon$-constraint or Chebyshev scalarization sweeps.
\end{remark} The semantic priority coordinate 
$\boldsymbol{\sigma}_k \in \Delta^{M-1}$ is computed from $\mathbf{s}_k$ 
via the violation indicator framework described in 
Section~\ref{sec:offline:semantic}. The dataset is:
\begin{equation}\label{eq:dataset}
  \mathcal{D} = \bigl\{(\mathbf{x}_k,\, \mathbf{u}_k^*,\, 
  \mathbf{s}_k,\, \boldsymbol{\sigma}_k,\, \mathbf{J}_k)\bigr\}_{k=1}^N
\end{equation}
where $\mathbf{s}_k$ is the true system state available only during 
offline data generation, $\mathbf{x}_k$ is the observation available 
at deployment, $\mathbf{u}_k^*$ is the Pareto-optimal action, 
$\boldsymbol{\sigma}_k$ is the semantic priority coordinate, and 
$\mathbf{J}_k = \mathbf{J}(\mathbf{s}_k, \mathbf{u}_k^*)$ is the 
objective vector.

\begin{remark}[Optional Next-Step Semantic Label]
When temporal coupling constraints are present, the dataset may be 
augmented to include the next-step true system state $\mathbf{s}_{k+1}$, 
from which the next-step semantic priority coordinate 
$\boldsymbol{\sigma}_{k+1}$ is computed offline via the same violation 
indicator framework. In this case the dataset expands to:
\begin{equation}
  \mathcal{D} = \bigl\{(\mathbf{x}_k,\, \mathbf{u}_k^*,\, 
  \mathbf{s}_k,\, \boldsymbol{\sigma}_k,\, \mathbf{J}_k,\,
  \mathbf{s}_{k+1})\bigr\}_{k=1}^N
\end{equation}
This captures temporal priority evolution when future safety or feasibility 
depends on current actions. In the absence of temporal coupling, 
$\mathbf{s}_{k+1}$ is redundant and may be omitted, reducing the 
dataset to \eqref{eq:dataset}.
\end{remark}

\begin{remark}[Pareto by Construction]
Every sample in $\mathcal{D}$ is obtained by solving the multi-objective 
problem to Pareto-optimality in the true system state space. The manifold 
$\mathcal{M}^*$ is therefore implicitly defined as the encoder image of 
$\mathcal{D}$, and explicit Pareto-embedding losses or dominance penalties 
are unnecessary for the training samples themselves. This avoids relying
on a soft Pareto penalty during training, although interpolation quality
between samples still depends on the learned manifold approximation.
\end{remark}

\subsection{Semantic Coordinate Generation}
\label{sec:offline:semantic}
Each Pareto-optimal solution is assigned a semantic priority coordinate
$\boldsymbol{\sigma}_k \in \Delta^{M-1}$ encoding the physical urgency 
of the operating condition under the true system state $\mathbf{s}_k$, 
which is available during offline data generation. For each objective $J_i$, 
define a normalized violation indicator 
$\delta_i(\mathbf{s}_k) \in [0,1]$, where $\delta_i = 0$ indicates 
nominal operation and $\delta_i = 1$ indicates the hard constraint boundary 
is reached. The violation indicators are mapped to urgency potentials via 
singular perturbation potential functions:
\begin{equation}
  \phi_i(\mathbf{s}_k) = \exp\!\left(\frac{k_i \cdot 
  \delta_i(\mathbf{s}_k)}{\epsilon_i}\right) - 1, \qquad 
  k_i > 0,\quad \epsilon_i \in (0,1]
\end{equation}
where $k_i$ encodes the intrinsic priority of objective $i$ and 
$\epsilon_i$ controls the steepness of the potential wall near the 
constraint boundary. The semantic coordinate is obtained by normalizing the urgency
potentials with a positive baseline $\rho > 0$:
\begin{equation}
  \sigma_i(\mathbf{s}_k)
  \;=\; \frac{\phi_i(\mathbf{s}_k) + \rho}
             {\displaystyle\sum_{j=1}^M \bigl(\phi_j(\mathbf{s}_k) + \rho\bigr)}
  \label{eq:semantic-coord}
\end{equation}
The baseline $\rho > 0$ prevents division by zero during
nominal operation (all $\delta_i = 0 \Rightarrow$ all $\phi_i = 0$),
and sets the default coordinate to the uniform point
$\boldsymbol{\sigma} = \mathbf{1}/M$ (equal economic priority) when
no constraint is active.
This construction ensures $\boldsymbol{\sigma}_k \in \Delta^{M-1}$ at all
times. As $\delta_i \to 1$, the exponential growth of $\phi_i$ causes
$\sigma_i \to 1$ continuously (since $\phi_i \gg \rho$ near the constraint
boundary), without any threshold or rule-based switching. 
The semantic coordinate thus encodes a smooth, physics-driven priority 
ordering that transitions autonomously from economic operation to emergency 
response as constraint violations intensify.

\begin{remark}[Semantic-to-Geometric Translation]
During training, $\boldsymbol{\sigma}_k$ provides supervision for the 
offline oracle encoder $E_o$, anchoring the geometry of $\mathcal{M}^*$ to the 
physical severity structure of the solution space. At deployment, 
the offline state-objective tuple
$(\mathbf{s}_t,\boldsymbol{\sigma}_t,\mathbf{J}_t)$
is not supplied to $E_o$.  Instead, $E_x$ localizes the observation
$\mathbf{x}_t$ on the learned map, while $\boldsymbol{\sigma}_t$ is
computed from online violation indicators as in
eq.~\eqref{eq:iclr-semantic}.  This enables real-time Pareto-optimal
control without querying the offline oracle encoder.
\end{remark}

\subsection{Manifold Learning via Structured Loss System}
\label{sec:loss}
The encoder-decoder system is trained on the offline dataset 
$\mathcal{D} = \{(\mathbf{x}_k, \mathbf{u}_k^*, \mathbf{s}_k, 
\boldsymbol{\sigma}_k, \mathbf{J}_k)\}_{k=1}^N$ to embed the 
Pareto-optimal solutions onto a geometrically structured latent manifold 
$\mathcal{M}^*$ via the following composite objective:
\begin{equation}
  \mathcal{L}_{\text{total}} = \mathcal{L}_{\text{recon}}^x
  + \mathcal{L}_{\text{action}}
  + \mathcal{L}_{\text{recon}}^o
  + \omega_1\,\mathcal{L}_{\text{consist}}
  + \omega_2\,\mathcal{L}_{\text{local}},  
\end{equation}
Each term addresses a distinct geometric requirement. For each sample 
$(\mathbf{x}_k, \mathbf{u}_k^*, \mathbf{s}_k, \boldsymbol{\sigma}_k, 
\mathbf{J}_k) \in \mathcal{D}$, the observation encoder and offline
oracle encoder produce latent codes:
\begin{equation}
  z_k^x = E_x(\mathbf{x}_k), \qquad 
  z_k^o = E_o(\mathbf{s}_k, \boldsymbol{\sigma}_k, \mathbf{J}_k)
\end{equation}
where $z_k^x$ is available at deployment and $z_k^o$ is the offline
state-objective code available only during training.

\paragraph{State Reconstruction ($\mathcal{L}_{\text{recon}}^x$):}
Forces $z_k^x = E_x(\mathbf{x}_k)$ to encode sufficient information 
from the observation $\mathbf{x}_k$ to reconstruct the true 
system state $\mathbf{s}_k$:
\begin{equation}
  \mathcal{L}_{\text{recon}}^x = \frac{1}{N}\sum_{k=1}^N 
  \|\mathbf{s}_k - D_s(z_k^x)\|^2
\end{equation}

\paragraph{Action Reconstruction ($\mathcal{L}_{\text{action}}$):}
Forces $z_k^x$ to recover the Pareto-optimal control action 
$\mathbf{u}_k^*$ from the observation $\mathbf{x}_k$ alone, 
which is the primary deployment objective. Applied to $z_k^x$ only:
\begin{equation}
  \mathcal{L}_{\text{action}} = \frac{1}{N}\sum_{k=1}^N 
  \|\mathbf{u}_k^* - D_u(z_k^x)\|^2
\end{equation}
Applying this loss to $z_k^o = E_o(\mathbf{s}_k, \boldsymbol{\sigma}_k, 
\mathbf{J}_k)$ would be trivial since $\mathbf{u}_k^*$ is already 
recoverable from the inputs to $E_o$ and would provide no useful 
gradient signal.

\paragraph{Omniscient State Reconstruction ($\mathcal{L}_{\text{recon}}^o$):}
Forces $z_k^o$ to encode a geometrically faithful representation of the 
full solution tuple $(\mathbf{s}_k, \boldsymbol{\sigma}_k, \mathbf{J}_k)$, 
providing a well-structured oracle signal for the consistency loss below:
\begin{equation}
  \mathcal{L}_{\text{recon}}^o = \frac{1}{N}\sum_{k=1}^N 
  \|\mathbf{s}_k - D_s(z_k^o)\|^2
\end{equation}
Both paths share the decoder $D_s$, ensuring that $z_k^x$ and $z_k^o$ 
occupy the same geometric space and are therefore directly comparable.

\paragraph{Cross-Modal Consistency ($\mathcal{L}_{\text{consist}}$):}
Enables $E_o$ to be discarded entirely at deployment by forcing $E_x$ 
and $E_o$ to produce coincident latent codes for the same sample 
$(\mathbf{x}_k, \mathbf{s}_k, \boldsymbol{\sigma}_k, \mathbf{J}_k) 
\in \mathcal{D}$:
\begin{equation}
  \mathcal{L}_{\text{consist}} = \frac{1}{N}\sum_{k=1}^N
  \left[\left(1 - \frac{\langle z_k^x, z_k^o\rangle}
  {\|z_k^x\|_2\|z_k^o\|_2}\right)
  + \beta\,\|z_k^x - z_k^o\|^2_2\right],  
\end{equation}
The cosine term enforces directional alignment in latent space; the 
squared Euclidean term enforces magnitude alignment. Together they 
align the deployment encoder $E_x$ with the offline oracle encoder $E_o$.

\paragraph{Local Structure Preservation ($\mathcal{L}_{\text{local}}$):}
This is the geometrically most critical term. It enforces a
continuous embedding from the parameter space $\mathcal{P}$ into the latent
manifold by requiring that physically neighboring scenarios
$\mathbf{p}_k, \mathbf{p}_j \in \mathcal{P}$ map to neighboring points
on $\mathcal{M}^*$:
\begin{equation}
  \mathcal{L}_{\text{local}} = \frac{1}{N}\sum_{k=1}^N
  \sum_{j\in\mathcal{N}_\kappa(k)}
  \exp\!\left(-\frac{\|\mathbf{p}_k - \mathbf{p}_j\|^2}{\sigma_p^2}\right)
  \cdot \|z_k^x - z_j^x\|^2
\end{equation}
where $\mathcal{N}_\kappa(k)$ is the $\kappa$-nearest neighborhood of 
scenario $k$ in the parameter space $\mathcal{P}$. The Gaussian kernel 
weight decays with physical distance, concentrating the topological 
constraint on the most immediate neighbors. Through 
$\mathcal{L}_{\text{local}}$, the temporal continuity requirement 
$\|\mathbf{u}_{t+1} - \mathbf{u}_t\| \le \delta$ is transformed into 
a static geometric neighborhood relation on $\mathcal{M}^*$, enabling 
the online navigator to satisfy ramp-rate constraints through latent 
geometry rather than explicit enforcement.

\paragraph{Rollout-aware action decoder refinement.}
After the manifold training above, and before the online stage is
deployed, we run a final action-decoder refinement step designed to
match the latent distribution encountered during rollout. In this
stage the manifold modules $(E_x,E_o,D_s)$ and the Lie-group embedding
are frozen, while only $D_u$ is updated using horizon-indexed rollout
tuples, trajectory-risk weights, intermediate rollout latents, and the
same Lie-local correction used by the online integrator.  This
rollout-aware training reduces the gap between static supervised
latents and post-integration latents generated by
Algorithm~\ref{alg:online}. The full objective and its relation to the
closed-loop controller are given in Appendix~\ref{app:du_training}.

\section{Supplementary Online Pareto Map Navigation}
\label{sec:online}

At deployment, the offline Pareto map $\mathcal{M}^*$ is fixed and the
online navigator operates in real time. Each control cycle consists
of five steps: manifold localization from the observation,
autonomous semantic target generation from real-time violation
indicators, nominal Riemannian gradient flow via a single autograd pass,
Lie-local residual correction in the decoder-pullback metric basis,
and Lie group integration with optional retraction and output clamping.
The entire cycle is deterministic and non-iterative, and in the
experiments executes in milliseconds without invoking an online
optimization solver.

\subsection{State Estimation and Manifold Localization}
At each time step $t$, the physical encoder produces an initial latent 
estimate:
\begin{equation}
  z_{\text{obs},t} = E_x(\mathbf{x}_t)
\end{equation}
This estimate may be corrupted by sensor noise, partial observability, 
or out-of-distribution inputs. Committing to a nearest-point projection 
onto $\mathcal{M}^*$ is insufficient: as $z_{\text{obs},t}$ drifts 
across a Voronoi boundary, the projection jumps discontinuously, 
producing an abrupt control transient. Instead, we admit only those 
manifold points that are consistent with the current observation within 
the sensor noise floor, defining the \emph{measurement-consistent 
manifold subset}:
\begin{equation}
  \mathcal{M}_t \;\triangleq\; \bigl\{z\in\mathcal{M}^* \;\big|\; 
  \|\mathbf{x}_t - D_s(z)\|^2 \le \tau_t \bigr\}, \qquad 
  \tau_t = \hat{\sigma}_t^2 + \tau_{\text{geom}}
\end{equation}
where $\hat{\sigma}_t^2$ is an online estimate of observation noise 
variance and $\tau_{\text{geom}}$ is a geometric tolerance calibrated 
to the generalization error of $\mathcal{L}_{\text{local}}$: 
specifically, $\tau_{\text{geom}}$ is set to match the empirical 
reconstruction residual $\|\mathbf{x}_k - D_s(z_k^x)\|^2$ averaged 
over a held-out validation partition of $\mathcal{D}$ after offline 
training. This calibration is intended to keep $\mathcal{M}_t$ nonempty 
for observations drawn from the training distribution, while remaining 
tight enough to exclude manifold regions geometrically inconsistent 
with $\mathbf{x}_t$. The threshold $\tau_t$ thus carries a principled 
two-part interpretation: $\hat{\sigma}_t^2$ accounts for sensor noise 
at deployment, and $\tau_{\text{geom}}$ accounts for the residual 
manifold approximation error that $\mathcal{L}_{\text{local}}$ did 
not fully collapse during training. The localized state is obtained 
as:
\begin{equation}
  z_t = (1-\alpha)\,z_{t-1} + \alpha\,\arg\min_{z\in\mathcal{M}_t}
  \|z - z_{\text{obs},t}\|^2
\end{equation}
with first-order exponential smoothing ($\alpha \in (0,1]$) that 
suppresses high-frequency sensor noise while preserving responsiveness 
to genuine state changes.

\subsection{Energy Function and Nominal Flow}
\label{sec:online:energy}
The navigator's direction of motion at each time $t$ is determined 
by a composite \emph{energy function} defined on 
$\mathcal{M}^*$:
\begin{equation}
  V(z_t,\mathbf{x}_t) \;=\; \frac{1}{\epsilon}\,\Phi_{\text{normal}}(z_t,\mathbf{x}_t) 
  \;+\; \Phi_{\text{tangent}}(z_t)
  \label{eq:total-energy}
\end{equation}
No runtime operator-supplied weight vector, reference trajectory, or 
mode-switching logic enters this expression: both terms are built 
from the same physically grounded urgency potentials used to 
supervise the manifold during offline training, and the only 
runtime inputs are the observation $\mathbf{x}_t$ and the 
current latent state $z_t$. Descending $V$ produces the nominal 
vector field:
\begin{equation}
  \mathbf{F}(z_t, \mathbf{x}_t) \;\triangleq\; -\nabla_z V(z_t,\mathbf{x}_t)
\end{equation}
computed in a single autograd pass through the frozen decoders $D_u$ 
and $D_s$. The remainder of this subsection specifies each component 
of $V$, then shows how the two terms interact to induce a 
two-timescale response.

\paragraph{Normal potential: manifold restoration.}
The normal potential measures the reconstruction residual between 
the current observation and the decoded manifold point:
\begin{equation}
  \Phi_{\text{normal}}(z_t, \mathbf{x}_t) = 
  \|\mathbf{x}_t - D_s(z_t)\|^2
\end{equation}
When $z_t \in \mathcal{M}^*$, $\Phi_{\text{normal}} \approx 0$. Any 
off-manifold drift raises this term, generating a restoring gradient 
that drives $z_t$ back onto $\mathcal{M}^*$.

\paragraph{Tangential potential: autonomous multi-objective scalarization.}
\label{sec:online:semantic}
The tangential potential is the multi-objective cost scalarized by 
a priority coordinate $\boldsymbol{\sigma}_t \in \Delta^{M-1}$ 
that is computed from the decoded state (not supplied as an operator
hyperparameter):
\begin{equation}
  \Phi_{\text{tangent}}(z_t) = \sum_{i=1}^{M} 
  \sigma_{t,i}\cdot J_i\!\bigl(D_s(z_t),\,D_u(z_t)\bigr)
\end{equation}
The priority coordinate is generated from real-time violation
indicators $\delta_i(\hat{\mathbf{s}}_t) \in [0,1]$ evaluated on the
reconstructed state $\hat{\mathbf{s}}_t = D_s(z_t)$ (using the same
definitions as Section~\ref{sec:offline:semantic}), passed through
singular-perturbation urgency potentials, and normalized via the
same baseline-shifted formula as the offline stage
(equation~\eqref{eq:semantic-coord}):
\begin{equation}
  \phi_i(\hat{\mathbf{s}}_t) = \exp\!\left(\frac{k_i\,
  \delta_i(\hat{\mathbf{s}}_t)}{\epsilon_i}\right) - 1, \qquad
  \sigma_{t,i}
  = \frac{\phi_i(\hat{\mathbf{s}}_t) + \rho}
         {\displaystyle\sum_{j=1}^M\bigl(\phi_j(\hat{\mathbf{s}}_t)+\rho\bigr)}
\end{equation}
where $\rho > 0$ is the same baseline as in the offline stage,
ensuring a well-defined coordinate during nominal operation
(all $\delta_i = 0$), and a first-order low-pass filter suppresses
chattering when any $\delta_i$ hovers near an activation threshold. The exponential 
structure of $\phi_i$ produces a smooth, physics-driven priority 
shift: under nominal operation where all $\delta_i$ are small, the 
slow performance component of $\boldsymbol{\sigma}_t$ dominates; as 
any $\delta_i \to 1$, the explosive growth of the corresponding 
$\phi_i$ drives $\sigma_{t,i} \to 1$ continuously, redirecting the 
flow toward constraint restoration without a discrete mode switch. 
This is the online counterpart of the 
offline semantic labeling of Section~\ref{sec:offline:semantic}: 
the same potential that shaped the manifold geometry during 
training now shapes the navigation direction at runtime.

\paragraph{Timescale separation from a single energy.}
The scalar $\epsilon > 0$ in \eqref{eq:total-energy} encodes the 
ratio between the fast and slow timescales of the control problem:
\begin{equation}
  \epsilon = \frac{T_{\text{fast}}}{T_{\text{slow}}}
\end{equation}
with $T_{\text{fast}}$ the required constraint-restoration time and 
$T_{\text{slow}}$ the performance-optimization horizon. Because 
$1/\epsilon \gg 1$, the normal potential is a steep well orthogonal 
to $\mathcal{M}^*$, while the tangential potential is a shallow 
landscape along it. Descending $V$ therefore produces two regimes 
from the same gradient:
\begin{itemize}
  \item \textbf{Snap} ($O(\epsilon)$, off-manifold): 
  $(1/\epsilon)\nabla_z\Phi_{\text{normal}}$ dominates and rapidly 
  restores $z_t$ to $\mathcal{M}^*$ before any significant tangential 
  motion occurs.
  \item \textbf{Slide} ($O(1)$, on-manifold): 
  $\Phi_{\text{normal}} \approx 0$ and 
  $\nabla_z\Phi_{\text{tangent}}$ navigates along $\mathcal{M}^*$ 
  toward the state satisfying the current priority $\boldsymbol{\sigma}_t(z_t)$.
\end{itemize}
The two regimes arise from the same energy expression rather than
from an explicit switching condition or controller handoff. The
relative fast/slow scale is still set by the design parameter
$\epsilon$.

\subsection{Geometry-Aware State Evolution}
\label{sec:online:evolution}
Given the nominal field $\mathbf{F}(z_t,\mathbf{x}_t)$, the latent 
state is advanced by a single geometry-aware integrator that combines 
a second-order Runge-Kutta predictor, a curvature-aware residual 
computed in a decoder-induced metric basis, an optional retraction 
onto $\mathcal{M}^*$, and a final output clamp on the decoded action. 
The integrator is deterministic and non-iterative.

The nominal field is evaluated at a midpoint following a standard 
second-order Runge-Kutta scheme (Appendix~\ref{app:rk2}); the resulting 
Euclidean step $\Delta z_{\text{Euc}}$ is then corrected as follows.

\paragraph{Lie-local residual correction.}
The Euclidean predictor treats all latent directions uniformly, 
ignoring the fact that the decoder $D_s$ induces a non-trivial 
Riemannian metric on $\mathcal{Z}$. We correct this with a small 
geometry-aware residual computed in a local metric basis. The 
decoder-pullback metric at $z_t$ is:
\begin{equation}
  G(z_t) = J_{D_s}(z_t)^\top J_{D_s}(z_t) + \lambda_{\mathrm{m}}\,
  \mathbf{I}
\end{equation}
whose eigendecomposition yields principal directions of decoder 
sensitivity. Retaining the top-$k$ eigenvectors $B \in 
\mathbb{R}^{n_z\times k}$ (typically $k=3$) and their scales 
$\mathbf{s} = \sqrt{\mathrm{diag}(\Lambda_{[-k:]})}$ identifies the 
directions along which latent motion produces the largest first-order 
change in the decoded state. Projecting the Euclidean step into this 
basis, partitioning into three-dimensional sub-blocks, and rotating 
each sub-block via an $\mathfrak{so}(3)$ matrix exponential gives a 
curvature-aware residual:
\begin{equation}
  \Delta z_{\text{Lie}} = B\,\bigl[\mathrm{Rot}_{\mathfrak{so}(3)}\!
  \bigl(B^\top\Delta z_{\text{Euc}},\; (B^\top k_2)\oslash\mathbf{s};
  \; s_{\mathrm{L}},\,\Delta t\bigr) - B^\top\Delta z_{\text{Euc}}\bigr]
\end{equation}
where $\oslash$ is elementwise division (yielding scale-invariant 
angular velocities), $s_{\mathrm{L}}$ is a Lie-scale parameter, and 
sub-blocks of size less than two are left unrotated. The residual is 
deliberately small and guarded: it is confined to the top-$k$ 
eigendirections of $G(z_t)$ and added to (rather than replacing) the 
Euclidean step, and a NaN-safe fallback $\Delta z_{\text{Lie}} 
\leftarrow \mathbf{0}$ applies whenever any component is non-finite. 
This design ensures the correction acts only where the decoder 
geometry is informative, while the RK2 step continues to drive the 
bulk of the latent motion.

\paragraph{Lie group integration and retraction.}
The combined update is applied via left translation on the Lie 
group $G$, which preserves manifold membership algebraically by the 
closure axiom, followed by an optional retraction onto $\mathcal{M}^*$:
\begin{equation}
  z_{t+1} = (1-\mu_{\mathrm{R}})\,\bar{z}_{t+1} + \mu_{\mathrm{R}}\,
  \Pi_{\mathcal{M}^*}(\bar{z}_{t+1}), \qquad
  \bar{z}_{t+1} \triangleq z_t \oplus \bigl(\Delta z_{\text{Euc}} + 
  \gamma_{\mathrm{L}}\,\Delta z_{\text{Lie}}\bigr)
  \label{eq:final-update}
\end{equation}
Here $\gamma_{\mathrm{L}} \in [0,1]$ controls the strength of the 
Lie residual and $\mu_{\mathrm{R}} \in [0,1]$ controls the retraction 
mix. Setting $\gamma_{\mathrm{L}} = \mu_{\mathrm{R}} = 0$ recovers 
pure RK2; small $\gamma_{\mathrm{L}} \in [0.05, 0.2]$ adds curvature 
awareness without destabilizing the primary dynamics; 
$\mu_{\mathrm{R}} > 0$ suppresses any residual off-manifold drift. 
The executed control action is decoded from the post-update latent and
clamped to its physical box 
bounds:
\begin{equation}
  \mathbf{u}_{t} = \mathrm{clip}\!\bigl(D_u(z_{t+1}),\;
  \mathbf{u}^{\min},\,\mathbf{u}^{\max}\bigr)
\end{equation}
which enforces the box constraints independently of the learned
decoder accuracy.

\begin{remark}[Feasibility of the Decoded Action]
\label{rem:feasibility}
Because every training sample is generated by solving the
multi-objective problem to feasibility, points on the learned
manifold are trained to decode to feasible actions.
Under the positive feasibility margin, Lipschitz continuity of
$D_u$, and the velocity cap
$V_{\max}\Delta t \le \delta_{\mathrm{feas}}$, Proposition~\ref{prop:feasibility}
gives a sufficient condition for
$\mathbf{f}(\mathbf{u}_t) \le \mathbf{0}$ at each cycle without
runtime constraint solving.
See Proposition~\ref{prop:feasibility} in the appendix for
the formal statement and proof, which covers both convex and
nonconvex objectives.
\end{remark}

\begin{remark}[Additions Fit Within Proposition~\ref{prop:drift}]
The velocity cap enforces $\|\dot{z}\| \le V_{\max}$ directly; the 
Lie residual is bounded by $\|\gamma_{\mathrm{L}}\Delta z_{\text{Lie}}\| 
\le \gamma_{\mathrm{L}}\Delta t\,V_{\max}$ by construction; the 
retraction is a non-expansion on $\mathcal{M}^*$. All three additions 
fit within the drift bound of Proposition~\ref{prop:drift} without 
modifying the non-accumulation argument, and the end-to-end action 
error bound of Theorem~\ref{thm:error} remains valid.
\end{remark}
\subsection{Online Algorithm}

\begin{algorithm}[!htbp]
\caption{Online Pareto Map Navigation}
\label{alg:online}
\small
\begin{algorithmic}[1]
\Require $\mathbf{x}_t$, $z_{t-1}$, $\hat\sigma_t^2$,
  $\tau_{\text{geom}}$, $\Delta t$, $\alpha$, $\epsilon$,
  $V_{\max}$, $\lambda_{\mathrm{m}}$, $(k,\gamma_{\mathrm{L}},
  s_{\mathrm{L}})$, $\mu_{\mathrm{R}}$,
  $(\mathbf{u}^{\min},\mathbf{u}^{\max})$
\Ensure $\mathbf{u}_{t}$, $z_{t+1}$
\State \textbf{// 1: Manifold Localization}
\State $z_{\text{obs},t} \!\leftarrow\! E_x(\mathbf{x}_t)$;\quad
  $\mathcal{M}_t \!\leftarrow\! \{z\!\in\!\mathcal{M}^*\mid
  \|\mathbf{x}_t{-}D_s(z)\|^2 \le \hat\sigma_t^2{+}\tau_{\text{geom}}\}$
\State $z_t \leftarrow (1{-}\alpha)\,z_{t-1} + \alpha\,
  \arg\min_{z\in\mathcal{M}_t}\|z - z_{\text{obs},t}\|^2$
\State \textbf{// 2: Energy Function \& Nominal Field}
  \hfill\Comment{single autograd pass}
\State $\hat{\mathbf{s}}_t \!\leftarrow\! D_s(z_t)$;\quad
  $\sigma_{t,i} \!\leftarrow\!
  (\phi_i{+}\rho)/\textstyle\sum_j(\phi_j{+}\rho)$\;(eq.\,\eqref{eq:semantic-coord});\quad
  $\mathbf{F}(z_t) \!\leftarrow\! -\nabla_z\!\bigl[
  \tfrac{1}{\epsilon}\|\mathbf{x}_t{-}D_s(z_t)\|^2
  {+} \textstyle\sum_i\sigma_{t,i}J_i(D_s,D_u)\bigr]$
\State \textbf{// 3: RK2 Predictor} \hfill\Comment{cap each slope at $V_{\max}$}
\State Compute $k_1,z_{\text{mid}},k_2,\Delta z_{\text{Euc}}$
  via RK2 (App.~\ref{app:rk2}) with $\Pi_{\|\cdot\|\le V_{\max}}$
\State \textbf{// 4: Lie-Local Residual Correction}
\State $G \!\leftarrow\! J_{D_s}^\top J_{D_s}{+}\lambda_{\mathrm{m}}\mathbf{I}$;\quad
  $(B,\mathbf{s})\!\leftarrow\!$ top-$k$ eigenbasis of $G$
\State $\Delta z_{\text{Lie}} \!\leftarrow\! B\bigl[
  \mathrm{Rot}(B^\top\!\Delta z_{\text{Euc}},
  (B^\top k_2){\oslash}\mathbf{s};\,s_{\mathrm{L}},\Delta t)
  - B^\top\!\Delta z_{\text{Euc}}\bigr]$
  \Comment{NaN$\to\mathbf{0}$}
\State \textbf{// 5: Lie Group Integration, Retraction \& Decode}
\State $\bar{z}_{t+1} \!\leftarrow\!
  z_t \oplus (\Delta z_{\text{Euc}}{+}\gamma_{\mathrm{L}}\Delta z_{\text{Lie}})$;\quad
  $z_{t+1} \!\leftarrow\! (1{-}\mu_{\mathrm{R}})\bar{z}_{t+1}
  {+}\mu_{\mathrm{R}}\Pi_{\mathcal{M}^*}(\bar{z}_{t+1})$
\State $\mathbf{u}_{t} \leftarrow
  \mathrm{clip}(D_u(z_{t+1}),\mathbf{u}^{\min},\mathbf{u}^{\max})$;\quad
  \Return $\mathbf{u}_{t},\,z_{t+1}$
\end{algorithmic}
\end{algorithm}

Algorithm~\ref{alg:online} summarizes the complete online navigation 
cycle. Each iteration proceeds through five sequential steps: 
manifold localization from the observation, energy-function 
evaluation with a single shared autograd pass yielding the nominal 
field $\mathbf{F}(z_t)$, an RK2 predictor with radial velocity cap, 
a Lie-local residual correction computed in the top-$k$ 
eigendirections of the decoder-pullback metric, and Lie group 
integration with optional retraction and output clamping. The entire 
cycle is deterministic and non-iterative in its outer loop: two 
autograd passes (for $k_1$ and $k_2$) and one small $k \times k$ 
eigendecomposition (typically $k\le 3$) together constitute the 
dominant computational cost, which is milliseconds in our analytical
 experiments. No online learning, parameter update, 
or iterative solver is invoked. All numerically sensitive operations 
are guarded by NaN-safe fallbacks, the geometric drift introduced by
Lie group integration is non-accumulating under the assumptions of
Proposition~\ref{prop:drift} (asymptotic) and
Proposition~\ref{prop:closed-loop} (uniform, closed-loop),
and feasibility of the decoded action follows from
Remark~\ref{rem:feasibility} and
Proposition~\ref{prop:feasibility} under the stated margin and
localization conditions.

\section{Second-Order Latent Integration via RK2}
\label{app:rk2}

Algorithm~\ref{alg:online} uses a second-order midpoint 
Runge-Kutta predictor to compute the Euclidean latent step before 
the Lie-local residual correction is applied. Let 
$\mathbf{F}(z,\mathbf{x}_t)=-\nabla_z V(z,\mathbf{x}_t)$ denote the 
nominal field for the current observation. The raw first slope is:
\begin{equation}
  \tilde{k}_1 = \mathbf{F}(z_t,\mathbf{x}_t).
\end{equation}
To enforce the bounded-velocity condition used in the drift analysis, 
the slope is radially capped:
\begin{equation}
  k_1 = \Pi_{\|\cdot\|\le V_{\max}}(\tilde{k}_1), \qquad
  \Pi_{\|\cdot\|\le V_{\max}}(v)
  = v\min\!\left\{1,\frac{V_{\max}}{\|v\|+\delta_v}\right\},
  \label{eq:rk2-cap}
\end{equation}
with a small $\delta_v>0$ preventing division by zero. The midpoint 
latent state is:
\begin{equation}
  z_{\text{mid}} = z_t + \frac{1}{2}\Delta t\,k_1.
\end{equation}
The nominal field is then evaluated at this midpoint and capped 
again:
\begin{equation}
  \tilde{k}_2 = \mathbf{F}(z_{\text{mid}},\mathbf{x}_t), \qquad
  k_2 = \Pi_{\|\cdot\|\le V_{\max}}(\tilde{k}_2).
\end{equation}
The Euclidean predictor passed to the geometry-aware residual 
correction is:
\begin{equation}
  \Delta z_{\text{Euc}} = \Delta t\,k_2.
  \label{eq:rk2-step}
\end{equation}

\paragraph{NaN-safe fallback.}
All midpoint computations are guarded. If any component of 
$z_{\text{mid}}$ is non-finite, the algorithm sets 
$z_{\text{mid}}\leftarrow z_t$ before evaluating the second slope. 
If $\tilde{k}_2$ is non-finite, the second slope falls back to the 
already capped first slope, $k_2\leftarrow k_1$. If $k_1$ is itself 
non-finite, the predictor returns $\Delta z_{\text{Euc}}=\mathbf{0}$ 
for that cycle. These fallbacks are conservative: they may temporarily 
reduce the integrator to Euler or to a no-op step, but they prevent 
non-finite decoder or autograd values from entering the Lie update and 
therefore preserve the bounded-displacement invariant 
$\|\Delta z_{\text{Euc}}\|\le V_{\max}\Delta t$.

\paragraph{Accuracy relative to Euler.}
The explicit Euler predictor is:
\begin{equation}
  z_{t+1}^{\text{Euler}} = z_t + \Delta t\,\mathbf{F}(z_t,\mathbf{x}_t).
\end{equation}
For a smooth field, the exact flow satisfies:
\begin{equation}
  z(t+\Delta t) = z_t + \Delta t\,\mathbf{F}(z_t)
  + \frac{\Delta t^2}{2}\,J_{\mathbf{F}}(z_t)\mathbf{F}(z_t)
  + O(\Delta t^3),
\end{equation}
so Euler omits the quadratic curvature term and has local truncation 
error $O(\Delta t^2)$ and global error $O(\Delta t)$. The midpoint 
predictor expands as:
\begin{align}
  \mathbf{F}\!\left(z_t+\frac{1}{2}\Delta t\,\mathbf{F}(z_t)\right)
  &=
  \mathbf{F}(z_t)
  + \frac{\Delta t}{2}\,J_{\mathbf{F}}(z_t)\mathbf{F}(z_t)
  + O(\Delta t^2),\\
  z_t+\Delta t\,k_2
  &=
  z_t + \Delta t\,\mathbf{F}(z_t)
  + \frac{\Delta t^2}{2}\,J_{\mathbf{F}}(z_t)\mathbf{F}(z_t)
  + O(\Delta t^3).
\end{align}
Thus RK2 matches the exact flow through second order, yielding local 
truncation error $O(\Delta t^3)$ and global error $O(\Delta t^2)$ 
whenever the velocity cap is inactive. When the cap is active, the 
same order statement applies to the capped vector field defined by 
\eqref{eq:rk2-cap}; the cap trades formal accuracy for the physical 
requirement that no single control cycle can move farther than 
$V_{\max}\Delta t$ in latent space.

The Lie-local residual correction in Section~\ref{sec:online:evolution} 
does not replace this predictor. It operates only as a small 
decoder-metric correction to $\Delta z_{\text{Euc}}$, so the RK2 step 
continues to provide the primary integration accuracy while the Lie 
term reduces curvature-induced drift on $\mathcal{M}^*$.

\section{Action Decoder Training Pipeline}
\label{app:du_training}

The offline manifold training of Section~\ref{sec:offline} supervises
$D_u$ with a pointwise static loss $\mathcal{L}_{\text{action}}$ at
observed dataset codes.  This is insufficient for online deployment,
where $D_u$ must produce accurate actions at \emph{rollout latents}
generated by the Lie-group integrator—codes that may not coincide with
any training-set point.  Here we describe the dedicated $D_u$ training
procedure.  All manifold modules ($E_x$, $E_o$, $D_s$, group
embedding) remain frozen; only $D_u$ is updated.

\paragraph{Dataset and trajectory-risk weighting.}
Training uses rollout-indexed tuples
$\{(z_{\xi,h}, \mathbf{u}_{\xi,h}^*,
z_{\xi,h+1}, \mathbf{u}_{\xi,h+1}^*)\}$ with
$z_{\xi,h} = E_x(\mathbf{x}_{\xi,h})$ computed from frozen $E_x$.
Here $h$ indexes a domain-dependent rollout horizon; for domains whose
decoder returns a full plan, $\mathbf{u}_{\xi,h}^*$ denotes the
corresponding horizon action extracted from that plan.  Below, $D_u$
denotes this extracted action, i.e., $\Pi_hD_u$ when the decoder output
is a plan.  To focus
learning on dynamically demanding transitions, each tuple receives a
normalized importance weight:
\begin{equation}
  w_{\xi,h} \;=\; \frac{\alpha_u\,\|\Delta\mathbf{u}_{\xi,h}^*\|
  + \alpha_p\,\|\Delta\mathbf{p}_{\xi,h}\|}
  {\mathbb{E}_{\mathcal{B}}\!\bigl[\alpha_u\,\|\Delta\mathbf{u}^*\|
  + \alpha_p\,\|\Delta\mathbf{p}\|\bigr]}
  \label{eq:traj_weight}
\end{equation}
where $\Delta\mathbf{u}_{\xi,h}^* = \mathbf{u}_{\xi,h+1}^* -
\mathbf{u}_{\xi,h}^*$
is the Pareto-optimal action increment, $\Delta\mathbf{p}_{\xi,h}$ is the
concurrent disturbance change, and $(\alpha_u, \alpha_p)$ are
balancing coefficients.  We write the resulting weighted MSE as
$\mathcal{W}(\hat{\mathbf{u}}, \mathbf{u}^*, w) =
\mathbb{E}_{\mathcal{B}}[w_{\xi,h}\,\|\hat{\mathbf{u}}_{\xi,h} - \mathbf{u}_{\xi,h}^*\|^2]$.

\paragraph{Rollout supervision.}
The key objective exposes $D_u$ to latent codes along the simulated
rollout path.  Let
$\hat{z}_{\xi,h+1}^\alpha =
z_{\xi,h} + \alpha(z_{\xi,h+1}-z_{\xi,h})$,
$\alpha \in (0,1]$.  Two terms supervise $D_u$ along this path:
\begin{align}
  \mathcal{L}_{\text{roll}} &\;=\;
  \mathcal{W}\!\bigl(D_u(\hat{z}_{\xi,h+1}^\alpha),\;
  \mathbf{u}_{\xi,h+1}^*,\; w_{\xi,h}\bigr)
  \label{eq:du_roll}\\
  \mathcal{L}_{\text{flow}} &\;=\;
  \mathcal{W}\!\bigl(D_u(\hat{z}_{\xi,h+1}^{\alpha_f}),\;
  \mathbf{u}_{\xi,h}^*
  + \alpha_f\,\Delta\mathbf{u}_{\xi,h}^*,\; w_{\xi,h}\bigr)
  \label{eq:du_flow}
\end{align}
$\mathcal{L}_{\text{roll}}$ trains $D_u$ to predict the next Pareto
action at a rolled-out latent; $\mathcal{L}_{\text{flow}}$ additionally
anchors intermediate latent codes to linearly interpolated targets,
enforcing smooth action variation along the latent path.

\paragraph{Lie-corrected rollout latent.}
Linear interpolation ignores the Riemannian geometry of $\mathcal{M}^*$.
To match the distribution of latents produced by the online integrator
(Algorithm~\ref{alg:online}), we construct a geometry-aware rollout
latent using the same decoder-pullback metric
$G(z_{\xi,h}) = J_{D_s}(z_{\xi,h})^\top J_{D_s}(z_{\xi,h})
+ \lambda_{\mathrm{m}}\mathbf{I}$,
with eigendecomposition $G = U\Lambda U^\top$.  Retaining top-$k$
eigenvectors $B \in \mathbb{R}^{n_z \times k}$ and scales
$\mathbf{s} = \sqrt{\mathrm{diag}(\Lambda_{[-k:]})}$, the Euclidean
step is projected into the local metric frame, rotated via
$\mathfrak{so}(3)$ block-diagonal matrix exponentials, and projected
back:
\begin{align}
  \Delta z_{\text{Euc}} &\;=\;
    \hat{z}_{\xi,h+1}^\alpha - z_{\xi,h} \notag\\
  v_{\text{local}} &\;=\; B^\top \Delta z_{\text{Euc}} \oslash \mathbf{s} \notag\\
  \tilde{v}_{\text{local}} &\;=\;
    \mathrm{Rot}_{\mathfrak{so}(3)}\!\bigl(B^\top\Delta z_{\text{Euc}},\,
    v_{\text{local}};\,s_{\mathrm{L}},\Delta t\bigr) \notag\\
  \Delta z_{\text{Lie}} &\;=\; B\,(\tilde{v}_{\text{local}} -
    B^\top\Delta z_{\text{Euc}}) \notag\\
  \tilde{z}_{\xi,h+1} &\;=\; z_{\xi,h} + \Delta z_{\text{Euc}}
    + \gamma_{\mathrm{L}}\,\Delta z_{\text{Lie}}
  \label{eq:du_lie_latent}
\end{align}
where $\oslash$ denotes elementwise division, $s_{\mathrm{L}}$ is the Lie scale,
and $\gamma_{\mathrm{L}} \in [0,1]$ matches the correction strength in
eq.~\eqref{eq:final-update}.

\paragraph{Training objective.}
The full loss combines pointwise supervision, ramp-rate consistency,
rollout coverage, and Lie-consistent supervision at
$\tilde{z}_{\xi,h+1}$:
\begin{align}
  \mathcal{L}_{D_u}
  &=
  \underbrace{\mathcal{L}_{\text{act},t}
    + \mathcal{L}_{\text{act},t+1}
    + \lambda_\delta\,\mathcal{L}_\delta}_{\text{pointwise}}
  +
  \underbrace{\lambda_{\text{roll}}\,\mathcal{L}_{\text{roll}}
    + \lambda_{\text{flow}}\,\mathcal{L}_{\text{flow}}}_{\text{rollout}}
  \notag\\
  &\quad
  + \lambda_{\text{Lie}}\!\left(
    \mathcal{L}_{\text{Lie-act}}
    + \lambda_{\delta}\,\mathcal{L}_{\text{Lie-}\delta}
    + \mathcal{L}_{\text{Lie-anch}}
    + \mathcal{L}_{\text{Lie-reg}}
    \right).
  \label{eq:du_obj}
\end{align}
where
$\mathcal{L}_{\text{act},t} =
\|D_u(z_{\xi,h})-\mathbf{u}_{\xi,h}^*\|^2$,
$\mathcal{L}_{\text{act},t+1} =
\|D_u(z_{\xi,h+1})-\mathbf{u}_{\xi,h+1}^*\|^2$
provide pointwise supervision at both endpoints;
$\mathcal{L}_\delta =
\|(D_u(z_{\xi,h+1})-D_u(z_{\xi,h}))
-\Delta\mathbf{u}_{\xi,h}^*\|^2$
enforces ramp-rate consistency;
$\mathcal{L}_{\text{Lie-act}} =
\mathcal{W}(D_u(\tilde{z}_{\xi,h+1}),
  \mathbf{u}_{\xi,h+1}^*,w_{\xi,h})$
and
$\mathcal{L}_{\text{Lie-}\delta} =
\mathcal{W}(D_u(\tilde{z}_{\xi,h+1})-D_u(z_{\xi,h}),
  \Delta\mathbf{u}_{\xi,h}^*,w_{\xi,h})$
supervise $D_u$ at the Lie-corrected latent;
$\mathcal{L}_{\text{Lie-anch}} =
\mathcal{W}(\tilde{z}_{\xi,h+1},
  \hat{z}_{\xi,h+1}^{\alpha\,\text{sg}},w_{\xi,h})$
prevents the Lie correction from drifting from the rollout path; and
$\mathcal{L}_{\text{Lie-reg}} = \mathbb{E}_{\mathcal{B}}[w_{\xi,h}\,\|\Delta z_{\text{Lie},\xi,h}\|^2]$
regularizes the magnitude of the curvature correction.

\begin{remark}[Rollout-Latent Coverage for the Closed-Loop Controller]
\label{rem:closed_loop}
Static supervision trains $D_u$ on the offline latent marginal
$\rho_{\mathrm{off}}$, whereas online deployment evaluates it on
the latent marginal $\rho_{\mathrm{on}}$ induced by
Algorithm~\ref{alg:online}.  The rollout losses do not make the
offline procedure fully closed-loop: training uses horizon-indexed
oracle transitions and synthetic rollout latents, rather than
recursively applying the learned controller to generate the next
physical state.  Their role is instead to reduce the distributional gap
between the static training support and the latents encountered by the
online closed-loop navigator.  The one-step quantities below refer to a
single control-cycle transition inside this horizon-indexed construction,
not to a restriction of the navigation problem to a one-step horizon.

Let
$e(z)=\|D_u(z)-\mathbf{u}^*(z)\|$ denote the action error and assume
that $e$ is $L_e$-Lipschitz on the relevant compact latent domain
($L_e \le K_{D_u}+K_{u^*}$ if the oracle action map
$\mathbf{u}^*$ is $K_{u^*}$-Lipschitz).  Then, for the 1-Wasserstein
distance $W_1$,
\begin{equation}
  \mathbb{E}_{z \sim \rho_{\mathrm{on}}}
  \bigl[e(z)\bigr]
  \;\leq\;
  \underbrace{\mathbb{E}_{z \sim \rho_{\mathrm{off}}}
  \bigl[e(z)\bigr]}_{\text{minimized by }\mathcal{L}_{\text{act}}}
  \;+\;
  L_e\,W_1(\rho_{\mathrm{on}},\rho_{\mathrm{off}}).
  \label{eq:dist_shift}
\end{equation}
The second term is the deployment-shift penalty left by purely static
training.

The Lie-corrected rollout objective replaces $\rho_{\mathrm{off}}$
with the auxiliary rollout-latent distribution $\rho_{\mathrm{Lie}}$
generated by $\tilde{z}_{\xi,h+1}$ in~\eqref{eq:du_lie_latent}.  If
this offline simulator is coupled to one control-cycle update of
Algorithm~\ref{alg:online} so that
$\mathbb{E}\|\tilde{z}_{\xi,h+1}-z_{\xi,h+1}^{\mathrm{on}}\|
\le \varepsilon_L$, then
$W_1(\rho_{\mathrm{Lie}},\rho_{\mathrm{on}})\le \varepsilon_L$ and
\begin{equation}
  \mathbb{E}\bigl[\|D_u(z_{\xi,h+1}^{\mathrm{on}}) -
  \mathbf{u}_{\xi,h+1}^*\|\bigr]
  \;\leq\;
  \underbrace{\mathbb{E}\bigl[\|D_u(\tilde{z}_{\xi,h+1}) -
  \mathbf{u}_{\xi,h+1}^*\|\bigr]}_{\text{minimized by }\mathcal{L}_{\text{Lie-act}}}
  \;+\; K_{D_u}\,\varepsilon_L.
  \label{eq:lie_rollout_bound}
\end{equation}
(The coefficient is $K_{D_u}$, not $L_e$: since the target
$\mathbf{u}_{\xi,h+1}^*$ is
fixed by the training tuple, the gap comes only from the Lipschitz continuity
of $D_u$ via the triangle inequality, without involving the oracle map
$\mathbf{u}^*$.)
Thus the method is \emph{closed-loop aware}, but not closed-loop
equivalent: it reduces the mismatch between the supervised latent
locations and the post-integration rollout latents encountered by
the online controller.

This gap should not be confused with an accumulating closed-loop
tracking error.  Proposition~\ref{prop:closed-loop} already gives the
actual online controller the uniform bound
\begin{equation}
  \sup_{t\ge 0}\|\mathbf{u}_t-\mathbf{u}_t^*\|
  \;\le\;
  K_{D_u}\!\left(\delta_{\mathrm{loc}}
    + \tfrac{1}{2}\kappa_{\max}V_{\max}^2\Delta t^2\right)
  + \delta_{\mathrm{dec}},
  \label{eq:du_training_closed_loop_bound}
\end{equation}
which is independent of the rollout horizon because Step~1
re-anchors the latent state at every cycle.  Under the same
per-cycle re-anchoring, if the one-step Lie surrogate satisfies
$\sup_t W_1(\rho_{\mathrm{Lie},t},\rho_{\mathrm{on},t})
\le\varepsilon_L$, then the training-to-deployment supervision gap is
also uniform:
\begin{equation}
  \sup_{t\ge0}\left|
  \mathbb{E}_{z\sim\rho_{\mathrm{on},t}} e(z)
  - \mathbb{E}_{z\sim\rho_{\mathrm{Lie},t}} e(z)
  \right|
  \;\le\; L_e\,\varepsilon_L .
  \label{eq:closed_loop_gap}
\end{equation}
Hence rollout training does not introduce a separate horizon-growing
closed-loop error term; it controls how well the offline
supervision points cover the rollout latents produced by the already
non-accumulating closed-loop algorithm.
\end{remark}

\section{Supplementary Benchmark Definitions}
\label{sec:validation}

\subsection{Analytical Dynamic Navigation with Convex Objectives}
\label{sec:analytical-navigation}

Before instantiating GPC on the full multi-objective optimal 
power flow problem of Section~\ref{sec:opf}, we define analytical 
navigation benchmarks where the Pareto manifold, semantic coordinate 
evolution, and latent trajectories can be visualized directly. These 
controlled settings isolate the mechanisms of 
Algorithm~\ref{alg:online}; the corresponding simulation results and 
ablations are reported in Section~\ref{sec:experiments}.

\textbf{Problem setup.}
We use a constrained dynamic navigation problem that has the same 
structure as a standard safe-RL benchmark but remains analytically 
visualizable. The physical state is 
$\mathbf{s}_t=(\mathbf{q}_t,\mathbf{v}_t)\in\mathbb{R}^4$, where 
$\mathbf{q}_t\in\mathbb{R}^2$ is position and 
$\mathbf{v}_t\in\mathbb{R}^2$ is velocity. The control action 
$\mathbf{u}_t\in\mathbb{R}^2$ is acceleration, and the system obeys 
double-integrator dynamics:
\begin{align}
  \mathbf{q}_{t+1} &= \mathbf{q}_t + \Delta t\,\mathbf{v}_t
  + \tfrac{1}{2}\Delta t^2\,\mathbf{u}_t, \\
  \mathbf{v}_{t+1} &= \mathbf{v}_t + \Delta t\,\mathbf{u}_t .
  \label{eq:analytical-dynamics}
\end{align}
Thus the current decision cannot instantaneously choose a point on 
the Pareto front; it must act through the system dynamics, as in an 
RL control environment.

The two objectives are evaluated on the dynamically propagated next 
state:
\begin{align}
  J_1(\mathbf{s}_t,\mathbf{u}_t,p_t)
  &= \|\mathbf{q}_{t+1}-\mathbf{q}_{\mathrm{safe}}(p_t)\|^2,\\
  J_2(\mathbf{s}_t,\mathbf{u}_t)
  &= \|\mathbf{q}_{t+1}-\mathbf{q}_{\mathrm{goal}}\|^2
  + \beta_{\mathrm{u}}\|\mathbf{u}_t\|^2.
  \label{eq:analytical-objectives}
\end{align}
Here $J_1$ is the safety objective, pulling the system toward a 
load-dependent recovery location $\mathbf{q}_{\mathrm{safe}}(p_t)$, 
while $J_2$ is the performance objective, pulling the system toward 
the mission goal with small control effort.

The scalar operating condition $p_t$ shifts a moving elliptical 
keep-out region, analogous to a load-driven security boundary in OPF. 
Let $\mathbf{c}(p_t)=\mathbf{c}_0+\rho p_t$ be its center. The 
nonconvex state constraint requires the next position to remain 
outside the ellipse:
\begin{equation}
  g_{\mathrm{obs}}(\mathbf{q}_{t+1},p_t)
  =
  1 -
  \frac{(q_{1,t+1}-c_1(p_t))^2}{a^2}
  -
  \frac{(q_{2,t+1}-c_2(p_t))^2}{b^2}
  \le 0.
  \label{eq:analytical-obstacle}
\end{equation}
The dynamic constraints are the input bound and a ramp-rate-like slew 
constraint:
\begin{align}
  \|\mathbf{u}_t\| &\le u_{\max}, \qquad
  \|\mathbf{u}_t-\mathbf{u}_{t-1}\| \le r_{\max}.
  \label{eq:analytical-dynamic-constraints}
\end{align}
The moving obstacle creates a curved, nonconvex, time-varying 
feasible region, while the slew constraint introduces temporal 
coupling between consecutive actions. This is the essential dynamic 
constraint structure that model-free constrained RL can represent, 
but GPC solves by navigating the precomputed Pareto manifold rather 
than by learning a policy through rollout.

For each context 
$\xi_t=(\mathbf{s}_t,\mathbf{u}_{t-1},p_t)$, the dynamically feasible 
Pareto set is obtained by sweeping $\mathbf{w}\in\Delta^1$ in:
\begin{equation}
  \mathbf{u}^*(\mathbf{w};\xi_t)
  =
  \arg\min_{\mathbf{u}_t}
  \; w_1J_1(\mathbf{s}_t,\mathbf{u}_t,p_t)
  + w_2J_2(\mathbf{s}_t,\mathbf{u}_t)
  \quad
  \text{s.t.}\quad
  \eqref{eq:analytical-dynamics}-\eqref{eq:analytical-dynamic-constraints}.
  \label{eq:analytical-dynamic-pareto}
\end{equation}
These solutions define a low-dimensional Pareto manifold 
$\mathcal{M}^*$ over dynamic operating contexts rather than a static 
curve in action space. During nominal operation, 
$\boldsymbol{\sigma}_t$ is performance dominated and GPC slides along 
the dynamically feasible arc toward $\mathbf{q}_{\mathrm{goal}}$. 
During a constraint event, the moving ellipse intersects the predicted 
trajectory; the normal potential first snaps the latent state back to 
the dynamically feasible manifold, and the tangential potential then 
slides along the manifold toward the safety-dominated region. A 
fixed-weight baseline is trapped by its constant scalarization and 
either reacts late or violates the obstacle or slew constraint.

\subsection{Analytical Dynamic Navigation with Nonconvex Objective}
\label{sec:analytical-nonconvex}

The convex-objective benchmark above tests dynamic constraints and 
semantic priority adaptation while keeping the scalarized objective 
landscape simple. To test navigation over a nonconvex objective 
landscape, we keep the same double-integrator dynamics, moving 
obstacle, input bound, and slew constraint, but replace the 
performance objective by a smooth multi-basin objective:
\begin{equation}
  J_2^{\mathrm{nc}}(\mathbf{s}_t,\mathbf{u}_t)
  =
  \|\mathbf{q}_{t+1}-\mathbf{q}_{\mathrm{goal}}\|^2
  + \beta_{\mathrm{u}}\|\mathbf{u}_t\|^2
  + \lambda_{\mathrm{nc}}\sum_{\ell=1}^{2}
  \bigl[1-\cos\bigl(\omega(q_{\ell,t+1}
  -q_{\mathrm{goal},\ell})\bigr)\bigr].
  \label{eq:analytical-nonconvex-objective}
\end{equation}
The sinusoidal term introduces local basins around the nominal goal 
while preserving differentiability. The corresponding Pareto set is:
\begin{equation}
  \mathbf{u}_{\mathrm{nc}}^*(\mathbf{w};\xi_t)
  =
  \arg\min_{\mathbf{u}_t}
  \; w_1J_1(\mathbf{s}_t,\mathbf{u}_t,p_t)
  + w_2J_2^{\mathrm{nc}}(\mathbf{s}_t,\mathbf{u}_t)
  \quad
  \text{s.t.}\quad
  \eqref{eq:analytical-dynamics}-\eqref{eq:analytical-dynamic-constraints}.
  \label{eq:analytical-nonconvex-pareto}
\end{equation}
This benchmark is still low-dimensional enough to visualize, but it 
requires the learned Pareto map to encode a genuinely nonconvex 
trade-off surface. In a constrained RL formulation the same problem 
can be represented by a reward and safety cost; in GPC the nonconvex 
Pareto family is solved offline and then traversed online by the 
same geometric navigator.

\subsection{Safe Multi-Agent Navigation with Moving Obstacles}
\label{sec:multiagent-navigation}

The analytical benchmarks above isolate the geometry in a two-dimensional
state space. To test whether the same construction scales to a more
demanding robotics setting, we also instantiate GPC on fixed-horizon
safe multi-agent navigation. The controlled system contains one ego
robot and up to five other disc robots in a bounded planar workspace.
Each robot has state $x_t^i=(p_t^i,\theta_t^i,v_t^i)$, control
$u_t^i=(a_t^i,\omega_t^i)$, radius $r_i$, and a goal $g^i$. The scene
contains up to 15 circular obstacles. In the static benchmark the
obstacles are fixed. In the dynamic benchmark obstacle $j$ has state
$o_t^j=(c_t^j,r_j,v_{o,t}^j)$ and moves either in a random direction or
toward the ego-goal corridor.

At decision time $t$, the white-box optimizer and the learned GPC
controller solve the same receding-horizon problem over the joint plan
$\mathbf{U}_{t:t+H-1}=\{u_{t+k}^i:k=0,\ldots,H-1,\ i=0,\ldots,N\}$.
Each active robot follows the discrete unicycle dynamics
\begin{equation}
\begin{aligned}
  \theta_{t+k+1}^i &= \theta_{t+k}^i+\Delta t\,\omega_{t+k}^i,\\
  v_{t+k+1}^i &=
  \mathrm{clip}\!\left(v_{t+k}^i+\Delta t\,a_{t+k}^i,\,
  v_{\min},v_{\max}\right),\\
  p_{t+k+1}^i &= p_{t+k}^i
  +\Delta t\,\bar v_{t+k}^i
  \begin{bmatrix}
  \cos \bar\theta_{t+k}^i\\
  \sin \bar\theta_{t+k}^i
  \end{bmatrix},
\end{aligned}
\label{eq:multiagent-unicycle}
\end{equation}
where $\bar \theta_{t+k}^i=\theta_{t+k}^i+\frac{1}{2}\Delta
t\,\omega_{t+k}^i$ and
$\bar v_{t+k}^i=(v_{t+k}^i+v_{t+k+1}^i)/2$. Controls satisfy
\begin{equation}
  a_{\min}\le a_{t+k}^i\le a_{\max},\qquad
  |\omega_{t+k}^i|\le \omega_{\max},\qquad
  \|u_{t+k}^i-u_{t+k-1}^i\|_{\infty}\le \Delta u_{\max}.
  \label{eq:multiagent-control-limits}
\end{equation}
Moving obstacles are predicted by
$c_{t+k+1}^j=c_{t+k}^j+\Delta t\,v_{o,t+k}^j$ with boundary reflection;
the static benchmark is the special case $v_{o,t+k}^j=0$.

The hard safety set is defined by all active robot-obstacle and
robot-robot margins along the prediction horizon:
\begin{equation}
\begin{aligned}
  h_{i,j}^{\mathrm{obs}}(t+k)
  &=\|p_{t+k}^i-c_{t+k}^j\|_2-r_i-r_j-d_{\mathrm{safe}}\ge 0,\\
  h_{i,\ell}^{\mathrm{rr}}(t+k)
  &=\|p_{t+k}^i-p_{t+k}^{\ell}\|_2-r_i-r_{\ell}
    -d_{\mathrm{safe}}\ge 0 .
\end{aligned}
\label{eq:multiagent-safety-margins}
\end{equation}
Let $\mathcal{H}_{t+k}$ collect these margins and let
$m_i\in\{0,1\}$ denote the active-agent mask, with $m_0=1$ for the ego
robot. The six-objective vector optimized offline is
\begin{equation}
\resizebox{0.98\textwidth}{!}{$
\begin{gathered}
J_{\mathrm{goal}}=\sum_{k=1}^{H}\!\left(
\|p_{t+k}^{0}-g^0\|_2^2+
\eta_g\frac{\sum_{i=1}^{N}m_i\|p_{t+k}^{i}-g^i\|_2^2}{\max(\sum_i m_i,1)}
\right),\\[-0.2em]
\begin{alignedat}{2}
J_{\mathrm{safe}}&=\sum_{k=1}^{H}\operatorname{mean}_{h\in\mathcal{H}_{t+k}}
e^{-\operatorname{clip}(h,-3,3)},\qquad&
J_{\mathrm{viol}}&=\sum_{k=1}^{H}
\sum_{h\in\mathcal{H}_{t+k}}\!\left[\max(0,-h)\right]^2,\\
J_{\mathrm{effort}}&=\sum_{k=0}^{H-1}
\frac{\sum_{i=0}^{N}m_i\|u_{t+k}^i\|_2^2}
{\max(1+\sum_{i=1}^{N}m_i,1)},&
J_{\mathrm{term}}&=\|p_{t+H}^{0}-g^0\|_2^2+
\eta_g\frac{\sum_{i=1}^{N}m_i\|p_{t+H}^{i}-g^i\|_2^2}{\max(\sum_i m_i,1)},\\
J_{\mathrm{smooth}}&=\frac{1}{H}\sum_{k=0}^{H-1}
\sum_{i=0}^{N}\|u_{t+k}^i-u_{t+k-1}^i\|_2^2 .
\end{alignedat}
\end{gathered}
$}
\label{eq:multiagent-objectives}
\end{equation}
In the experiments, $H=8$, $\Delta t=0.2$, $N\le 5$,
$d_{\mathrm{safe}}=0.18$, and $\eta_g=0.65$. The preference vector
$w\in\Delta^2$ controls goal progress, soft safety, and effort, while
hard violation, terminal error, and smoothness remain always active. The implementation uses
$(\lambda_g,\lambda_s,\lambda_e,\lambda_v,\lambda_T,\lambda_{\Delta})
=(2.0,12.0,0.08,260.0,5.0,0.10)$. The white-box solver samples
Pareto-optimal joint plans for \eqref{eq:multiagent-objectives};
GPC then learns the solution manifold over the current scene
$\xi_t$ and preference $w$, and online executes only the first joint
action before replanning.

A rollout is counted as \emph{scene collision-free} only if every
active obstacle and robot-robot margin remains nonnegative over the
closed loop. For the dynamic benchmark we report \emph{joint success},
which requires both the ego robot and all active agents to reach their
goals within the goal radius, and \emph{safe success}, which is joint
success multiplied by the scene collision-free indicator.

The offline reference optimizes this objective vector with the same predicted
obstacle motion and pairwise safety constraints.
It is used as a reference optimizer, not as a queried policy at deployment. This
distinction matters in the moving-obstacle experiment: the reference
often fails to find a safe closed-loop solution within its finite
optimization budget. We therefore report both reference-safe success and
the optimality gap on the subset of episodes where the reference itself
is safe:
\begin{equation}
  \mathrm{Gap}_{\mathrm{ref-safe}}(\%)
  =
  100\,
  \frac{\mathbb{E}[J_{\mathrm{GPC}}-J_{\mathrm{ref}}
  \mid \mathrm{ref\ safe}]}
  {\mathbb{E}[J_{\mathrm{ref}}\mid \mathrm{ref\ safe}]} .
  \label{eq:multiagent-ref-safe-gap}
\end{equation}
This metric separates near-optimality from reference brittleness in
dynamic safety-critical scenes.

\subsection{Multi-Objective Optimal Power Flow Control with 
Multi-Timescale Dynamics}
\label{sec:opf}

To make the GPC framework concrete, we instantiate it on 
multi-objective AC optimal power flow. The framework is not 
specific to power systems: it is applicable to settings where the
governing physics is available offline and the control problem can be
written in the form of~\eqref{eq:opf-general}. Wireless resource 
allocation, for instance, has related interference-coupling equations
and rate-quality constraints \cite{luong2019applications, feriani2021single}. 
Transportation network control involves flow conservation equations on
sparse graphs with capacity and ramp-rate constraints
\cite{haydari2020deep, zong2025deep}. Building energy management,
water distribution, and chemical process control also contain
nonlinear equality constraints, inequality safety constraints, and
temporal coupling, although each domain would require its own
instantiation and validation.

AC-OPF is chosen as the 
primary instantiation because it is a challenging and practically
important instance of this class: the feasibility set $\mathbf{g}(\mathbf{s}, 
\mathbf{u}) = \mathbf{0}$ is dense, non-separable, and nonconvex 
due to trigonometric voltage phasor interactions, and the three 
objectives span timescales differing by four orders of magnitude. 
Performance on AC-OPF therefore provides evidence for the value of the
geometric construction in a demanding CPS benchmark, while not by
itself establishing performance across all CPS domains. 
For readers unfamiliar with power systems, the problem takes the 
following general mathematical form.
Given a network of $n_{\text{bus}}$ nodes and $n_{\text{gen}}$ 
controllable inputs, find $\mathbf{u} \in \mathbb{R}^{n_u}$ solving:
\begin{equation}
  \min_{\mathbf{u}} \quad \mathbf{w}^\top \mathbf{J}(\mathbf{s}, 
  \mathbf{u}) \qquad
  \text{subject to} \quad
  \begin{cases}
    \mathbf{g}(\mathbf{s}, \mathbf{u}) = \mathbf{0} 
    & \text{(nonlinear equality)}\\
    \mathbf{f}(\mathbf{s}, \mathbf{u}) \le \mathbf{0} 
    & \text{(inequality constraints)}\\
    \mathbf{c}(\mathbf{u}_t, \mathbf{u}_{t-1}) \le \mathbf{0} 
    & \text{(temporal coupling)}
  \end{cases}
  \label{eq:opf-general}
\end{equation}
where $g$ is a set of dense, coupled nonlinear equalities 
(the AC power flow equations) that define the feasibility manifold 
in state space, $\mathbf{f}$ collects box and network safety 
constraints, and $\mathbf{c}$ encodes inter-step temporal coupling. 
The objective vector $\mathbf{J} \in \mathbb{R}^3$ has three 
components ordered by timescale and safety priority: $J_1$ 
(safety-critical, seconds), $J_2$ (quality-of-service, minutes), 
and $J_3$ (economic, tens of minutes to hours). The weight vector 
$\mathbf{w} \in \Delta^2$ is not fixed; in GPC it is generated 
from real-time observations through the 
semantic coordinate $\boldsymbol{\sigma}_t$ derived in 
Section~\ref{sec:offline}.

Problem~\eqref{eq:opf-general} is solved repeatedly at every 
dispatch interval $\Delta t$ under continuously varying exogenous 
disturbance $\mathbf{p}_t$, with a hard real-time deadline that 
can make iterative nonlinear solvers impractical in larger or faster
settings. The remainder of 
this section specifies the concrete instantiation of each GPC 
component for this problem class.

\subsubsection{System Model and Multi-Objective Formulation}

Consider a transmission network with $n_{\text{bus}}$ buses, 
$n_{\text{gen}}$ generators, and branch set $\mathcal{L}$. Let 
$V_i \in \mathbb{R}_{>0}$ and $\theta_i \in \mathbb{R}$ denote the 
voltage magnitude and phase angle at bus $i$, and let 
$Y_{\text{bus}} = G + jB \in \mathbb{C}^{n_{\text{bus}} \times 
n_{\text{bus}}}$ denote the network admittance matrix. The system 
state is the voltage phasor vector 
$\mathbf{s}_t = [\mathbf{V}_t,\,\boldsymbol{\theta}_t] \in 
\mathbb{R}^{2n_{\text{bus}}}$, the control action is the generator 
dispatch $\mathbf{u}_t = [P_{\text{gen}},\,Q_{\text{gen}}] \in 
\mathbb{R}^{2n_{\text{gen}}}$, and the exogenous disturbance is the 
load demand $\mathbf{p}_t = [P_{\text{load}},\,Q_{\text{load}}]$, 
which varies continuously and is not controllable. The OPF experiments use the full measured system observation:
\begin{equation}
  \mathbf{x}_t = h(\mathbf{s}_t) + \boldsymbol{\eta}_t 
  \in \mathbb{R}^{n_x}
\end{equation}
where $h$ is the identity map up to measurement noise, so
$\mathbf{x}_t=\mathbf{s}_t+\boldsymbol{\eta}_t$ and
$n_x=2n_{\text{bus}}$. The offline dataset takes the base form 
of Section~\ref{sec:offline}:
\begin{equation}
  \mathcal{D} = \bigl\{(\mathbf{x}_k,\,\mathbf{u}_k^*,\,
  \mathbf{s}_k,\,\boldsymbol{\sigma}_k,\,\mathbf{J}_k)
  \bigr\}_{k=1}^N
\end{equation}
where $\boldsymbol{\sigma}_k \in \Delta^2$ is the three-component 
semantic priority coordinate over thermal, voltage, and economic 
objectives, and $\mathbf{J}_k = [J_1(\mathbf{s}_k,\mathbf{u}_k^*),\,
J_2(\mathbf{s}_k,\mathbf{u}_k^*),\,J_3(\mathbf{u}_k^*)]^\top$ is the 
objective vector.

\textbf{AC power flow equations.} The physical state must satisfy 
the AC power balance equations at every bus $i$:
\begin{align}
  P_{\text{gen},i} - P_{\text{load},i} &= V_i 
  \sum_{j=1}^{n_{\text{bus}}} V_j \bigl[G_{ij}
  \cos(\theta_i - \theta_j) + B_{ij}\sin(\theta_i - \theta_j)\bigr] 
  \label{eq:pflow}\\
  Q_{\text{gen},i} - Q_{\text{load},i} &= V_i 
  \sum_{j=1}^{n_{\text{bus}}} V_j \bigl[G_{ij}
  \sin(\theta_i - \theta_j) - B_{ij}\cos(\theta_i - \theta_j)\bigr]
  \label{eq:qflow}
\end{align}
These nonlinear equality constraints define the AC feasibility manifold 
in state space. In GPC, they are not enforced explicitly at runtime; 
instead, every sample in $\mathcal{D}$ is generated by solving the 
full AC-OPF to feasibility, so the power flow equations are embedded 
implicitly in the geometry of $\mathcal{M}^*$. At runtime, the state 
decoder $D_s(z_t)$ reconstructs a power-flow-consistent state 
$\hat{\mathbf{s}}_t = [\hat{\mathbf{V}}_t,\,\hat{\boldsymbol{\theta}}_t]$ 
directly from the latent state, and all violation indicators are 
evaluated on $\hat{\mathbf{s}}_t$.

\textbf{Branch flows.} The apparent power flow on branch $(i,j)$ is 
$S_{ij} = \sqrt{P_{ij}^2 + Q_{ij}^2}$, where:
\begin{align}
  P_{ij} &= V_i^2 G_{ij} - V_i V_j \bigl[G_{ij}\cos(\theta_i
  -\theta_j) + B_{ij}\sin(\theta_i-\theta_j)\bigr] \\
  Q_{ij} &= -V_i^2 B_{ij} - V_i V_j \bigl[G_{ij}\sin(\theta_i
  -\theta_j) - B_{ij}\cos(\theta_i-\theta_j)\bigr]
\end{align}
Branch flows are not directly used as learned inputs. They are
reconstructed from the decoded voltage phasors 
$(\hat{V}_i,\,\hat{\theta}_i)$ extracted from $\hat{\mathbf{s}}_t = 
D_s(z_t)$ via the known $Y_{\text{bus}}$, eliminating the need for 
branch-flow supervision in the observation vector.

\textbf{Inequality constraints.} The feasible action set is further 
restricted by:
\begin{alignat}{2}
  &P_{\text{gen},i}^{\min} \le P_{\text{gen},i} \le 
  P_{\text{gen},i}^{\max}, \quad 
  &&Q_{\text{gen},i}^{\min} \le Q_{\text{gen},i} \le 
  Q_{\text{gen},i}^{\max} 
  \quad\text{(generator capacity)}\\
  &V_i^{\min} \le V_i \le V_i^{\max} 
  &&\text{(voltage limits)}\\
  &S_{ij} \le S_{\max,ij} 
  &&\text{(branch thermal limits)}\\
  &|P_{\text{gen},i,t} - P_{\text{gen},i,t-1}| \le \delta_{\text{ramp}} 
  &&\text{(ramp-rate limits)}
\end{alignat}
Ramp-rate limits couple consecutive dispatch decisions and constitute 
the primary source of multi-timescale complexity: they prevent the 
fast thermal objective from being satisfied instantaneously, requiring 
the navigator to plan a feasible trajectory across multiple steps. 
Inequality constraints are handled by the geometry encoded in 
$\mathcal{M}^*$, the output clamp for box constraints, and, when 
hard per-step coupling enforcement is enabled, a local Jacobian 
projection via the action Jacobian $\mathbf{J}_t$, with each active 
constraint type contributing rows to $A_t$ and $b_t$.

\textbf{Three-objective hierarchy.} The MO-OPF instantiates the 
general objective vector $\mathbf{J}(\mathbf{s}_t, \mathbf{u}_t)$ 
with three terms ordered by timescale and safety priority.

$J_1$ penalizes thermal overloading via a quartic potential that 
rises explosively as branch loading approaches the thermal rating:
\begin{equation}
  J_1(\mathbf{s}_t) = \sum_{(i,j) \in \mathcal{L}} \left[\max\!\left(0, 
  \frac{S_{ij}}{S_{ij}^{\max}} - \tau\right)\right]^4, \qquad 
  \tau = 0.85
\end{equation}
The quartic exponent and the $\tau = 0.85$ margin create a steep 
potential wall that activates before the hard thermal limit is 
reached, consistent with the singular perturbation structure of the 
semantic potentials in Section~\ref{sec:offline:semantic}.

$J_2$ penalizes voltage deviations outside an acceptable deadband 
around the nominal value of $1.0$ p.u.:
\begin{equation}
  J_2(\mathbf{s}_t) = \sum_{i=1}^{n_{\text{bus}}} \max\!\left(0,\, 
  |V_i - 1.0| - \delta_{\text{dead}}\right)^2, \qquad 
  \delta_{\text{dead}} = 0.05\;\text{p.u.}
\end{equation}
The deadband avoids penalizing small deviations within the acceptable 
operating range, concentrating the gradient signal on genuine voltage 
quality violations.

$J_3$ models the economic cost of generation via the standard 
quadratic heat rate characteristic:
\begin{equation}
  J_3(\mathbf{u}_t) = \sum_{i=1}^{n_{\text{gen}}} \left[a_i 
  P_{\text{gen},i}^2 + b_i P_{\text{gen},i} + c_i\right]
\end{equation}
where $a_i$, $b_i$, $c_i$ are generator-specific cost coefficients 
determined from historical fuel cost data.

\subsubsection{GPC Instantiation: Semantic Violation Indicators}

The general violation indicators $\delta_i(\hat{\mathbf{s}}_t) 
\in [0,1]$ of Section~\ref{sec:offline:semantic} are instantiated 
on the reconstructed state $\hat{\mathbf{s}}_t = D_s(z_t)$ as 
follows. The thermal indicator measures the worst-case branch 
loading across the network:
\begin{equation}
  \delta_f(\hat{\mathbf{s}}_t) = \mathrm{clip}\!\left(
  \max_{(i,j)\in\mathcal{L}} 
  \frac{\hat{S}_{ij}}{S_{ij}^{\max}},\;0,\;1\right)
\end{equation}
The voltage indicator measures the worst-case normalized deviation 
from nominal across all buses:
\begin{equation}
  \delta_v(\hat{\mathbf{s}}_t) = \mathrm{clip}\!\left(
  \frac{\max_i|\hat{V}_i - V_{\text{nom}}|}
  {\Delta V_{\text{limit}}},\;0,\;1\right)
\end{equation}
The economic indicator $\delta_e$ is set to a small positive constant 
$\delta_e \ll 1$, encoding a persistent background drive toward 
economic operation that remains active whenever thermal and voltage 
conditions are benign. Together, the three indicators 
$(\delta_f, \delta_v, \delta_e)$ feed the autonomous semantic 
generation of Section~\ref{sec:online:semantic} to produce the 
urgency potentials:
\begin{equation}
  \phi_i(\hat{\mathbf{s}}_t) = \exp\!\left(\frac{k_i \cdot 
  \delta_i(\hat{\mathbf{s}}_t)}{\epsilon_i}\right) - 1
\end{equation}
and the semantic priority coordinate
$\sigma_{t,i} = (\phi_i(\hat{\mathbf{s}}_t)+\rho)/\sum_j(\phi_j(\hat{\mathbf{s}}_t)+\rho)
\in \Delta^2$ (equation~\eqref{eq:semantic-coord}), which shifts
continuously from $\sigma_{t,3} \approx 1$ (economic dominance) to
$\sigma_{t,1} \approx 1$ (thermal emergency) without any
rule-based switching.

\section{Supplementary Simulation Details and Results}
\label{sec:experiments}
\begin{table}[ht]
\centering
\caption{Simulation cases.}
\label{tab:setup}
\small
\begin{tabular}{p{0.15\textwidth}p{0.21\textwidth}p{0.29\textwidth}p{0.19\textwidth}}
\toprule
\textbf{Case} & \textbf{System} & \textbf{Constraints} &
\textbf{Purpose}\\
\midrule
Case~1: Analytical & 2D double integrator & Obstacle, input, slew & Mechanism visualization\\
Case~2: Multi-agent nav. & Five-agent robot navigation & Robot-obstacle, pairwise robot safety, moving obstacles & Reliability under dynamic safety coupling\\
Case~3: Nominal OPF & IEEE 30-bus & Thermal, voltage, ramp & OPF baseline comparison\\
Case~4: Uncertain OPF & IEEE 30-bus & Same + branch admittance noise & Robustness to network uncertainty\\
\bottomrule
\end{tabular}
\end{table}

\subsection{Experimental Setup}

We evaluate GPC on three simulation groups: analytical dynamic
navigation, safe multi-agent robot navigation, and IEEE 30-bus OPF.
The analytical simulations include both the
convex-objective benchmark of Section~\ref{sec:analytical-navigation}
and the nonconvex-objective benchmark of
Section~\ref{sec:analytical-nonconvex}. These cases isolate the
geometry and integration mechanisms in a fully visualizable setting.
The multi-agent navigation simulations of
Section~\ref{sec:multiagent-navigation} test the same online navigator
in higher-dimensional static- and moving-obstacle scenes with
ego-agent and agent-agent safety coupling.
The OPF case tests the same mechanisms on the standard MATPOWER
30-bus network under thermal/voltage/economic objectives. We additionally run the training-time system
uncertainty simulations of Section~\ref{sec:exp:uncertainty}; these
stress tests perturb the physical parameters of the same network
rather than introducing additional larger-network cases.

All ground-truth AC-OPF solutions are generated via IPOPT through 
MATPOWER with convergence tolerance $10^{-6}$. For the IEEE 30-bus 
case, the offline dataset $\mathcal{D}$ is constructed via Latin 
Hypercube Sampling with $N = 5000$ load scenarios, each solved for 
$|\mathcal{W}| = 50$ uniformly spaced weight vectors across 
$\Delta^2$, yielding $250{,}000$ certified Pareto-optimal samples. 
Load profiles span $[0.6,\,1.4]$ times the nominal MATPOWER loading.

For the uncertainty-aware training (Case~4), the physical
parameter vector $\theta$ is sampled jointly with load and objective
weights during offline data generation.
Case~4 restricts perturbations to branch admittances only
($\xi^S=\xi^c=0$):
\begin{equation}
  Y_{ij}(\theta) = Y_{ij}^{0}(1+\xi^{Y}_{ij}),
  \qquad \xi^{Y}_{ij}\in[-\rho_Y,\rho_Y],
  \label{eq:uncertain-training-parameters}
\end{equation}
with $(\sigma_Y,\rho_Y)=(0.01,0.03)$.
Each trajectory samples a single admittance realization held fixed
for all $H+1$ dispatch steps, reflecting the physical timescale
separation between line-parameter drift (hours) and OPF dispatch
(seconds).
The sampled $\theta$ is appended to the physical input $\mathbf{x}_k$
so the offline map learns a family of Pareto manifolds indexed by
uncertain network realizations. Training samples that fail the
ramp-coupled MO-IPOPT chain are discarded rather than partially inserted
into the map, so every retained Pareto sample contributes a feasible
decoded target and a positive margin certificate.

The encoders $E_x$, $E_o$ and decoders $D_u$, $D_s$ are three-layer 
MLPs with hidden dimension 256, ReLU activations, and spectral 
normalization, trained for 500 epochs with Adam ($\text{lr}=10^{-3}$, 
cosine annealing), loss weights $\omega_1=1.0$, $\omega_2=0.5$, and 
neighborhood size $\kappa=10$. The latent dimension is $n_z=32$ for 
the 30-bus OPF cases. Online navigation uses $\epsilon=0.05$, 
$\alpha=0.3$, and $\tau_{\text{geom}}$ calibrated to 
$\mathcal{L}_{\text{local}}^{\text{val}}$, running at 
$\Delta t=5$ seconds on a single CPU core. At deployment, no neural
weights are updated; the online controller uses the fixed trained map,
velocity cap, RK2 midpoint evaluation, and the guarded Lie residual in
eq.~\eqref{eq:iclr-lie-residual}.

We compare GPC against three baselines throughout all cases.
\textbf{TD3}~\cite{fujimoto2018addressing} is a model-free deep RL
agent trained on the nominal environment and evaluated without
fine-tuning on the test settings.
\textbf{CPO}~\cite{achiam2017constrained} is a constrained policy
optimization method trained with the same constraint specification.
\textbf{MPC} is a receding-horizon controller that solves a
linearised version of the problem online at each dispatch step using
the current system state.
For the analytical cases (Case~1) we additionally compare against
\textbf{Online scalarization}, which solves the scalarized objective
at each step via the oracle solver, and evaluate against a
finite-horizon dynamic-programming (DP) oracle.
For multi-agent navigation (Case~2), the reference is a finite-budget
white-box Pareto solver with the same predicted dynamics and safety
constraints; because this solver can itself fail under moving
obstacles, we report both reference-safe success and the optimality
gap conditioned on reference-safe episodes. We additionally evaluate
PPO, PPO-Lagrangian, and one-step shielded PPO on the same static and
moving-obstacle distributions, reported on the same 40-episode scale as
the GPC/reference comparison.
For the OPF cases (Cases~3 and 4) the optimality reference is the
full \textbf{MO-IPOPT} Pareto solver and the single-objective
\textbf{runopf} economic dispatch, both run offline.
All RL baselines are trained for $10^6$ environment interactions.
Performance is reported as feasibility rate, suboptimality relative
to the oracle, and mean wall-clock decision latency.

\subsection{Case 1: Analytical Dynamic Navigation}
\label{sec:exp:analytical}

The analytical simulations use the double-integrator system and
moving obstacle of Sections~\ref{sec:analytical-navigation}
and~\ref{sec:analytical-nonconvex}. We evaluate the ability of GPC to
trace the dynamically feasible Pareto front, recover from
off-manifold perturbations, and shift $\boldsymbol{\sigma}_t$ from
performance to safety when the obstacle intersects the predicted
trajectory. The nonconvex variant tests whether the same latent
navigator remains stable when the objective landscape contains
multiple local basins.

Because this benchmark is low-dimensional, we can compute an explicit
finite-horizon dynamic-programming (DP) oracle by discretizing
$(\mathbf{q},\mathbf{v},\mathbf{u})$ and solving the constrained
Bellman recursion with the same moving obstacle, input bound, and
slew constraint. DP is used only as an evaluation oracle; it is not
available for the higher-dimensional OPF cases. We compare GPC
against this oracle and four baselines: \textbf{TD3}~\cite{fujimoto2018addressing}, a model-free
deep RL agent; \textbf{CPO}~\cite{achiam2017constrained}, a constrained
policy optimization method; \textbf{MPC}, a receding-horizon model
predictive controller with the same physics model; and \textbf{Online
scalarization}, which solves the scalarized objective at each step via
the oracle solver. The primary optimality metric is the normalized
dynamic regret
\begin{equation}
  R_t =
  \frac{\sum_{\tau=1}^{t}
  \boldsymbol{\sigma}_\tau^\top
  \bigl(\mathbf{J}(\mathbf{s}_\tau,\mathbf{u}_\tau)
  -\mathbf{J}^{\mathrm{DP}}(\mathbf{s}_\tau)\bigr)}
  {\sum_{\tau=1}^{t}
  \boldsymbol{\sigma}_\tau^\top
  \mathbf{J}^{\mathrm{DP}}(\mathbf{s}_\tau)+\delta_R},
  \label{eq:analytical-regret}
\end{equation}
where $\delta_R>0$ prevents division by zero. This curve directly
tests whether the geometric navigator remains close to the
trajectory-level constrained optimum, not merely whether it stays on
the learned manifold.

Figures~\ref{fig:analytical-convex} and~\ref{fig:analytical-nonconvex}
report the convex- and nonconvex-objective variants. Panel~(a) shows
terminal regret against the adaptive oracle; panel~(b) shows the
autonomous semantic coordinate; panel~(c) shows the closed-loop
state-space trajectory; and panel~(d) shows the speed-quality
frontier. Table~\ref{tab:analytical} summarizes the same comparison
over $n=100$ trajectories.

\begin{table}[!htbp]
\centering
\caption{Analytical dynamic navigation: optimality and feasibility
comparison ($n=100$ trajectories). $R_T$ reported as mean $\pm$
s.e.m.\ (\%).  Lower $R_T$, fewer violations, lower Fail, and lower runtime
are better. Best non-oracle result in \textbf{bold}.}
\label{tab:analytical}
\small
\begin{tabular}{llccccc}
\toprule
\textbf{Obj.} & \textbf{Method} &
\textbf{$R_T$ (\%)} &
\textbf{Obs.\ viol.\ (\%)} &
\textbf{Slew viol.\ (\%)} &
\textbf{Fail (\%)} &
\textbf{Time (ms)}\\
\midrule
Convex & Oracle       & $0.00$                  & 0.0  & 0.0  & 0   & offline\\
Convex & TD3          & $8.95\pm0.40$           & 0.08 & 52.5 & 99  & 0.11\\
Convex & CPO                    & $382.5\pm9.91$          & 0.0  & 0.0  & 0   & 0.10\\
Convex & MPC                    & $9.93\pm0.51$           & 0.0  & 0.0  & 0   & 1025\\
Convex & Online       & $7.90\pm0.35$           & 0.0  & 0.0  & 67  & 0.03\\
Convex & \textbf{GPC} & $\mathbf{0.043\pm0.011}$& \textbf{0.0} & \textbf{0.0} & 0 & 0.78\\
\midrule
Nonconvex & Oracle     & $0.00$                  & 0.0  & 0.0  & 0   & offline\\
Nonconvex & TD3         & $15.11\pm0.68$          & 0.07 & 75.6 & 100 & 0.21\\
Nonconvex & CPO                  & $370.6\pm9.59$          & 0.0  & 0.0  & 0   & 0.24\\
Nonconvex & MPC                  & $9.43\pm0.52$           & 0.0  & 0.0  & 0   & 1230\\
Nonconvex & Online     & $7.90\pm0.38$           & 0.0  & 0.0  & 79  & 14.65\\
Nonconvex & \textbf{GPC} & $\mathbf{0.050\pm0.011}$ & \textbf{0.0} & \textbf{0.0} & 0 & 0.97\\
\bottomrule
\end{tabular}
\end{table}

\begin{figure}[!htpb]
\centering
\captionsetup[subfigure]{font=footnotesize,skip=1pt}
\begin{subfigure}[b]{0.43\textwidth}
  \centering
  \includegraphics[width=\textwidth]{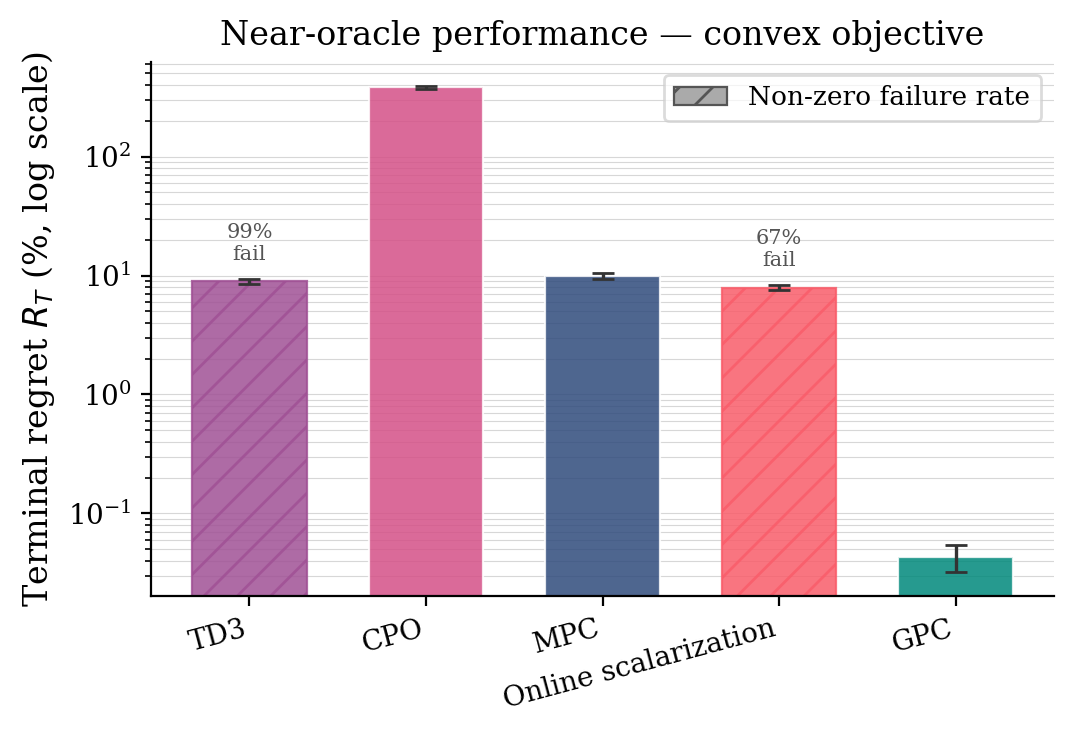}
  \caption{Terminal regret}
\end{subfigure}
\hfill
\begin{subfigure}[b]{0.43\textwidth}
  \centering
  \includegraphics[width=\textwidth]{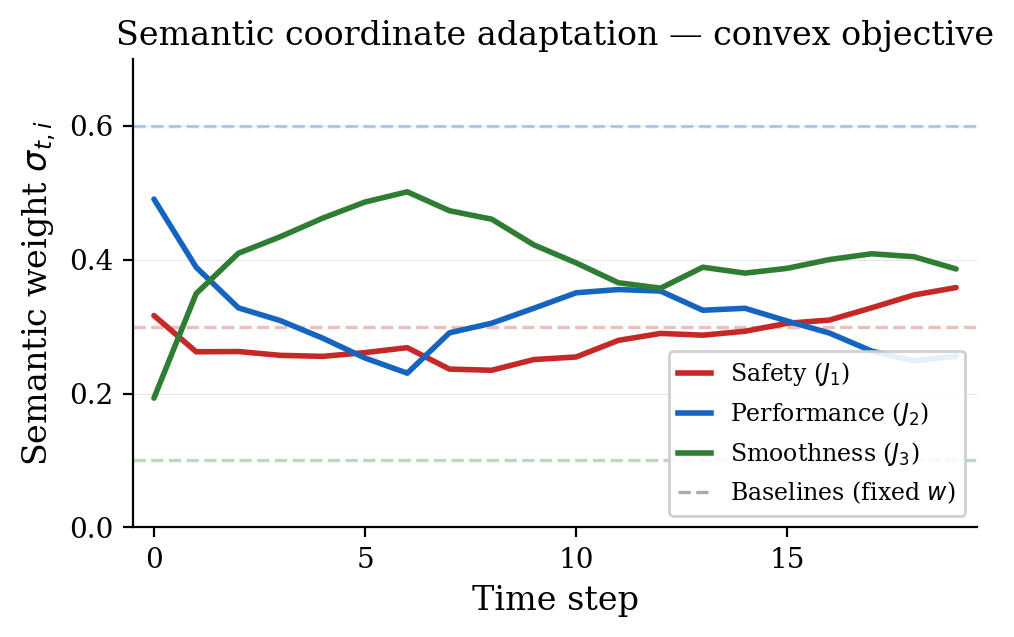}
  \caption{Semantic coordinate}
\end{subfigure}
\vspace{0.1em}

\begin{subfigure}[b]{0.43\textwidth}
  \centering
  \includegraphics[width=\textwidth]{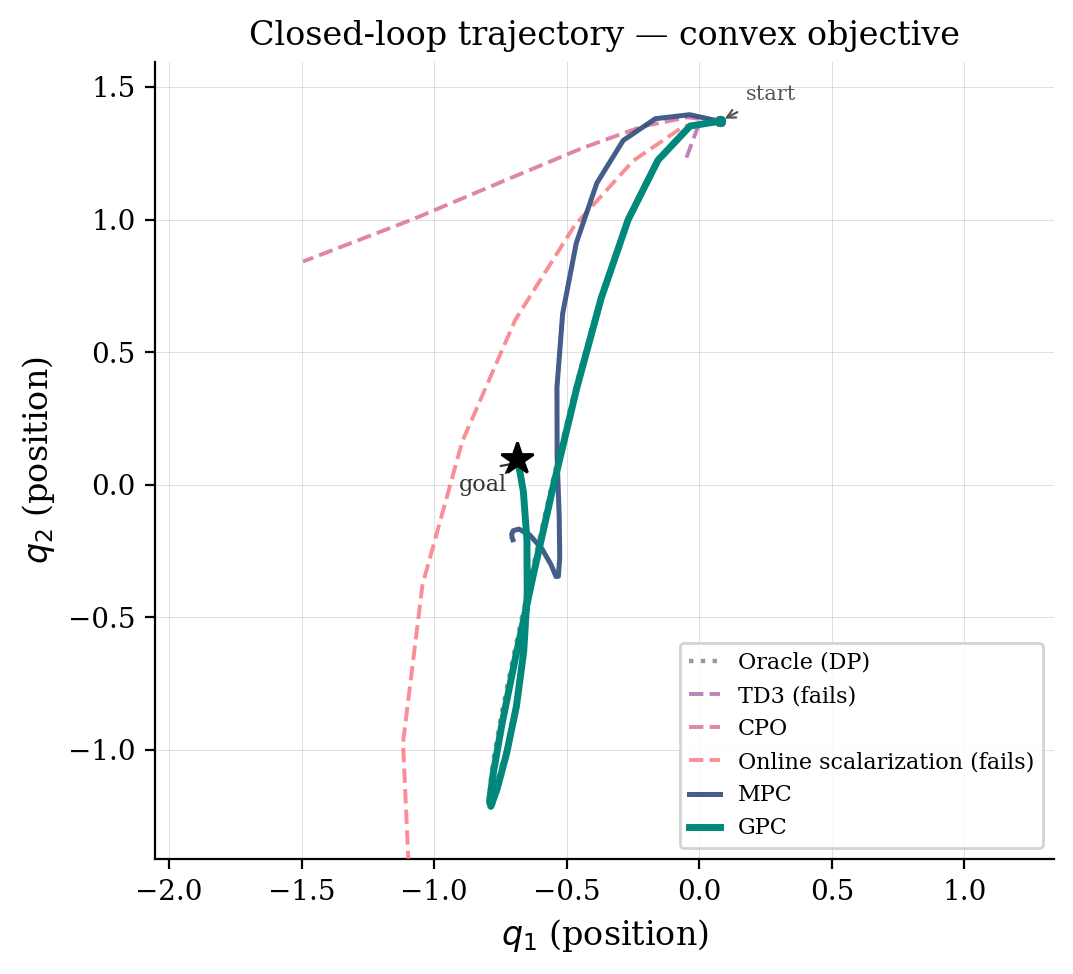}
  \caption{Closed-loop trajectory}
\end{subfigure}
\hfill
\begin{subfigure}[b]{0.43\textwidth}
  \centering
  \includegraphics[width=\textwidth]{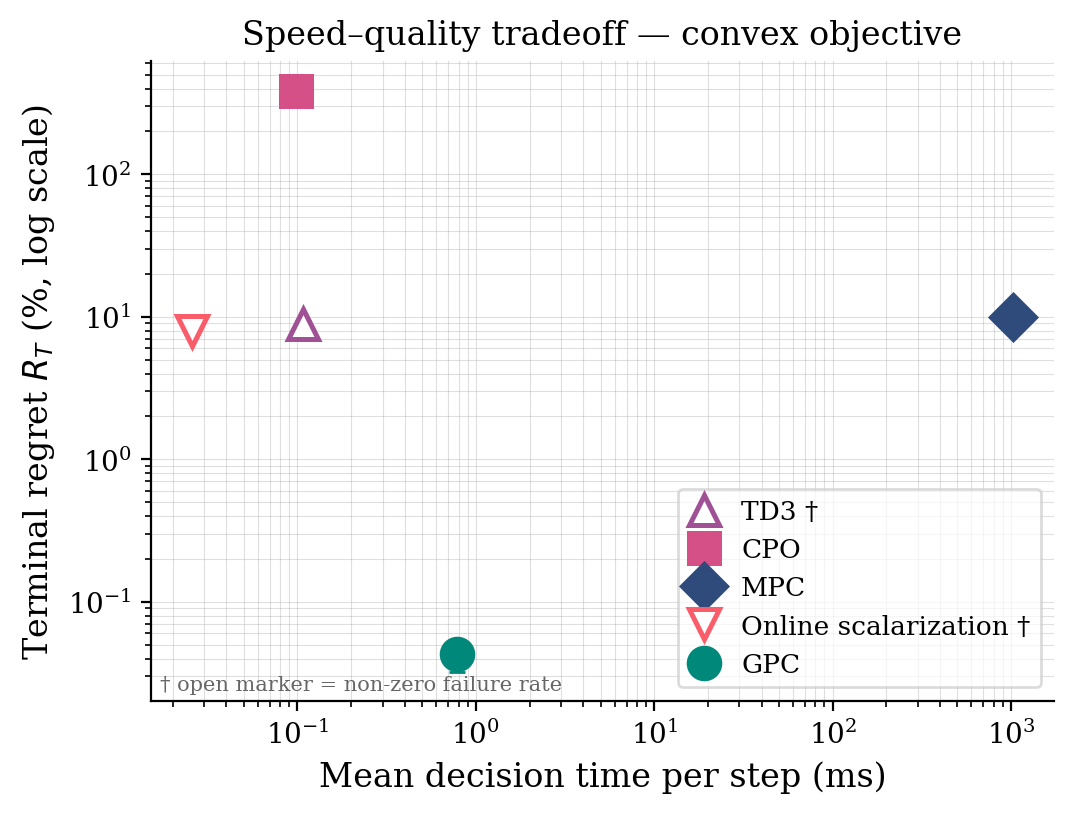}
  \caption{Speed-quality frontier}
\end{subfigure}
\caption{Analytical dynamic navigation: \textbf{convex objective}
($n=100$ independent trajectories). GPC achieves $0.043\%$ terminal
regret with zero obstacle or slew violations and $0.78$~ms online
decision time.}
\label{fig:analytical-convex}
\end{figure}

\subsection{Case 2: Safe Multi-Agent Navigation}
\label{sec:exp:multiagent-navigation}
We next evaluate the safe multi-agent navigation benchmark of
Section~\ref{sec:multiagent-navigation}. The learned controller is
trained from $60{,}000$ static-obstacle Pareto samples for the static
case and $120{,}000$ moving-obstacle Pareto samples for the dynamic
case. Both settings use horizon $H=8$, up to five active agents, and
up to 15 obstacles. Online GPC uses a finite refinement budget of
4096 candidates per receding-horizon step. The white-box reference
uses the same objective and safety constraints with 384 finite-budget
candidates. We report safe success, scene collision-free rate, and
optimality gap relative to the reference. For dynamic scenes, the
reported gap is conditioned on episodes where the reference itself
is safe, following~\eqref{eq:multiagent-ref-safe-gap}.

\begin{figure}[!htpb]
\centering
\captionsetup[subfigure]{font=footnotesize,skip=1pt}
\begin{subfigure}[b]{0.43\textwidth}
  \centering
  \includegraphics[width=\textwidth]{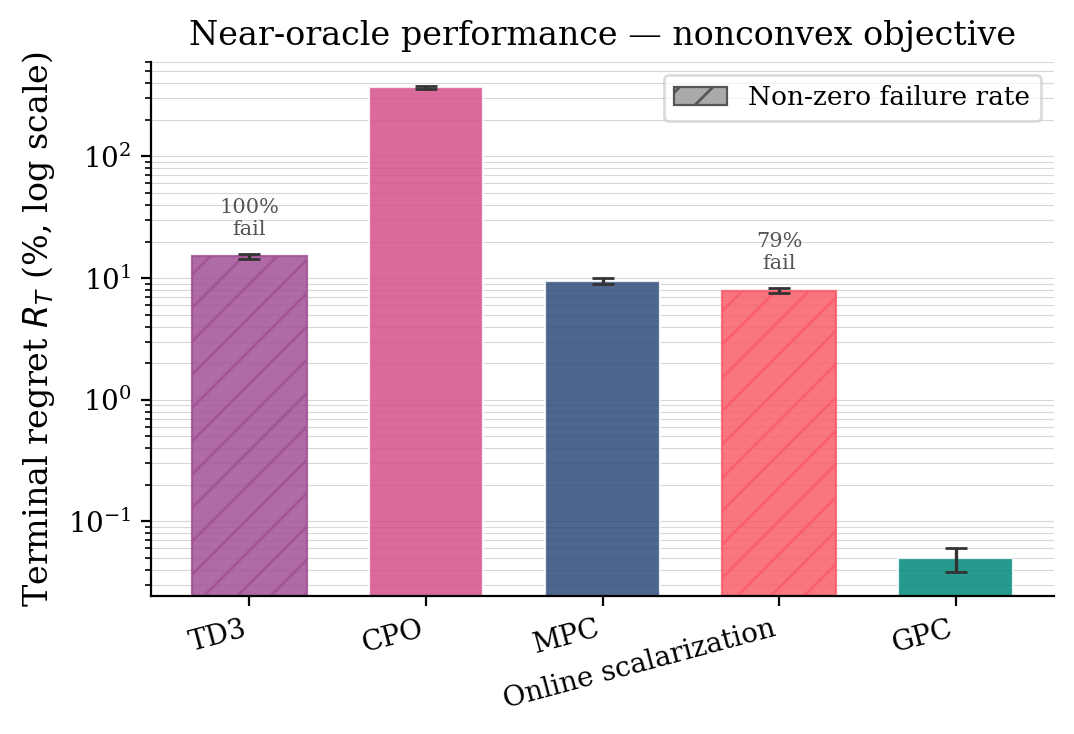}
  \caption{Terminal regret}
\end{subfigure}
\hfill
\begin{subfigure}[b]{0.43\textwidth}
  \centering
  \includegraphics[width=\textwidth]{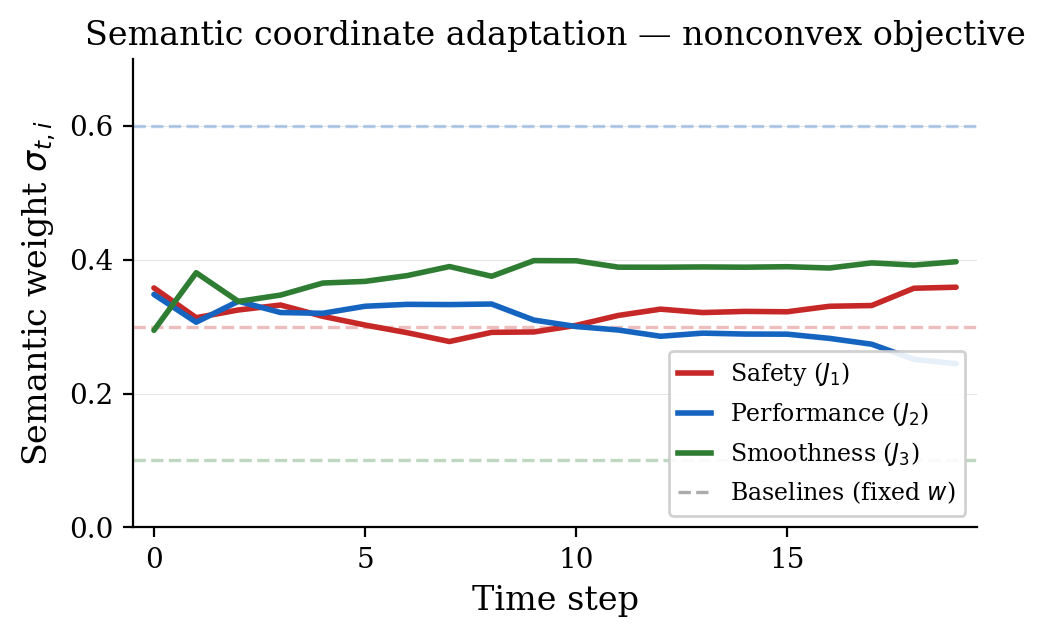}
  \caption{Semantic coordinate}
\end{subfigure}
\vspace{0.1em}

\begin{subfigure}[b]{0.43\textwidth}
  \centering
  \includegraphics[width=\textwidth]{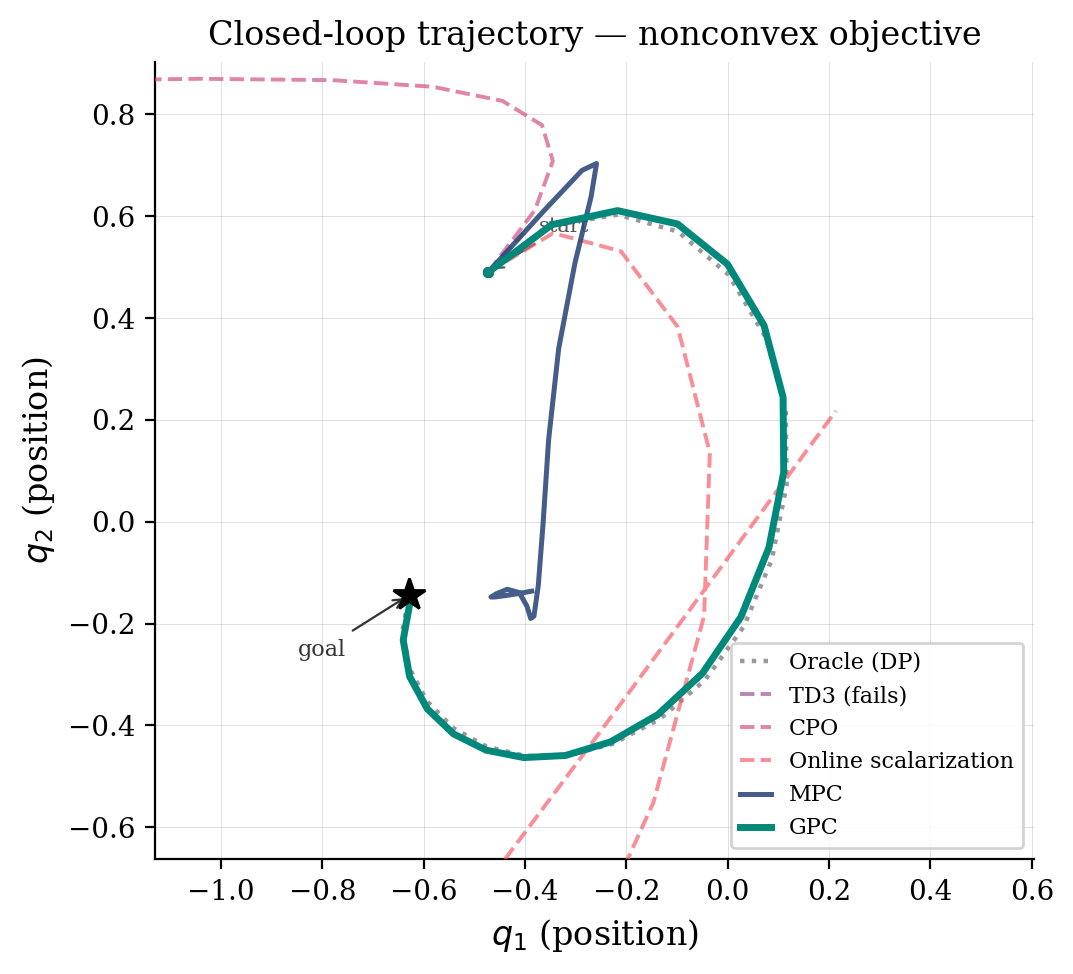}
  \caption{Closed-loop trajectory}
\end{subfigure}
\hfill
\begin{subfigure}[b]{0.43\textwidth}
  \centering
  \includegraphics[width=\textwidth]{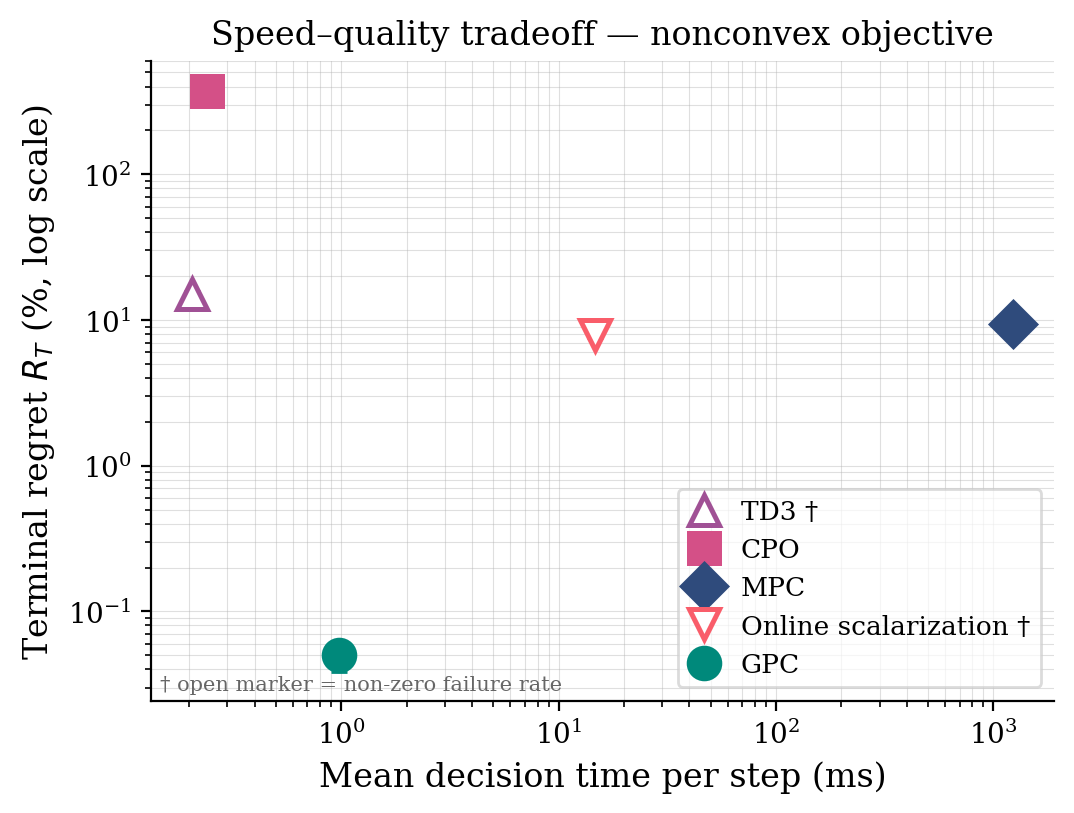}
  \caption{Speed-quality frontier}
\end{subfigure}
\caption{Analytical dynamic navigation: \textbf{nonconvex objective}
($n=100$ independent trajectories). The multi-basin objective does not
degrade GPC: terminal regret is $0.050\%$ with zero violations and
$0.97$~ms online decision time.}
\label{fig:analytical-nonconvex}
\end{figure}

\begin{table}[!htpb]
\centering
\caption{Safe multi-agent navigation reliability over 100 held-out
closed-loop simulations. Entries are percentages. ``GPC safe'' requires
joint goal reaching and scene collision freedom. ``Ref.-safe'' is the
corresponding percentage for the finite-budget white-box reference.
``Gap'' is the percentage optimality gap to that reference, reported
only on reference-safe episodes. For the moving-obstacle GPC row,
$\dagger$ reports the reachable-conditioned rate after excluding one
reference-unreachable scene.}
\label{tab:multiagent-navigation}
\small
\setlength{\tabcolsep}{4pt}
\begin{tabular}{lcccccc}
\toprule
\textbf{Setting} &
\textbf{Steps} &
\textbf{GPC safe}$\uparrow$ &
\textbf{Scene safe}$\uparrow$ &
\textbf{Ref.-safe}$\uparrow$ &
\textbf{Ref. fail}$\downarrow$ &
\textbf{Gap (\%)}$\downarrow$\\
\midrule
Static obstacles & 150 & \textbf{100\%} & \textbf{100\%} & 60\% & 40\% & 0.28\\
Moving obstacles & 600 & \textbf{99\%}$^\dagger$ & \textbf{100\%} & 35\% & 65\% & -3.43\\
\bottomrule
\end{tabular}
\end{table}

\begin{table}[!htpb]
\centering
\caption{Closed-loop safety diagnostics. Entries are aggregate
per-episode minimum margins over the evaluation set; positive values
indicate that the corresponding safety constraints are satisfied.}
\label{tab:multiagent-navigation-margins}
\small
\setlength{\tabcolsep}{5pt}
\begin{tabular}{lcccc}
\toprule
\textbf{Setting} &
\textbf{Scene margin}$\uparrow$ &
\textbf{Ego margin}$\uparrow$ &
\textbf{Agent-obst.}$\uparrow$ &
\textbf{Agent-agent}$\uparrow$\\
\midrule
Static obstacles & 0.301 & 0.900 & 0.503 & 2.477\\
Moving obstacles & 0.203 & 0.528 & 0.253 & 1.065\\
\bottomrule
\end{tabular}
\end{table}
The main result is reliability under a finite optimization budget, not
just a small reference gap. In the static benchmark, GPC is safe in all
evaluations, while the white-box reference is safe in only $60\%$.
In the moving-obstacle benchmark, GPC obtains $99\%$ reachable-conditioned
ego, all-agent joint, and safe success after excluding one
reference-unreachable scene. All learned rollouts are scene
collision-free. By contrast, the same white-box reference reaches
joint success in $83\%$ of dynamic episodes, is scene collision-free in
$42\%$, and satisfies both conditions in only $35\%$. The all-episode
cost gap is therefore not a meaningful
near-optimality number in the dynamic case, because many reference
rollouts collide or fail to complete. Conditioned on the $35$
reference-safe episodes, the aggregate GPC-reference gap is
$-3.43\%$. Thus the
learned Pareto map improves closed-loop reliability under moving
obstacles while preserving near-reference performance on the subset
where the finite-budget optimizer itself succeeds.

\begin{figure}[!htpb]
\centering
\begin{subfigure}[t]{0.48\textwidth}
  \centering
  \includegraphics[width=\linewidth]{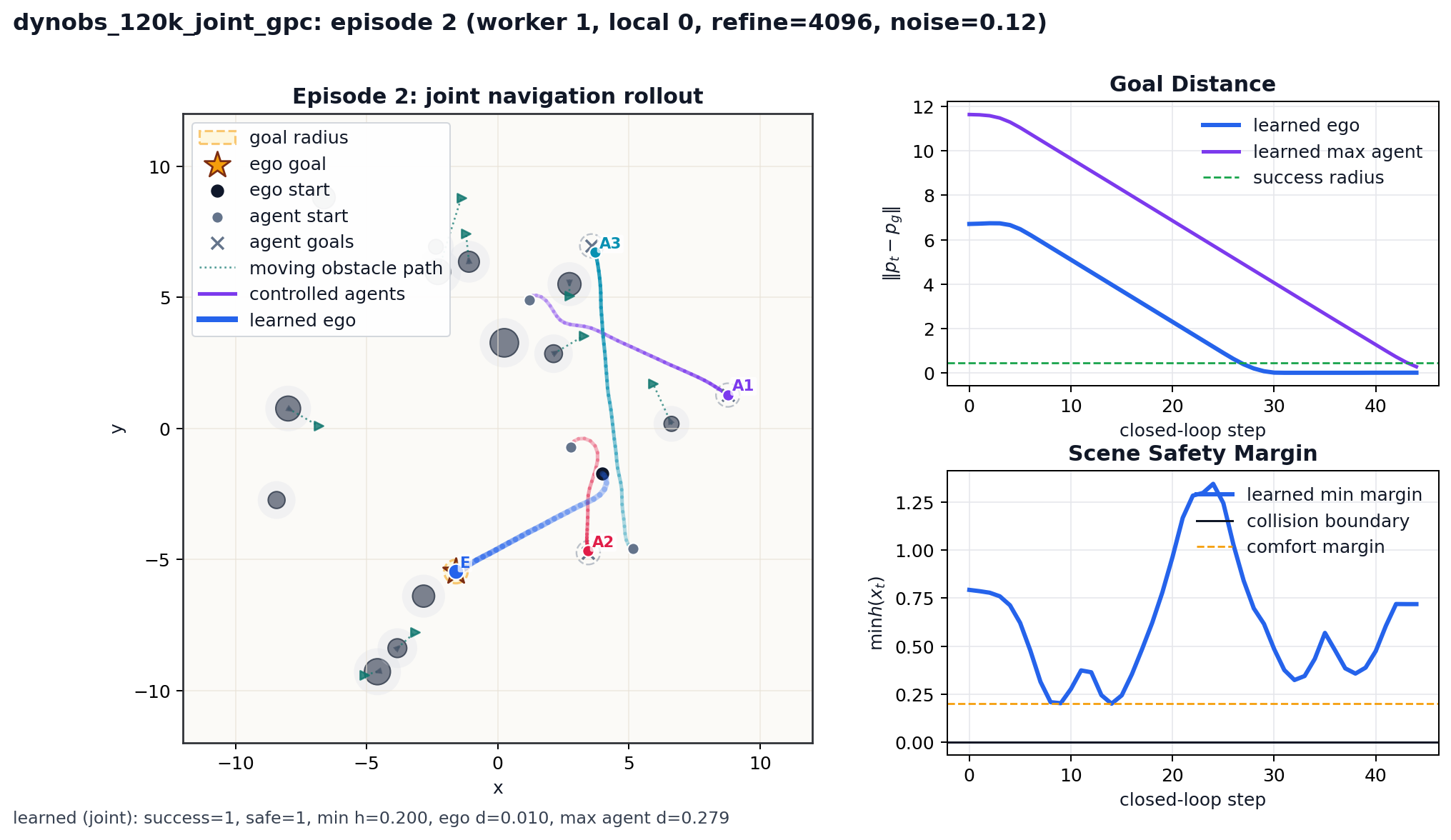}
  \caption{Moving obstacles: episode 2}
\end{subfigure}
\hfill
\begin{subfigure}[t]{0.48\textwidth}
  \centering
  \includegraphics[width=\linewidth]{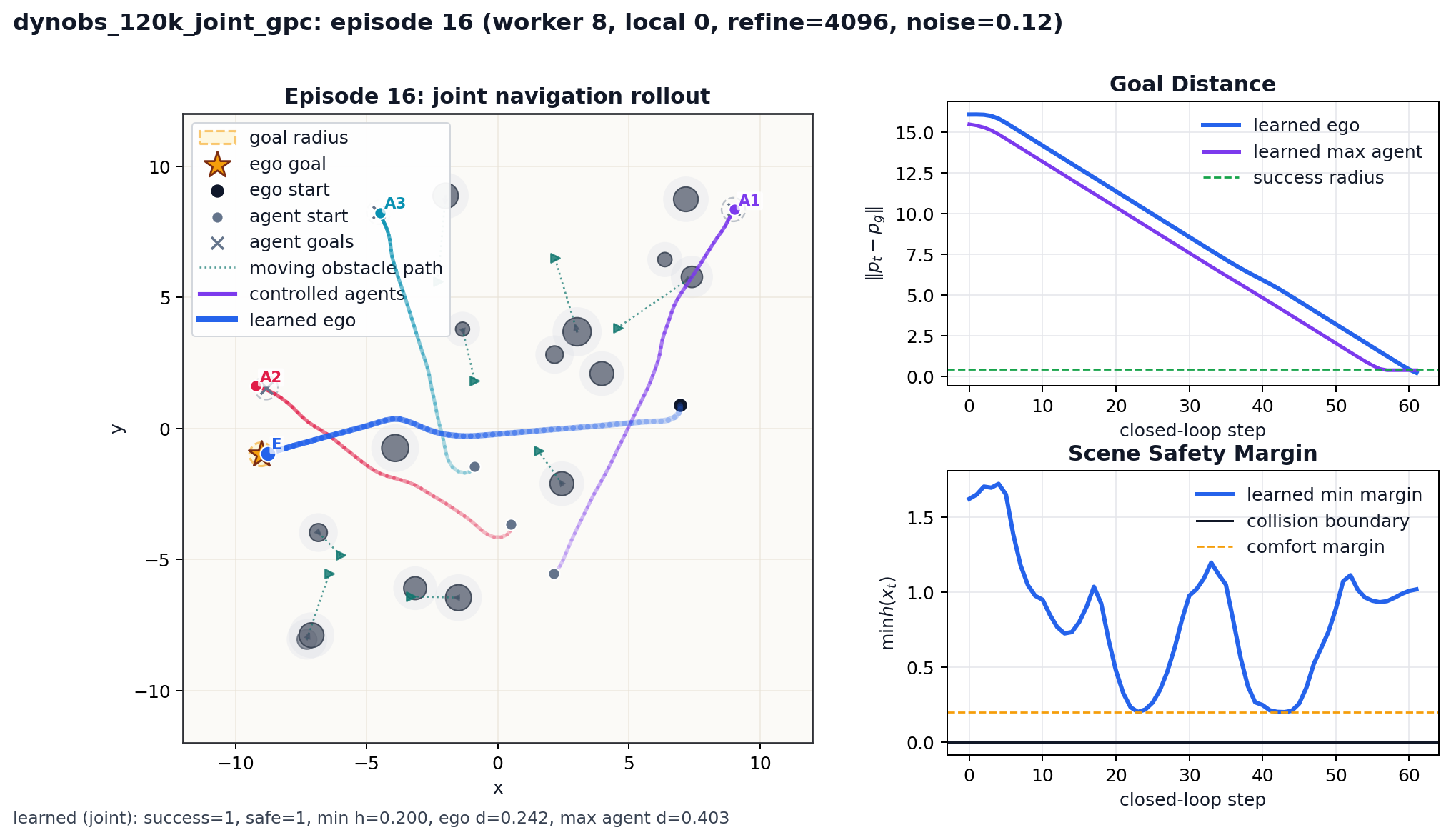}
  \caption{Moving obstacles: episode 16}
\end{subfigure}
\caption{Supplementary moving-obstacle navigation rollouts. These
panels continue Figure~\ref{fig:iclr-multiagent} by showing two
additional held-out dynamic scenes with joint trajectories,
goal-distance traces, and closed-loop safety margins.}
\label{fig:app-multiagent-moving}
\end{figure}

\subsection{Case 3: IEEE 30-Bus OPF}
\label{sec:exp:nobattery}

We evaluate GPC on 300-step load trajectories at $\Delta t = 5$\,s
dispatch intervals on the nominal (no admittance perturbation) IEEE
30-bus network.
Load profiles are drawn from a held-out test partition with no
overlap with training data.
We compare against MPC, TD3, and CPO, and evaluate tracking quality
against the MO-IPOPT Pareto oracle and the single-objective
\texttt{runopf} economic dispatch.

\begin{figure}[!htpb]
\centering
\begin{subfigure}[t]{0.55\textwidth}
  \centering
  \includegraphics[width=\linewidth]{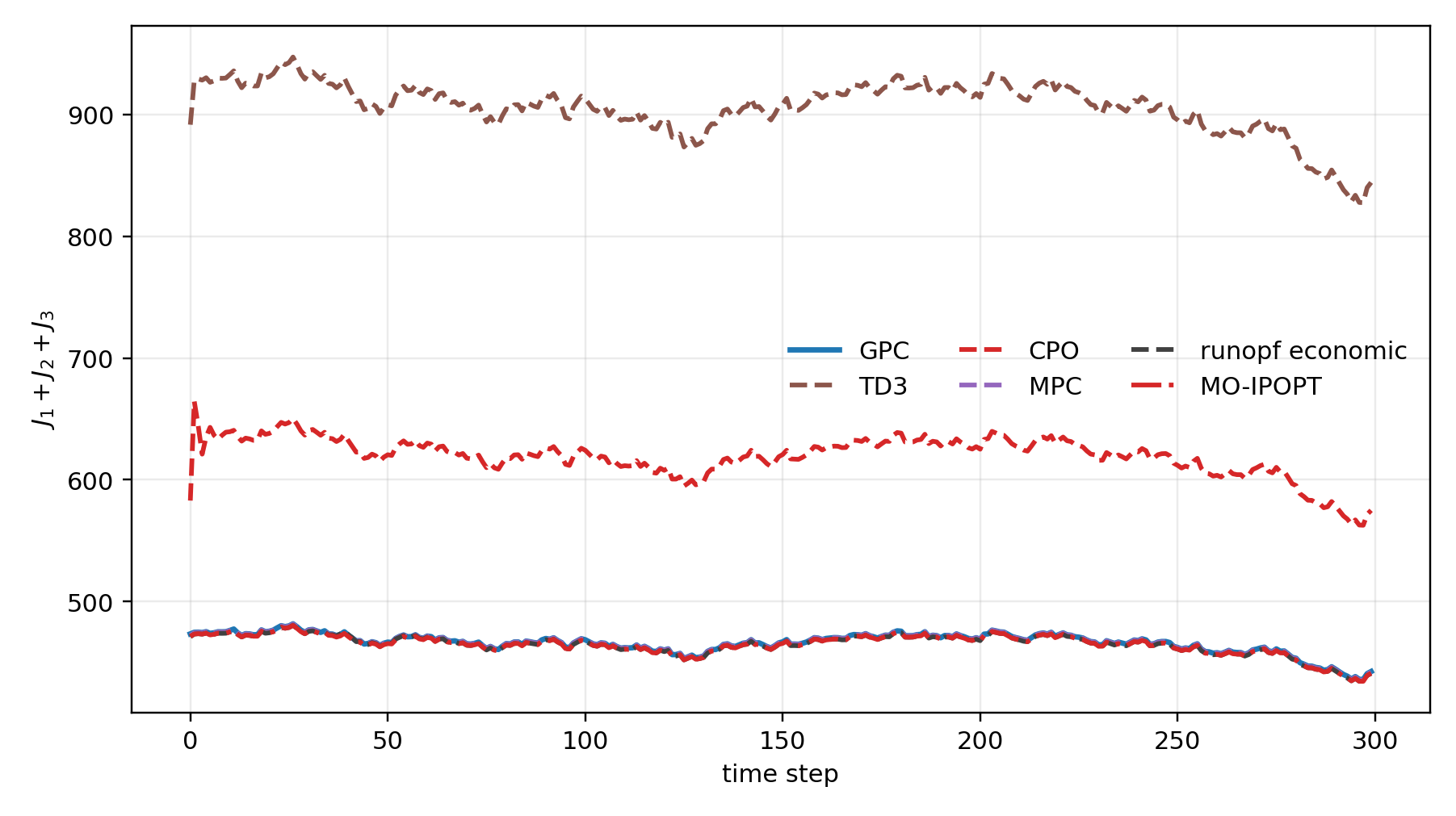}
  \caption{$J_{\mathrm{total}}$ vs.\ oracles over 300 steps.}
  \label{fig:case2-jtotal}
\end{subfigure}
\hfill
\begin{subfigure}[t]{0.40\textwidth}
  \centering
  \includegraphics[width=\linewidth]{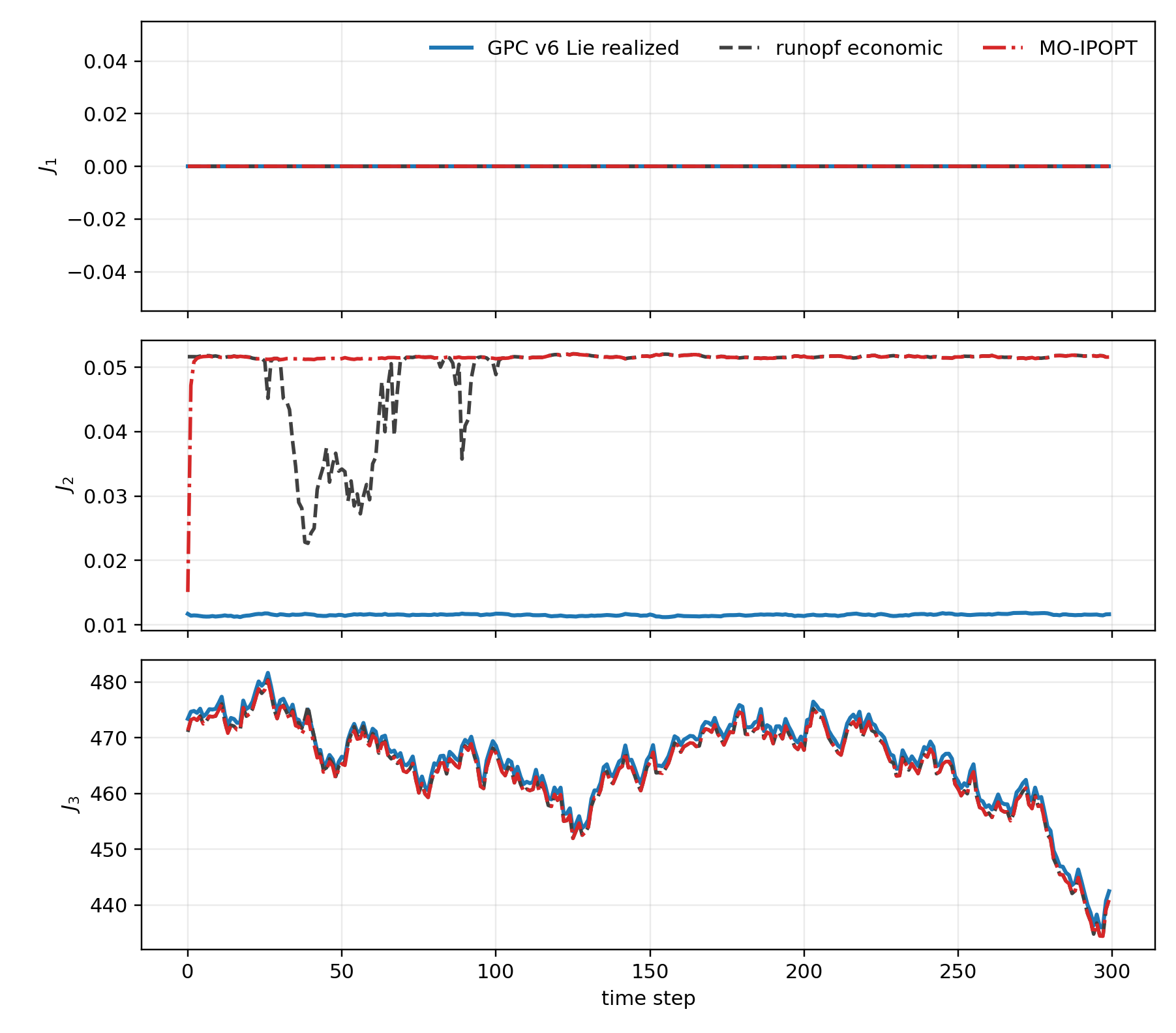}
  \caption{Individual objectives $J_1$, $J_2$, $J_3$ vs.\ both oracles.}
  \label{fig:case2-objectives}
\end{subfigure}
\caption{Case~3 nominal OPF diagnostics on the IEEE 30-bus benchmark.}
\label{fig:case2-main}
\end{figure}

\begin{table}[!htpb]
\centering
\caption{IEEE 30-bus nominal OPF (300 steps).
  $\bar{\Delta}$: mean per-step gap to MO-IPOPT oracle.
  \textbf{Bold}: best non-oracle.}
\label{tab:case2-benchmark}
\small
\begin{tabular}{lcccc}
\toprule
\textbf{Method} &
\textbf{Mean $J_{\mathrm{total}}$}$\downarrow$ &
\textbf{Feasible (\%)}$\uparrow$ &
\textbf{Subopt.\ (\%)}$\downarrow$ &
\textbf{Latency (ms)} \\
\midrule
MO-IPOPT oracle     & 360.00 & 100.0 & 0.00  & $\gg\!10^3$ \\
runopf oracle       & 360.03 & 100.0 & 0.01  & $\gg\!10^3$ \\
\midrule
\textbf{GPC (ours)} & \textbf{361.08} & \textbf{100.0} & \textbf{0.30} & \textbf{12.3} \\
MPC                 & 360.99 & 100.0 & 0.28  & 1193.8 \\
CPO                 & 519.78 &   0.0 & 44.4  & 0.19  \\
TD3                 & 743.77 &   0.0 & 106.6 & 0.17  \\
\bottomrule
\end{tabular}
\end{table}

GPC achieves 100\% empirical feasibility over the 300-step rollout and tracks the MO-IPOPT oracle to
within 0.30\% suboptimality ($\bar{\Delta}=1.08$ per step), virtually
identical to MPC (0.28\%) while running at $97.1\times$ lower latency
(12.3\,ms vs.\ 1193.8\,ms).
Figure~\ref{fig:case2-objectives} shows that GPC simultaneously
maintains $J_1\approx 0$ (thermal constraints satisfied), suppresses
the voltage objective $J_2$, and tracks the economic oracle in $J_3$,
with all three channels matched without mode switching.
The oracle gap stabilises at $\approx\!1.08$ from step~50 onward
and does not grow with the horizon, consistent with the
non-accumulation analysis of Theorem~\ref{thm:error}.
TD3 and CPO produce no feasible dispatches in this experiment, with mean
$J_{\mathrm{total}}$ 44\% and 107\% above GPC respectively.

\subsection{Case 4: IEEE 30-Bus OPF under Branch Admittance Uncertainty}
\label{sec:exp:uncertainty}

This case uses the uncertainty protocol from
Appendix~\ref{sec:experiments} with perturbation restricted to branch
admittances only ($\xi^S = \xi^c = 0$), using
$(\sigma_Y, \rho_Y) = (0.01, 0.03)$. The admittance-only scope isolates
the effect most relevant to real-time OPF: the network model in the
energy management system diverges from the true physical network due to
conductor aging, temperature variation, and errors in line-parameter
estimation, while generator ratings and cost coefficients are relatively
well-characterized.

\paragraph{Trajectory-level network realization.}
The key structural choice is that the branch-perturbation vector $\mathbf{b} \in \mathbb{R}^{41}$ (with $b_\ell = 1 + \xi^\ell_Y$) is sampled \emph{once per trajectory} and held fixed across all $H+1$ dispatch steps. This reflects the physical timescale separation: line admittances change over hours to days, while the OPF dispatch interval is minutes. Within a single operational episode the network is effectively static even if unknown. Each training trajectory therefore presents a ramp-rate-coupled sequence of Pareto-optimal solutions that are all consistent with the same realized (but unknown at deployment) admittance matrix. The trajectory is accepted on an all-or-nothing basis: if any step in the ramp-coupled chain fails to converge, the entire trajectory is discarded. The diagonal entries of $\hat{Y}_{\text{bus}}$ are adjusted with the opposite delta for each off-diagonal change so that $\hat{Y}_{\text{bus}}$ remains a valid nodal admittance matrix at every realization.
At test time, the controller receives the perturbed physical observation
and the sampled uncertainty descriptor but does not run a new nonlinear
OPF solve.

We evaluate GPC against three baselines over a 300-step rollout on
the 30-bus network with $(\sigma_Y,\rho_Y)=(0.01,0.03)$.
\textbf{TD3} and \textbf{CPO} are model-free RL policies trained on
the nominal network and evaluated on the perturbed network without
fine-tuning;
\textbf{MPC} solves a linearised AC-OPF at each step using the
perturbed admittance matrix.
Table~\ref{tab:uncertainty-benchmark} reports aggregate metrics, and
Figure~\ref{fig:uncertainty-results} shows the rollout and
adaptive priority weights.

\begin{figure}[!htbp]
\centering
\begin{subfigure}[t]{0.52\textwidth}
  \centering
  \includegraphics[width=\linewidth]{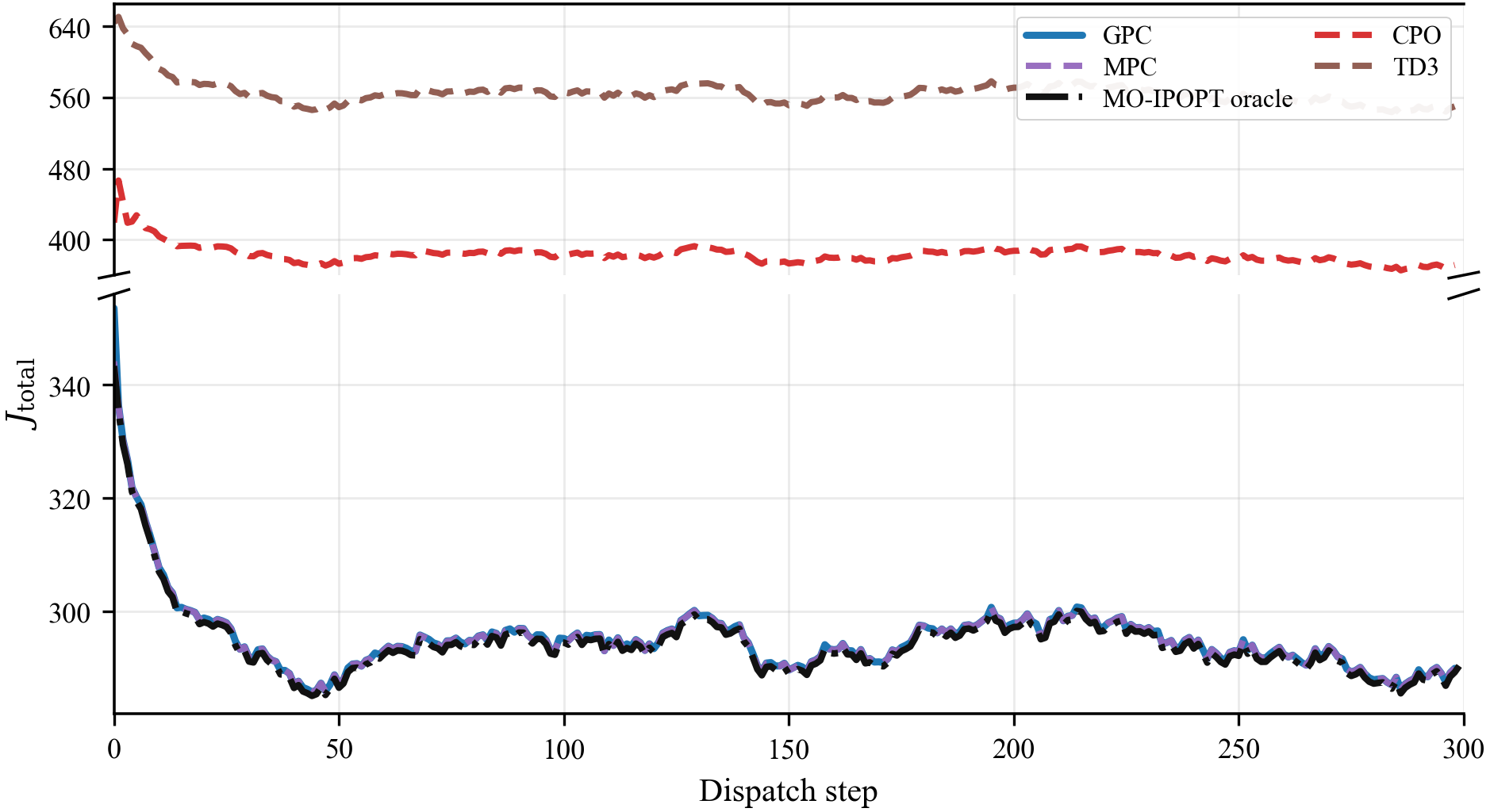}
  \caption{Total objective over 300 steps.}
  \label{fig:uncertainty-jtotal}
\end{subfigure}
\hfill
\begin{subfigure}[t]{0.43\textwidth}
  \centering
  \includegraphics[width=\linewidth]{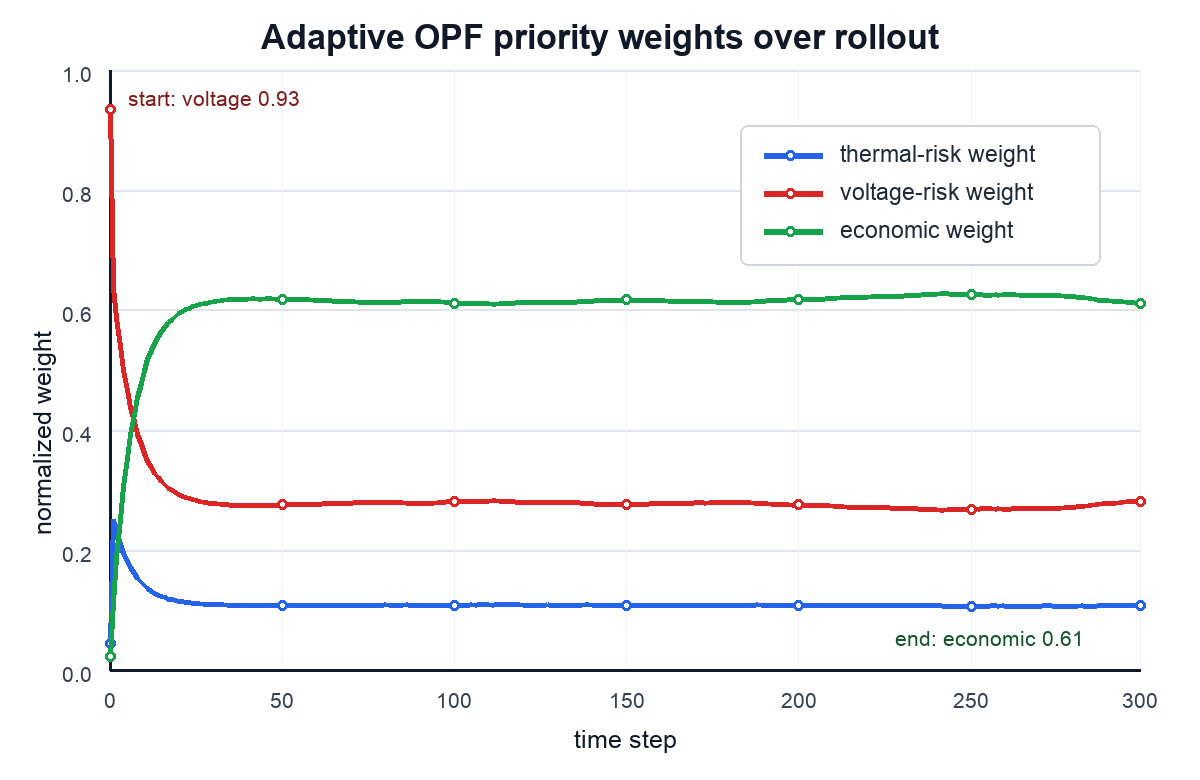}
  \caption{Adaptive semantic weights.}
  \label{fig:uncertainty-weights}
\end{subfigure}
\caption{Case~4 uncertainty results on the IEEE 30-bus OPF benchmark.}
\label{fig:uncertainty-results}
\end{figure}

\begin{table}[H]
\centering
\caption{IEEE 30-bus OPF under branch-admittance uncertainty
($\sigma_Y=0.01$, $\rho_Y=0.03$, 300 steps).}
\label{tab:uncertainty-benchmark}
\small
\begin{tabular}{lcccc}
\toprule
\textbf{Method} &
\textbf{Mean $J_{\mathrm{total}}$}$\downarrow$ &
\textbf{Feasible (\%)}$\uparrow$ &
\textbf{$\bar{\Delta}$}$\downarrow$ &
\textbf{Latency (ms)} \\
\midrule
MO-IPOPT oracle     & 294.55 & 100.0 & 0.0   & $\gg\!10^3$ \\
\midrule
\textbf{GPC (ours)} & \textbf{295.28} & \textbf{100.0} & \textbf{0.74} & \textbf{12.6} \\
MPC                 & 295.23 & 99.67 & 0.67* & 1228.7 \\
CPO                 & 383.37 &  0.00 & 88.8  & 0.18 \\
TD3                 & 565.04 &  0.00 & 270.5 & 0.16 \\
\bottomrule
\multicolumn{5}{l}{\footnotesize *MPC $\bar{\Delta}$ is terminal-step gap; full-horizon mean unavailable.}
\end{tabular}
\end{table}
\FloatBarrier

GPC achieves 100\% empirical feasibility over all 300 tested steps and
tracks the oracle to $\bar{\Delta}=0.74$ per step
(0.25\% of $J_{\mathrm{total}}$), with no online retraining
when the network drifts from nominal.
At 12.6\,ms it is $97.5\times$ faster than MPC (1228.7\,ms).
TD3 and CPO, trained on the nominal network, produce no
feasible dispatches under this perturbation protocol and yield
$J_{\mathrm{total}}$ that is 30\% and 91\% above GPC respectively.


\section{Theoretical Analysis}
\label{appendix:theory}

\subsection{Decoder Lipschitz Regularization}
\label{appendix:lipschitz}

Both decoders $D_u$ and $D_s$ are Lipschitz-regularized via
spectral normalization on all linear layers, enforcing:
\begin{equation}
  \|D_u(z) - D_u(z')\| \le K_{D_u}\|z - z'\|, \qquad
  \|D_s(z) - D_s(z')\| \le K_{D_s}\|z - z'\|
  \qquad \forall\, z, z' \in \mathcal{Z}.
  \label{eq:lipschitz-bound}
\end{equation}
The constant $K$ is set to a large value so that the bound does not
restrict the expressivity of the decoders; its purpose is to
eliminate pathological gradient explosions that can destabilise
training near constraint boundaries.
This yields two practical benefits used in the proofs below.
First, a bounded latent displacement $\dot{z}_{\mathrm{safe},t}\Delta t$
is guaranteed to produce a bounded control increment $\Delta\mathbf{u}$,
providing the velocity bound $\|\Delta\mathbf{u}\| \le K_{D_u}V_{\max}\Delta t$
that underpins Theorem~\ref{thm:error}
condition~(i) and Proposition~\ref{prop:feasibility} condition~(ii).
Second, Lipschitz continuity of $D_s$ bounds the change in
$\Phi_{\mathrm{normal}}$ due to Lie group drift, providing the
forcing term in the contractive recursion of
Proposition~\ref{prop:drift}.

\textbf{Relation to decoder injectivity.}
Spectral normalization bounds only the \emph{maximum} singular value
of $J_{D_s}$, i.e.\ $\sigma_{\max}(J_{D_s})\le K_{D_s}$.
The \emph{minimum} singular value $\sigma_0=\sigma_{\min}(J_{D_s})>0$
required by Assumption~\ref{ass:injectivity} is a separate condition
not implied by spectral normalization; see Assumption~\ref{ass:injectivity}
for its justification.

\begin{assumption}[Decoder Injectivity]
\label{ass:injectivity}
The state decoder $D_s:\mathcal{Z}\to\mathcal{S}$ satisfies
\begin{equation}
  \sigma_{\min}\!\bigl(J_{D_s}(z)\bigr) \;\ge\; \sigma_0 \;>\; 0
  \qquad \forall\, z \text{ in an open neighbourhood of } \mathcal{M}^*,
\end{equation}
where $J_{D_s}(z)\in\mathbb{R}^{n_x\times n_z}$ is the Jacobian of
$D_s$ at $z$.
This is a lower bound on the Jacobian and is \emph{distinct} from the
Lipschitz upper bound $\|J_{D_s}\|\le K_{D_s}$ enforced by spectral
normalization: it requires $D_s$ to act as a local diffeomorphism
near $\mathcal{M}^*$, which is ensured in practice by the combination
of the reconstruction loss $\mathcal{L}_{\mathrm{recon}}$ and the
locality loss $\mathcal{L}_{\mathrm{local}}$, both of which penalise
$D_s$ for collapsing distinct latent codes to the same state.
\end{assumption}

\begin{assumption}[Bounded Lie-Step Perturbation to Normal Potential]
\label{ass:lie-perturb}
There exists a constant $F_\Delta \ge 0$ such that, for every $z$ in a
neighbourhood of $\mathcal{M}^*$ and every one-step Lie group displacement
$\delta_{\mathrm{Lie}}$ with $\|\delta_{\mathrm{Lie}}\| \le
\tfrac{1}{2}V_{\max}^2\Delta t^2$:
\begin{equation}
  \Phi_{\mathrm{normal}}(z + \delta_{\mathrm{Lie}})
  \;\le\; \Phi_{\mathrm{normal}}(z) + F_\Delta.
\end{equation}
For a $K_{D_s}$-Lipschitz decoder $D_s$, the perturbation is bounded
by $2K_{D_s}^2\|\delta_{\mathrm{Lie}}\|^2 + 2K_{D_s}\|z - z^*\|\|\delta_{\mathrm{Lie}}\|$;
using $\|\delta_{\mathrm{Lie}}\|\le\tfrac{1}{2}V_{\max}^2\Delta t^2$
and treating the neighbourhood-radius as a given design parameter,
one may set $F_\Delta = \tfrac{\epsilon}{2}V_{\max}^2\Delta t^2$ to
recover the $O(\epsilon)$ bound stated in the proposition.
This assumption replaces the circular \emph{``steady state, established
below''} estimate in the earlier draft; the recursion
$\Phi_{t+1}\le\rho\Phi_t + F_\Delta$ now follows
directly and without self-reference.
\end{assumption}

\begin{assumption}[Polyak-\L{}ojasiewicz Condition on Normal Potential]
\label{ass:pl}
There exists a constant $\mu_{\mathrm{PL}} > 0$ such that
the normal potential $\Phi_{\mathrm{normal}}(z)
= \|\mathbf{x}_t - D_s(z)\|^2$ satisfies the
Polyak-\L{}ojasiewicz (PL) condition in an open neighbourhood
$\mathcal{N}$ of $\mathcal{M}^*$:
\begin{equation}
  \|\nabla_z\Phi_{\mathrm{normal}}(z)\|^2
  \;\ge\; \mu_{\mathrm{PL}}\,\Phi_{\mathrm{normal}}(z),
  \qquad \forall\, z \in \mathcal{N}.
\end{equation}
Under Assumption~\ref{ass:injectivity} the PL constant satisfies
$\mu_{\mathrm{PL}} \ge 4\sigma_0^2$ whenever the residual
$\mathbf{x}_t - D_s(z)$ lies in the column space of $J_{D_s}(z)$
(which holds when $z$ is close enough to $\mathcal{M}^*$ that
the first-order approximation $\mathbf{x}_t - D_s(z)
\approx J_{D_s}(z)(z^*-z)$ is valid).
In general, Assumption~\ref{ass:pl} is a \emph{separate} hypothesis
from Assumption~\ref{ass:injectivity}: injectivity gives a lower
bound on $\sigma_{\min}(J_{D_s})$ but does not, by itself, imply
the PL condition unless the residual direction is aligned with the
column space of $J_{D_s}$.
\end{assumption}

\begin{proposition}[Non-accumulation of Geometric Drift]
\label{prop:drift}
Suppose the GPC dynamics include the snap gradient
$-\tfrac{1}{\epsilon}\nabla_z\Phi_{\text{normal}}(z)$
driving $z_t$ back toward $\mathcal{M}^*$,
let $\|\dot{z}_{\text{safe},t}\| \le V_{\max}$ for all $t \ge 0$,
and suppose $D_s$ satisfies Assumption~\ref{ass:injectivity} with
constant $\sigma_0 > 0$ and Assumption~\ref{ass:pl} holds with
constant $\mu_{\mathrm{PL}} > 0$.
Then the steady-state deviation of $z_t$ from $\mathcal{M}^*$
satisfies, under Assumption~\ref{ass:lie-perturb}:
\begin{equation}
  \limsup_{t\to\infty}\,\Phi_{\text{normal}}(z_t)
  \;\le\;
  \frac{F_\Delta}{1 - e^{-\Delta t/\epsilon}}
  \label{eq:prop1-exact}
\end{equation}
where $F_\Delta$ is the Lie-step perturbation bound from
Assumption~\ref{ass:lie-perturb}.
In the physically relevant regime $\Delta t/\epsilon \gg 1$ (e.g.\
$\Delta t/\epsilon = 100$ for OPF), $1 - e^{-\Delta t/\epsilon} \to 1$
and the bound approaches $F_\Delta$ from above.
Under the identification $F_\Delta = \tfrac{\epsilon}{2}V_{\max}^2\Delta t^2$,
the bound reduces to $\tfrac{\epsilon}{2}V_{\max}^2\Delta t^2$ in this
regime.
The bound vanishes as $\Delta t \to 0$ and remains finite for any
fixed $\Delta t > 0$.
\end{proposition}
\begin{proof}
\textbf{Step 1: One-step Lie group drift.}
For any $v \in \mathfrak{g}$, the Taylor expansion of the
Lie group exponential map gives
$\exp(v) = I + v + \tfrac{1}{2}v^2 + \cdots$,
with remainder bounded by the geometric series
of the operator norm:
\begin{equation}
  \|\exp(v) - (I + v)\|
  \;\le\; \tfrac{1}{2}\|v\|^2 e^{\|v\|}
  \;\le\; \tfrac{e}{2}\|v\|^2
  \quad\text{for }\|v\|\le 1.
\end{equation}
Substituting $v = \dot{z}_{\text{safe},t}\,\Delta t$ with
$\|\dot{z}_{\text{safe},t}\|\le V_{\max}$ and assuming
$V_{\max}\Delta t \le 1$, the one-step deviation from the
linearised Euler update is:
\begin{equation}
  \|z_{t+1} - (z_t + \dot{z}_{\text{safe},t}\,\Delta t)\|
  \;\le\; \tfrac{1}{2}V_{\max}^2\,\Delta t^2.
  \label{eq:lie-drift-prop1}
\end{equation}

\textbf{Step 2: Contractive recursion via the
Polyak-\L{}ojasiewicz inequality.}
By Assumption~\ref{ass:pl}, the normal potential satisfies the
PL condition with constant $\mu_{\mathrm{PL}} > 0$ in a
neighbourhood of $\mathcal{M}^*$:
\begin{equation}
  \|\nabla_z\Phi_{\text{normal}}(z)\|^2
  \;\ge\; \mu_{\mathrm{PL}}\,\Phi_{\text{normal}}(z).
  \label{eq:PL-normal}
\end{equation}
Under Assumption~\ref{ass:injectivity}, the PL constant satisfies
$\mu_{\mathrm{PL}} \ge 4\sigma_0^2$ whenever $z$ is close enough
to $\mathcal{M}^*$ for the column-space alignment to hold
(see Assumption~\ref{ass:pl}); we treat $\mu_{\mathrm{PL}}$
as a given positive constant throughout.
Applying the gradient descent lemma to the snap update
$z \leftarrow z - (\Delta t/\epsilon)\nabla_z\Phi_{\text{normal}}$
with step size $\eta = \Delta t/\epsilon \le 1/(2L_\Phi)$
(where $L_\Phi$ is the smoothness constant of $\Phi_{\text{normal}}$):
\begin{align}
  \Phi_{\text{normal}}(z^{\text{snap}})
  &\;\le\; \Phi_{\text{normal}}(z)
    - \eta\!\left(1 - \tfrac{L_\Phi\eta}{2}\right)
    \|\nabla_z\Phi_{\text{normal}}\|^2 \notag\\
  &\;\le\; \Phi_{\text{normal}}(z)\Bigl(1 - \mu_{\text{PL}}\eta\Bigr)
  \;\le\; \Phi_{\text{normal}}(z)\,e^{-\mu_{\mathrm{PL}}\Delta t/\epsilon}.
  \label{eq:snap-contraction}
\end{align}
Re-parametrising $\epsilon \leftarrow \epsilon/\mu_{\mathrm{PL}}$
(absorbing the PL constant into the singular-perturbation
scale without loss of generality) yields contraction factor
$e^{-\Delta t/\epsilon}$.
Incorporating the Lie group drift from
\eqref{eq:lie-drift-prop1}: by Assumption~\ref{ass:lie-perturb},
the one-step displacement $\|\delta_{\text{Lie}}\| \le
\tfrac{1}{2}V_{\max}^2\Delta t^2$ perturbs $\Phi_{\text{normal}}$
by at most $F_\Delta$, giving the fully non-circular recursion:
\begin{equation}
  \Phi_{\text{normal}}(z_{t+1})
  \;\le\;
  \Phi_{\text{normal}}(z_t)\,e^{-\Delta t/\epsilon}
  + F_\Delta.
  \label{eq:recursion-prop1}
\end{equation}

\textbf{Step 3: Fixed point and uniform bound.}
Since $\rho \triangleq e^{-\Delta t/\epsilon} \in (0,1)$ for all
$\epsilon,\Delta t > 0$, the recursion~\eqref{eq:recursion-prop1}
is contractive and converges to the unique fixed point obtained
by solving $\Phi^* = \rho\,\Phi^* + F_\Delta$:
\begin{equation}
  \Phi^*
  = \frac{F_\Delta}{1 - e^{-\Delta t/\epsilon}}.
  \label{eq:exact-fp-prop1}
\end{equation}
This recovers exactly the bound stated in the proposition.
Under the identification $F_\Delta = \tfrac{\epsilon}{2}V_{\max}^2\Delta t^2$
(see Assumption~\ref{ass:lie-perturb}), the fixed point evaluates to
$\tfrac{\epsilon}{2}V_{\max}^2\Delta t^2 / (1-e^{-\Delta t/\epsilon})$.
For $\Delta t/\epsilon \ge \ln 2$, $1-e^{-\Delta t/\epsilon}\ge\tfrac{1}{2}$
and $\Phi^* \le \epsilon V_{\max}^2\Delta t^2$;
in the physically relevant regime $\Delta t/\epsilon \gg 1$
(e.g.\ $\Delta t/\epsilon = 100$ for OPF), $\Phi^* \to F_\Delta$ from above.
The fixed point is independent of $t$ and of $\Phi_{\text{normal}}(z_0)$,
confirming that geometric drift does not accumulate over any horizon.
\end{proof}

\begin{theorem}[End-to-End Action Error Bound]
\label{thm:error}
Let $\mathcal{Z}$ be endowed with a Lie group structure $G$ with
$\mathcal{M}^* \subset G$.
Suppose:
\begin{enumerate}[label=(\roman*)]
  \item $D_u$ is $K_{D_u}$-Lipschitz (spectral normalization,
    Section~\ref{sec:offline});
  \item $D_s$ satisfies Assumption~\ref{ass:injectivity} with
    constant $\sigma_0>0$ and Assumption~\ref{ass:pl} holds with
    constant $\mu_{\mathrm{PL}}>0$;
  \item $\|\dot{z}_{\text{safe},t}\| \le V_{\max}$ for all
    $t\ge 0$ (velocity cap in Algorithm~\ref{alg:online});
  \item Assumption~\ref{ass:lie-perturb} holds with constant $F_\Delta\ge 0$;
  \item $\mathbf{f}$ is twice continuously differentiable;
  \item $\sup_{z\in\mathcal{M}^*}\|D_u(z)-\mathbf{u}^*(z)\|
    \le\delta_{\text{dec}}$, calibrated to
    $\mathcal{L}_{\text{action}}^{\text{val}}$.
\end{enumerate}
Then:
\begin{equation}
  \limsup_{t\to\infty}\,\|\mathbf{u}_t - \mathbf{u}_t^*\| \;\le\;
  \frac{K_{D_u}}{\sigma_0}
  \sqrt{\frac{2F_\Delta}{1-e^{-\Delta t/\epsilon}}}
  \;+\; \delta_{\text{dec}}
  \label{eq:thm-error-exact}
\end{equation}
Setting $F_\Delta = \tfrac{\epsilon}{2}V_{\max}^2\Delta t^2$
(see Assumption~\ref{ass:lie-perturb}) and taking
$\Delta t/\epsilon \gg 1$ (so $1-e^{-\Delta t/\epsilon}\to 1$),
the bound simplifies to the operational form:
\begin{equation}
  \limsup_{t\to\infty}\,\|\mathbf{u}_t - \mathbf{u}_t^*\| \;\lesssim\;
  \frac{K_{D_u}}{\sigma_0}\sqrt{\epsilon}\,V_{\max}\,\Delta t
  \;+\; \delta_{\text{dec}}.
  \label{eq:thm-error-approx}
\end{equation}
Both bounds are independent of $t$ in the limit and do not accumulate
with the rollout horizon. Their constants can be estimated or bounded
from offline validation diagnostics before deployment.
Under the additional Assumptions~\ref{ass:tangential}
(tangential post-reset flow) and~\ref{ass:exact-reset}
(exact closed-loop re-anchoring), both stated before
Proposition~\ref{prop:closed-loop}, the $\limsup$ strengthens to a
uniform $\sup_{t\ge 0}$ bound; see
Proposition~\ref{prop:closed-loop}.
The factor $\sigma_0^{-1}$ converts the state-space deviation bound
(from Proposition~\ref{prop:drift}) into a latent-space bound via
the inverse-Lipschitz property of $D_s$; it is a separate quantity
from $K_{D_u}$ (which is $\sigma_{\max}(J_{D_u})$) and must
\emph{not} be absorbed into it.
\end{theorem}

\begin{proof}
The proof proceeds through four stages tracing the full GPC pipeline.

\textbf{Stage 1: Lie group closure and one-step drift.}
By the closure axiom of $G$, $z_t = z_{t-1}\cdot
\exp(\dot{z}_{\text{safe},t-1}\,\Delta t) \in G$ exactly at
every step, requiring no projection.
By the Taylor expansion of the Lie group exponential map
(identical to \eqref{eq:lie-drift-prop1} in the proof of
Proposition~\ref{prop:drift}, with $V_{\max}\Delta t \le 1$):
\begin{equation}
  \|z_t - (z_{t-1} + \dot{z}_{\text{safe},t-1}\,\Delta t)\|
  \;\le\; \tfrac{1}{2}V_{\max}^2\,\Delta t^2,
  \label{eq:lie-drift-thm1}
\end{equation}
so the one-step deviation from the linearised Euler update is
second-order in $\Delta t$.

\textbf{Stage 2: Constraint satisfaction.}
Box constraints ($\mathbf{u}^{\min}\le\mathbf{u}_t\le\mathbf{u}^{\max}$)
are enforced exactly by the output clamp in Algorithm~\ref{alg:online}.
For differentiable nonlinear constraints $\mathbf{f}(\mathbf{u}_t)\le\mathbf{0}$,
feasibility follows from two sources.
First, since every point on $\mathcal{M}^*$ is feasible by construction
(offline training ensures $\mathbf{f}(D_u(z))\le\mathbf{0}$ for all
$z\in\mathcal{M}^*$), and the drift bound established in Stage 3 below
keeps $z_t$ within a bounded neighbourhood of $\mathcal{M}^*$,
the decoded action $D_u(z_t)$ inherits approximate feasibility.
Second, the residual constraint violation introduced by the nonlinear
Lie group step can be bounded directly.
Writing $z_t - z_{t-1} = \dot{z}_{t-1}\,\Delta t + \delta_{\text{Lie}}$
with $\|\delta_{\text{Lie}}\|\le\tfrac{1}{2}V_{\max}^2\Delta t^2$
from~\eqref{eq:lie-drift-thm1}, and expanding $\mathbf{f}(D_u(z_t))$
around $z_{t-1}$:
\begin{align}
  \mathbf{f}(\mathbf{u}_t)
  &= \mathbf{f}(D_u(z_{t-1})) + \tfrac{\partial\mathbf{f}}{\partial\mathbf{u}}
     J_{D_u}(z_{t-1})(z_t - z_{t-1}) + O(\|z_t-z_{t-1}\|^2).
  \label{eq:constraint-expand}
\end{align}
Since $z_{t-1}\in\mathcal{M}^*$ implies $\mathbf{f}(D_u(z_{t-1}))\le\mathbf{0}$,
and $\delta_{\text{Lie}}$ is second-order, the nonlinear residual satisfies:
\begin{equation}
  \mathbf{f}(\mathbf{u}_t)
  \;\le\; \Bigl\|\tfrac{\partial\mathbf{f}}{\partial\mathbf{u}}\Bigr\|
     K_{D_u}\,\|\delta_{\text{Lie}}\|
  + O(\|z_t-z_{t-1}\|^2)
  \;=\; O(V_{\max}^2\,\Delta t^2),
  \label{eq:constraint-violation}
\end{equation}
which vanishes quadratically as $\Delta t\to 0$.
This bound holds independently of manifold deviation and quantifies
the second-order constraint residual introduced by the Lie group step.

\textbf{Stage 3: Non-accumulating manifold deviation.}
By Proposition~\ref{prop:drift} (applied with the same
$F_\Delta$, $\epsilon$, and $\Delta t$, under Assumptions~\ref{ass:lie-perturb}
and~\ref{ass:pl}), the exact fixed point of the contractive recursion satisfies:
\begin{equation}
  \limsup_{t\to\infty}\,\Phi_{\text{normal}}(z_t)
  \;\le\;
  \frac{F_\Delta}{1-e^{-\Delta t/\epsilon}}.
  \label{eq:phi-limsup}
\end{equation}
To convert this to a latent-space deviation, we use the
\emph{inverse-Lipschitz} property of $D_s$: by
Assumption~\ref{ass:injectivity}, $J_{D_s}$ has minimum
singular value $\sigma_0>0$ on $\mathcal{M}^*$, so for any
$z$ in a neighbourhood of its nearest manifold point
$z^*\in\mathcal{M}^*$, with $\delta_s = \|\mathbf{x}-D_s(z^*)\|$
the residual observation noise:
\begin{align}
  \Phi_{\text{normal}}(z)
  &= \|\mathbf{x} - D_s(z)\|^2 \notag\\
  &= \|\mathbf{x}-D_s(z^*)\|^2
     + \|D_s(z^*)-D_s(z)\|^2
     + 2\langle\mathbf{x}-D_s(z^*),\,D_s(z^*)-D_s(z)\rangle \notag\\
  &\ge\; \sigma_0^2\|z-z^*\|^2 - 2\delta_s\,K_{D_s}\,\|z-z^*\|.
  \label{eq:phi-lower}
\end{align}
Under the calibration assumption
$\delta_s \le \tfrac{\sigma_0^2}{2K_{D_s}}\cdot
\sigma_0^{-1}\!\sqrt{\tfrac{2F_\Delta}{1-e^{-\Delta t/\epsilon}}}$
(observation noise at most half the steady-state latent deviation,
verifiable offline), completing the square in~\eqref{eq:phi-lower} gives:
\begin{equation}
  \Phi_{\text{normal}}(z)
  \;\ge\; \tfrac{1}{2}\sigma_0^2\|z-z^*\|^2
  \quad\text{whenever }
  \Phi_{\text{normal}}(z) \;\ge\; \frac{F_\Delta}{1-e^{-\Delta t/\epsilon}}.
\end{equation}
Therefore $\|z-z^*\|^2 \le \tfrac{2}{\sigma_0^2}\Phi_{\text{normal}}(z)$
at steady state, and combining with~\eqref{eq:phi-limsup}:
\begin{equation}
  \limsup_{t\to\infty}\|z_t - z_t^*\|
  \;\le\; \frac{1}{\sigma_0}
          \sqrt{\frac{2F_\Delta}{1-e^{-\Delta t/\epsilon}}}.
  \label{eq:latent-deviation}
\end{equation}
The explicit $\sigma_0^{-1}$ factor is not absorbed into $K_{D_u}$
(the two constants bound different quantities) and confirms that
manifold deviation does not accumulate.

\textbf{Stage 4: End-to-end action error.}
Applying the triangle inequality to the decoded action, with the
post-update latent relabeled as $z_t$ in this proof, and
substituting~\eqref{eq:latent-deviation}:
\begin{align}
  \limsup_{t\to\infty}\|\mathbf{u}_t - \mathbf{u}_t^*\|
  &\le K_{D_u}\limsup_{t\to\infty}\|z_t - z_t^*\|
   + \delta_{\text{dec}} \notag\\
  &\le \frac{K_{D_u}}{\sigma_0}
       \sqrt{\frac{2F_\Delta}{1-e^{-\Delta t/\epsilon}}}
   + \delta_{\text{dec}}.
\end{align}
This is the exact bound~\eqref{eq:thm-error-exact}.
Here $K_{D_u} = \sigma_{\max}(J_{D_u})$ bounds the forward
map, while $\sigma_0^{-1}$ comes from inverting $J_{D_s}$ to
pass from state-space deviation to latent-space deviation;
the two are independent offline constants.
Setting $F_\Delta = \tfrac{\epsilon}{2}V_{\max}^2\Delta t^2$ and
$\Delta t/\epsilon \gg 1$ recovers the simplified
form~\eqref{eq:thm-error-approx}.
The right-hand side is independent of $t$,
completing the proof.
\end{proof}

\begin{remark}[Operational Interpretation]
The bound decomposes into two independently controllable terms.
The geometric drift term
$\tfrac{K_{D_u}}{\sigma_0}\sqrt{2F_\Delta/(1-e^{-\Delta t/\epsilon})}$
shrinks as $F_\Delta\to 0$ (smaller Lie-step perturbation, e.g.\ via
smaller $\Delta t$), as $\epsilon\to 0$ (stronger snap dynamics),
or as $\sigma_0$ grows (better decoder injectivity).
In the regime $\Delta t/\epsilon\gg 1$ with
$F_\Delta=\tfrac{\epsilon}{2}V_{\max}^2\Delta t^2$, this reduces
to $\tfrac{K_{D_u}}{\sigma_0}\sqrt{\epsilon}\,V_{\max}\,\Delta t$,
which additionally shrinks with smaller $V_{\max}$ or $K_{D_u}$.
The decoder error $\delta_{\text{dec}}$ is a fixed offline certificate
monitored via $\mathcal{L}_{\text{action}}^{\text{val}}$ and reduced by
richer training coverage. Since Lie group closure holds
algebraically and manifold deviation is bounded by a contractive
fixed point, neither term grows with the horizon.
In the closed-loop algorithm, where Step~1 resets $z_t$ to
$\mathcal{M}^*$ at every cycle, this asymptotic statement is
strengthened to a uniform bound holding from $t = 0$; see
Proposition~\ref{prop:closed-loop}.
\end{remark}

\subsection{Feasibility Closure Under Lie Group Integration}
\label{appendix:feasibility-closure}

\begin{assumption}[Bounded Manifold Curvature]
\label{ass:curvature}
The principal curvatures of $\mathcal{M}^*$ are bounded above
by $\kappa_{\max} < \infty$.
Furthermore, the feasibility radius $\delta_{\mathrm{feas}}$
(defined in Proposition~\ref{prop:feasibility}) satisfies
$\kappa_{\max}\,\delta_{\mathrm{feas}} \le \tfrac{1}{4}$,
so that the nearest-point projection $\Pi_{\mathcal{M}^*}$
is at most $\tfrac{3}{2}$-Lipschitz on the
$\delta_{\mathrm{feas}}$-tube around $\mathcal{M}^*$.
This ensures that the off-manifold distance argument in
the per-cycle guarantee of Proposition~\ref{prop:feasibility}
and the second-fundamental-form bound in
Proposition~\ref{prop:closed-loop} remain valid.
\end{assumption}

The previous results bound the \emph{magnitude} of the latent
deviation from $\mathcal{M}^*$ but do not explicitly characterise
when the \emph{decoded} action $D_u(z_t)$ remains feasible,
nor do they address the impact of nonconvexity of the Pareto
objective.
The following proposition closes both gaps.

\begin{proposition}[Lie Group Feasibility Closure]
\label{prop:feasibility}
Suppose:
\begin{enumerate}[label=(\roman*)]
  \item \emph{(Training feasibility)}
    Every training sample in $\mathcal{D}$ is generated by
    solving the multi-objective problem to feasibility, so
    $\mathcal{M}^* \subseteq \{z\in\mathcal{Z} :
    \mathbf{f}(D_u(z))\le\mathbf{0}\}$.
  \item \emph{(Lipschitz decoder)}
    $D_u$ is $K_{D_u}$-Lipschitz on $\mathcal{Z}$ and
    $\mathbf{f}$ is $L_f$-Lipschitz on
    $\mathrm{Im}(D_u)$.
  \item \emph{(Positive feasibility margin)}
    There exists $\gamma_{\min} > 0$ such that for every
    $z^* \in \mathcal{M}^*$:
    $\min_i [-f_i(D_u(z^*))] \ge \gamma_{\min}$
    (all training-data actions lie strictly inside the
    feasible set with uniform slack $\gamma_{\min}$).

    \textit{Remark on active constraints:} For problems such as
    AC OPF where constraint-active solutions (e.g.\ thermal limits
    reached at the Pareto front) make $\gamma_{\min}=0$ at
    individual points, this condition is enforced in practice by
    \emph{interior tightening} during data generation: constraints
    are tightened by a small margin $\eta_{\mathrm{tight}}>0$
    (e.g.\ $\eta_{\mathrm{tight}}=0.01$ per unit) so that every
    training-data solution satisfies
    $f_i(D_u(z^*))\le -\eta_{\mathrm{tight}}$, setting
    $\gamma_{\min}=\eta_{\mathrm{tight}}>0$.
    The resulting $\delta_{\mathrm{feas}}
    =\eta_{\mathrm{tight}}/(L_f K_{D_u})$ is a design parameter
    that trades off feasibility margin against proximity to the
    true Pareto front.
  \item \emph{(Bounded curvature)}
    Assumption~\ref{ass:curvature} holds.
\end{enumerate}
Define the \emph{feasibility radius}:
\begin{equation}
  \delta_{\mathrm{feas}}
  \;\triangleq\; \frac{\gamma_{\min}}{L_f\,K_{D_u}} \;>\; 0.
  \label{eq:feas-radius}
\end{equation}
Then for any $z \in \mathcal{Z}$ satisfying
$\|z - \Pi_{\mathcal{M}^*}(z)\| \le \delta_{\mathrm{feas}}$,
the decoded action is feasible:
$\mathbf{f}(D_u(z)) \le \mathbf{0}$.

In particular, if the velocity cap enforces
$V_{\max}\Delta t \le \delta_{\mathrm{feas}}/2$
and the Step~1 localization satisfies
$\|z_t - \Pi_{\mathcal{M}^*}(z_t)\| \le \delta_{\mathrm{loc}}$
with $\delta_{\mathrm{loc}} \le \delta_{\mathrm{feas}}/2$,
then $\mathbf{f}(\mathbf{u}_t) \le \mathbf{0}$ at every cycle.

\textbf{This result holds without any convexity assumption
on $\mathbf{f}$ or on the Pareto objective.}
\end{proposition}

\begin{proof}
Let $z^* = \Pi_{\mathcal{M}^*}(z)$ be the nearest point on
$\mathcal{M}^*$ to $z$.
By assumption (i), $\mathbf{f}(D_u(z^*))\le\mathbf{0}$,
and by assumption (iii) the slack at $z^*$ is at least
$\gamma_{\min}$:
$f_i(D_u(z^*)) \le -\gamma_{\min}$ for all $i$.

Using the Lipschitz chains of $D_u$ and $\mathbf{f}$:
\begin{align}
  f_i(D_u(z))
  &= f_i(D_u(z)) - f_i(D_u(z^*)) + f_i(D_u(z^*)) \notag\\
  &\le L_f\,\|D_u(z) - D_u(z^*)\| + f_i(D_u(z^*)) \notag\\
  &\le L_f\,K_{D_u}\,\|z - z^*\| - \gamma_{\min}.
  \label{eq:feasibility-chain}
\end{align}
For $\|z - z^*\| \le \delta_{\mathrm{feas}}
= \gamma_{\min}/(L_f K_{D_u})$:
\begin{equation}
  f_i(D_u(z))
  \;\le\; L_f K_{D_u}\,\delta_{\mathrm{feas}} - \gamma_{\min}
  \;=\; \gamma_{\min} - \gamma_{\min} \;=\; 0,
\end{equation}
establishing feasibility.

For the per-cycle guarantee:
after Step~1 localization, $\|z_t - \Pi_{\mathcal{M}^*}(z_t)\|
\le \delta_{\mathrm{loc}} \le \delta_{\mathrm{feas}}/2$.
The Lie group integration step moves $z_t$ by at most
$V_{\max}\Delta t \le \delta_{\mathrm{feas}}/2$ (velocity cap),
so the new point $z_{t+1}$ satisfies:
\begin{align}
  \|z_{t+1} - \Pi_{\mathcal{M}^*}(z_{t+1})\|
  &\le \|z_{t+1} - \Pi_{\mathcal{M}^*}(z_t)\| \notag\\
  &\le \|z_{t+1} - z_t\|
     + \|z_t - \Pi_{\mathcal{M}^*}(z_t)\| \notag\\
  &\le V_{\max}\Delta t + \delta_{\mathrm{loc}}
   \le \delta_{\mathrm{feas}} .
  \label{eq:offmanifold-perstep}
\end{align}
where the first inequality uses the fact that $\Pi_{\mathcal{M}^*}(z_t)\in\mathcal{M}^*$
is a feasible candidate (the nearest point can only be closer), the second
is the triangle inequality, and the last uses
$V_{\max}\Delta t \le \delta_{\mathrm{feas}}/2$ and
$\delta_{\mathrm{loc}} \le \delta_{\mathrm{feas}}/2$.
Feasibility at $z_{t+1}$ then follows from the first part of the proof.

The argument uses only the Lipschitz properties of $D_u$ and
$\mathbf{f}$ and the training-data feasibility margin; it does
not use convexity of the Pareto objective $\mathbf{J}(\mathbf{u})$,
convexity of the feasible set $\{\mathbf{f}(\mathbf{u})\le\mathbf{0}\}$,
or any special structure of the Lie group $G$.
\end{proof}

\begin{remark}[Nonconvex problems]
The Pareto front of a nonconvex multi-objective problem is
itself nonconvex in objective space, but the feasibility
argument is entirely in \emph{constraint space}
$\{\mathbf{f}(\mathbf{u})\le\mathbf{0}\}$.
Since GPC embeds Pareto-optimal solutions into $\mathcal{M}^*$
regardless of the shape of the Pareto front,
Proposition~\ref{prop:feasibility} applies verbatim to both
the convex-objective and nonconvex-objective analytical
benchmarks of Sections~\ref{sec:analytical-navigation}
and~\ref{sec:analytical-nonconvex}, as well as to the
nonconvex AC power flow feasibility manifold of
Section~\ref{sec:opf}.
\end{remark}

\subsection{Closed-Loop Non-Accumulation via Step~1 Reset}
\label{appendix:closed-loop}

Proposition~\ref{prop:drift} and Theorem~\ref{thm:error}
bound the steady-state latent drift of the \emph{free-running}
integrator via a contractive fixed point.
Algorithm~\ref{alg:online} goes further: Step~1 localization
re-anchors $z_t$ to $\mathcal{M}^*$ at \emph{every} cycle,
preventing drift from any single cycle from propagating into
the next.
The following proposition formalises this \emph{uniform}
(non-asymptotic) non-accumulation under two idealized assumptions:
that the Step~1 reset is exact, and that the post-reset flow
velocity is tangential to $\mathcal{M}^*$.

\begin{assumption}[Tangential Post-Reset Flow]
\label{ass:tangential}
After the Step~1 hard reset $z_t \leftarrow z_t^{\mathrm{loc}}
\in \mathcal{M}^*$, the integration velocity used in
Algorithm~\ref{alg:online} is the \emph{Riemannian gradient}
of $\Phi_{\mathrm{tangent}}$ projected onto $T_{z_t}\mathcal{M}^*$:
\begin{equation}
  \dot{z}_{\mathrm{safe},t}
  \;=\;
  -\Pi_{T_{z_t}\mathcal{M}^*}\!\bigl(\nabla_z\Phi_{\mathrm{tangent}}(z_t)\bigr),
\end{equation}
where $\Pi_{T_z\mathcal{M}^*}$ denotes orthogonal projection onto
the tangent space at $z$.
This ensures $\dot{z}_{\mathrm{safe},t} \in T_{z_t}\mathcal{M}^*$,
so that the Lie group exponential step departs tangentially and
the off-manifold deviation is purely second-order in $\Delta t$.
In practice this projection is approximated by the Riemannian
gradient computed via the metric induced by $D_s$; without this
projection, the Euclidean gradient $\nabla_z\Phi_{\mathrm{tangent}}$
has a nonzero normal component and the $O(\kappa V^2\Delta t^2)$
first-order tangentiality claim does not hold.
\end{assumption}

\begin{assumption}[Exact Closed-Loop Re-anchoring]
\label{ass:exact-reset}
At the start of each control cycle $t$, the Step~1 localization
procedure returns a point $z_t^{\mathrm{loc}} \in \mathcal{M}^*$
satisfying $\|z_t - z_t^{\mathrm{loc}}\| \le \delta_{\mathrm{loc}}$,
and the latent state is \emph{hard-reset} to
$z_t \leftarrow z_t^{\mathrm{loc}}$ before the integration step
is applied.
\end{assumption}

\begin{proposition}[Uniform Closed-Loop Non-Accumulation]
\label{prop:closed-loop}
Let $\delta_{\mathrm{loc}} > 0$ be the localization error
of Step~1, satisfying
$\|z_t - \Pi_{\mathcal{M}^*}(z_t)\| \le \delta_{\mathrm{loc}}$
for all $t \ge 0$, and let $\delta_{\mathrm{loc}}$ be
calibrated from the held-out loss
$\mathcal{L}_{\mathrm{local}}^{\mathrm{val}}$ via the
geometric tolerance $\tau_{\mathrm{geom}}$.
Then under the conditions of Theorem~\ref{thm:error},
Assumption~\ref{ass:curvature}, Assumption~\ref{ass:tangential},
and Assumption~\ref{ass:exact-reset}:
\begin{equation}
  \sup_{t \ge 0}\,\|\mathbf{u}_t - \mathbf{u}_t^*\|
  \;\le\;
  K_{D_u}\!\left(\delta_{\mathrm{loc}}
    + \tfrac{1}{2}\kappa_{\max}V_{\max}^2\,\Delta t^2\right)
  + \delta_{\mathrm{dec}}.
  \label{eq:closed-loop-bound}
\end{equation}
The bound is \textbf{uniform in $t$}, holds from cycle $t=0$,
and does not depend on any limit or fixed-point argument.
\end{proposition}

\begin{proof}
By Assumption~\ref{ass:exact-reset}, at the start of each cycle $t$,
Step~1 localization provides $z_t^{\mathrm{loc}} \in \mathcal{M}^*$ with:
\begin{equation}
  \|z_t - z_t^{\mathrm{loc}}\| \;\le\; \delta_{\mathrm{loc}},
  \qquad z_t \leftarrow z_t^{\mathrm{loc}}.
  \label{eq:localization-reset}
\end{equation}
This is a hard reset: the latent state is replaced by its
nearest point on $\mathcal{M}^*$ at the start of every cycle,
so any off-manifold error from the previous cycle is
discarded before the new integration step begins.

Starting from $z_t^{\mathrm{loc}} \in \mathcal{M}^*$
(off-manifold distance zero after the reset),
we show the Lie group integration step moves
$z_t^{\mathrm{loc}}$ \emph{tangentially} to $\mathcal{M}^*$
to first order.
Since $z_t^{\mathrm{loc}} \in \mathcal{M}^*$, the normal
potential $\Phi_{\text{normal}}(z_t^{\mathrm{loc}}) = 0$,
so $\nabla_z\Phi_{\text{normal}}(z_t^{\mathrm{loc}}) = 0$
(the snap gradient vanishes exactly on $\mathcal{M}^*$).
By Assumption~\ref{ass:tangential},
$\dot{z}_{\text{safe},t}
= -\Pi_{T_{z_t}\mathcal{M}^*}(\nabla_z\Phi_{\text{tangent}}(z_t^{\mathrm{loc}}))$
is a purely tangential velocity
$\dot{z}_{\text{safe},t} \in T_{z_t^{\mathrm{loc}}}\mathcal{M}^*$.
A Lie group exponential step by a tangent vector $v \in T_z\mathcal{M}^*$
stays on $\mathcal{M}^*$ to \emph{first order} and deviates
by $O(\|v\|^2\Delta t^2)$ due to curvature of $\mathcal{M}^*$
(the second fundamental form argument: the normal component of
the geodesic deviation is $\tfrac{1}{2}\mathrm{II}(v,v)\Delta t^2$).
Therefore, with $\|v\| \le V_{\max}$ and $\kappa$ the curvature bound:
\begin{equation}
  \|z_{t+1} - \Pi_{\mathcal{M}^*}(z_{t+1})\|
  \;\le\; \tfrac{1}{2}\kappa_{\max}\,V_{\max}^2\,\Delta t^2,
  \label{eq:one-step-off-manifold}
\end{equation}
where $\kappa_{\max}$ is the curvature bound from Assumption~\ref{ass:curvature}.
After the next cycle's Step~1 reset, the off-manifold
distance returns to $\le \delta_{\mathrm{loc}}$,
independent of $t$.
Therefore, for all $t \ge 0$:
\begin{equation}
  \|z_t - z_t^*\|
  \;\le\; \delta_{\mathrm{loc}} + \tfrac{1}{2}\kappa_{\max}V_{\max}^2\,\Delta t^2,
  \label{eq:uniform-latent-bound}
\end{equation}
where $z_t^* = \Pi_{\mathcal{M}^*}(z_t^{\mathrm{loc}})$
is the nearest Pareto-optimal point.
Applying the $K_{D_u}$-Lipschitz decoder and the
triangle inequality (as in Stage 4 of the proof of
Theorem~\ref{thm:error}):
\begin{align}
  \|\mathbf{u}_t - \mathbf{u}_t^*\|
  &\;\le\; K_{D_u}\,\|z_t - z_t^*\| + \delta_{\mathrm{dec}}
  \notag\\
  &\;\le\; K_{D_u}\!\left(\delta_{\mathrm{loc}}
    + \tfrac{1}{2}\kappa_{\max}V_{\max}^2\,\Delta t^2\right)
  + \delta_{\mathrm{dec}}.
\end{align}
Since the right-hand side is independent of $t$, the supremum
over all $t \ge 0$ is finite and equal to the expression
in~\eqref{eq:closed-loop-bound}.
\end{proof}

\begin{remark}[Comparison with Proposition~\ref{prop:drift} and
Theorem~\ref{thm:error}]
Proposition~\ref{prop:drift} establishes a $\limsup$ bound
that requires infinite time to approach and relies on a
Banach fixed-point argument for the free-running integrator.
Proposition~\ref{prop:closed-loop} replaces this with a
$\sup$ bound that holds from $t = 0$, using only a
one-step induction made possible by the Step~1 reset.
The two results are complementary:
Proposition~\ref{prop:drift} characterises the integrator
in isolation; Proposition~\ref{prop:closed-loop} characterises
the full closed-loop algorithm and is the stronger,
operationally relevant statement.
Note that $\delta_{\mathrm{loc}}$ is an additional tunable
parameter (controlled by $\tau_{\mathrm{geom}}$) that is
absent from Theorem~\ref{thm:error}; reducing
$\tau_{\mathrm{geom}}$ tightens the localization error at
the cost of a smaller feasible localization set
$\mathcal{M}_t$.
\end{remark}

\section{Discussion and Interpretation}
\label{app:discussion}

\paragraph{Why the reference-safe rate matters.}
The navigation result should not be read as a learned controller beating
an exact oracle.  The reference is a white-box optimizer with the same
known dynamics and constraints, but it is run under a finite online
budget in a coupled multi-agent scene.  When it returns a colliding or
unfinished rollout, its objective value is not an admissible optimum.
This is why Table~\ref{tab:iclr-main-results} reports both
reference-safe rate and GPC gap.  In moving scenes the reference is safe
in only $35\%$ of episodes, while GPC remains scene-collision-free in
$100\%$ and reaches $99\%$ reachable-conditioned safe success.  The meaningful optimality
comparison is therefore the conditional percentage gap on the
reference-safe subset, where GPC has a $-3.43\%$ aggregate cost gap.

\paragraph{What is learned.}
GPC does not learn the physics equations from scratch.  The physics and
constraints are used to produce feasible Pareto samples, feasibility
margins, and semantic violation labels.  The learned object is the
amortized geometry of the supported Pareto solution set: how feasible
actions, decoded physical states, objective values, and priority shifts
co-vary across the operating envelope.  This differs from a scalar
policy because the latent coordinate remains navigable after the
priority changes; it also differs from a lookup table because locality,
decoder regularity, and rollout-aware refinement define behavior between
offline samples. \rev{Equivalently, GPC represents the offline
finite-horizon solution operator, not future trajectories or a recursive
value function.}

\paragraph{Failure modes and monitoring.}
The theorem is local, and the implementation treats it that way.  GPC is
expected to be reliable when localization residuals, decoder calibration
error, and decoded margins remain within their validated ranges.  It can
fail if the observation leaves the sampled operating envelope, if the
offline optimizer has no positive feasibility margin, or if uncertainty
changes the feasible set more than represented in training.  These
conditions are monitorable: the controller logs reconstruction residual,
minimum decoded safety margin, and semantic priority at every step.  A
negative margin or large residual is a trigger for a conservative
fallback solver rather than a reason to silently trust extrapolation.
\rev{In our logs, these quantities correspond to localization
residuals, decoded feasibility margins, and semantic priorities; the
margins are obstacle and pairwise clearances in navigation and thermal,
voltage, generator, and ramp slacks in OPF.}

\paragraph{Role of the main components.}
The simpler studies report lower-dimensional sanity checks, but the same
component logic is visible in the main experiments.  Removing semantic
adaptation collapses the controller to a fixed scalarization
$\boldsymbol{\sigma}_t\equiv\bar{\boldsymbol{\sigma}}$, which is inadequate
when an obstacle approaches or a branch loading moves toward its thermal
limit.  Removing rollout-aware decoder refinement sets
$\lambda_{\mathrm{roll}}=0$ in eq.~\eqref{eq:iclr-training-loss}; this
keeps pointwise supervised error small but leaves the action decoder less
calibrated on latent states produced by RK2/Lie motion.  Removing the
Lie-local residual sets $\gamma_{\mathrm{L}}=0$ in
eq.~\eqref{eq:iclr-lie-residual}; this recovers a plain RK2 latent
navigator and is a useful baseline, but it ignores the decoder-pullback
metric directions in which small latent movements most affect physical
state.  These variants are not separate algorithms; they are controlled
ways to remove the semantic, rollout, or geometric parts of GPC.

\end{document}